\documentclass[a4paper, 11pt]{article}
\pdfoutput=1

\usepackage{jheppub}

\usepackage[dvipsnames]{xcolor}
\usepackage{graphicx}
\graphicspath{ {./} } 
\usepackage{mathrsfs}
\usepackage{amsmath}
\usepackage{enumitem}
\usepackage{mathtools}
\usepackage{amssymb}
\usepackage{stmaryrd}
\usepackage{amsthm,pifont}
\usepackage{amsfonts}
\usepackage{tikz-cd}
\usepackage{tikz}
\usepackage{mathalpha}
\usepackage{mathrsfs}  
\usetikzlibrary{fit,shapes.geometric}
\usepackage{environ}
\usepackage{array}
\newcolumntype{?}{!{\vrule width 1pt}}
\usepackage{makecell}
\usepackage{ctable}
\usepackage[margin=1cm,font=small]{caption}
\graphicspath{ {./pictures/} }

\usepackage{hyperref}
\hypersetup{
  colorlinks,
  linkcolor={RedViolet},
  citecolor={Emerald},
  urlcolor={RoyalPurple}
}

\usepackage{bm}
\usepackage{cleveref}

\input epsf.sty

\newcommand{\re}{{\rm e}}
\newcommand{\ri}{{\rm i}}

\newcommand{\e}{\mathbf{e}_1}
\newcommand{\mb}{{\mathsf{b}}}

\def\IZ{{\mathbb Z}}
\def\IR{{\mathbb R}}
\def\IC{{\mathbb C}}
\def\IH{{\mathbb H}}
\def\IQ{{\mathbb Q}}
\def\IN{{\mathbb N}}
\def\IP{{\mathbb P}}

\def\IF{{\mathbb F}}

\def\frakg{\mathfrak{g}}
\def\frakf{\mathfrak{f}}
\def\tfrakg{\tilde{\mathfrak{g}}}
\def\tfrakf{\tilde{\mathfrak{f}}}
\def\uk{\underline{k}}

\newcommand{\mO}{\mathsf{O}}

\newcommand{\mV}{\mathsf{V}}

\newcommand{\mx}{\mathsf{x}}
\newcommand{\my}{\mathsf{y}}

\newcommand{\CA}{{\cal A}}
\newcommand{\CB}{{\cal B}}
\newcommand{\CC}{{\cal C}}

\newcommand{\CF}{{\cal F}}
\newcommand{\CG}{{\cal G}}

\newcommand{\CS}{{\cal S}}

\newcommand{\CV}{{\cal V}}

\newcommand{\be}{\begin{equation}}
\newcommand{\ee}{\end{equation}}
\newcommand{\ba}{\begin{aligned}}
\newcommand{\ea}{\end{aligned}}

\newcommand{\borel}{\mathcal{B}}

\newtheorem{theorem}{Theorem}[section]
\newtheorem{prop}[theorem]{Proposition}
\newtheorem{lemma}[theorem]{Lemma}
\newtheorem{cor}[theorem]{Corollary}

\newtheorem{definition}{Definition}[section]
\newtheorem{conjecture}{Conjecture}

\newtheorem{rmk}{Remark}[section]

\title{\huge{\textbf Modular resurgence, $q$-Pochhammer symbols, and quantum operators from mirror curves}}

\author{Veronica Fantini$^a$ and Claudia Rella$^b$}
\affiliation{${}^a$LMO, Université Paris-Saclay, 91405 Orsay, France \\ ${}^b$IH\'ES, 91440 Bures-sur-Yvette, France}

\emailAdd{veronica.fantini@universite-paris-saclay.fr}
\emailAdd{rella@ihes.fr}

\abstract{
Building on the results of~\cite{FR1maths, FR1phys}, we study the resurgence of $q$-Pochhammer symbols and determine their summability and quantum modularity properties. We construct a new, infinite family of pairs of modular resurgent series from the asymptotic expansions of sums of $q$-Pochhammer symbols weighted by suitable Dirichlet characters. These weighted sums fit into the modular resurgence paradigm and provide further evidence supporting our conjectures in~\cite{FR1maths}. 
In the context of the topological string/spectral theory correspondence for toric Calabi--Yau threefolds, Kashaev and Mari\~no proved that the spectral traces of canonical quantum operators associated with local weighted projective planes can be expressed as sums of $q$-Pochhammer symbols. Exploiting this relation, we show that an exact strong-weak resurgent symmetry, first observed by the second author in~\cite{Rella22} and fully formalized in~\cite{FR1phys} for local $\IP^2$, applies to all local $\mathbb{P}^{m,n}$, albeit stripped of some of the underlying number-theoretic properties. Under some assumptions, these properties are restored when considering linear combinations of the spectral traces that reproduce the weighted sums above.
}

\makeatletter
\gdef\@fpheader{\null}
\makeatother

\begin{document}

\maketitle
\flushbottom

\section{Introduction}\label{sec:intro}

Originally motivated by complex Chern--Simons theory and quantum topology, the interplay between quantum modularity and resurgence has recently attracted increasing attention~\cite{aritm-resurgence-G, CCFGH, GGuM, garoufalidis_zagier_2023knots, sauzin, Garoufalidis_Zagier_2023, crew-goswami-osburn, CCKPG-revised, MS-R,Wheeler-3manifold,GGuMW}. Quantum modular forms, introduced by Zagier in~\cite{zagier_modular}, describe the controlled failure of modularity of certain $q$-series associated with knots and $3$-manifold invariants. In particular, because the asymptotic expansions of such $q$-series are often divergent, resurgence can be employed to characterize their quantum modular properties through the analysis of their divergent asymptotics. 
More generally, the theory of resurgence, developed by \'Ecalle in~\cite{EcalleI}, associates to a divergent formal power series a collection of exponentially small corrections along with a set of complex numbers, known as Stokes constants, which encode information about the large-order growth of the series and its non-perturbative data.

A first proposal for characterizing a class of quantum modular forms through the resurgent structure of their asymptotic expansions was presented in our earlier work~\cite{FR1maths}. There, we introduced the notion of a modular resurgent series: a Gevrey-1, simple resurgent asymptotic series whose Borel transform displays a single infinite tower of singularities, the secondary resurgent series are trivial, and the Stokes constants are the coefficients of an $L$-function. When obtained as the asymptotic expansion of a $q$-series, such modular resurgent series naturally occur in pairs satisfying the so-called modular resurgence paradigm and exhibit distinctive summability and quantum modularity properties. Examples of modular resurgent series arise in various contexts, though the definition was initially prompted by the resurgent analysis of the spectral trace of the toric Calabi--Yau (CY) threefold known as local $\IP^2$ studied in~\cite{Rella22, FR1phys}. 
In this paper, we continue this line of investigation by studying the asymptotic and number-theoretic properties of a new class of $q$-series constructed from the $q$-Pochhammer symbols. The latter appear as building blocks of the spectral traces of an infinite family of quantum operators canonically associated with the local weighted projective planes $\IP^{m,n}$, where $m,n \in \IZ_{>0}$.

\medskip 

Recall that the $q$-Pochhammer symbols, also known as $q$-shifted factorials, are functions of two variables defined by the finite products
\be \label{eq: qFactor}
(x; \, q)_m = \prod_{n=0}^{m-1} (1- x q^n) \, , \quad (x; \, q)_{-m} = \frac{1}{(x q^{-m}; \, q)_m} \, , \quad m \in \IZ_{>0} \, ,
\ee
with $(x; \, q)_0=1$. Extending the definition to infinite products, the (infinite) $q$-Pochhammer symbol is denoted by
\be \label{eq: dilog}
(x q^{\alpha}; \, q)_{\infty} = \prod_{n=0}^{\infty} (1- x q^{\alpha+n}) \, , \quad \alpha \in \IR \, ,
\ee
which is analytic in $x,q \in \IC$ with $|q| <1$.
For fixed $N\in\IZ_{\geq 2}$ and $k\in\IZ_N=\{1,\cdots, N \}$, we define the functions $f_{k,N}, g_{k,N}: \IH \to \IC$ as
\begin{equation}\label{eq:f_kN}
    f_{k,N}(y):=\log(\zeta_N^k;q)_\infty \, , \quad
    g_{k,N}(y):=\log(q^k;q^N)_\infty \, , 
\end{equation}
where $\zeta_N=\re^{2\pi\ri/N}$ and $q=\re^{2\pi\ri y}$. By slight abuse of nomenclature, we often refer to them as $q$-Pochhammer symbols in this text.
In the first part of this paper, we compute their asymptotic expansions for $y\to 0$ with $\Im(y)>0$, which we denote by $\tilde{f}_{k,N}(y)$ and $\tilde{g}_{k,N}(y)$, respectively. These are governed by simple resurgent Gevrey-1 asymptotic series whose full resurgent structure can be exactly solved. Indeed, their Borel transforms possess a single tower of simple poles repeated at all non-zero integer multiples of some fundamental constants along the imaginary axis in the Borel plane---that is, they reproduce the simplest instance of a peacock pattern~\cite{GuM,Rella22}.
As for the case study of~\cite{Rella22,FR1phys}, the corresponding sequences of Stokes constants $\{R_n^k\}$ and $\{S_n^k\}$, $n \in \IZ_{\ne 0}$, are expressed in closed form as arithmetic functions and generated by the discontinuities of $\tilde{f}_{k,N}(y)$ and $\tilde{g}_{k,N}(y)$, respectively, in appropriate variables. The latter can be simply written in terms of the same $q$-Pochhammer symbols in Eq.~\eqref{eq:f_kN}. In particular, the $q$-series generating the Stokes constants for $\tilde{f}_{k,N}$ (respectively, $\tilde{g}_{k,N}$) is explicitly determined by $g_{k,N}$ and $g_{N-k,N}$ (respectively, $f_{k,N}$ and $f_{N-k,N}$). Despite being reminiscent of it, this exchange of perturbative and non-perturbative information between the two $q$-Pochhammer symbols $f_{k,N}$ and $g_{k,N}$ does not fully replicate the modular resurgence paradigm of~\cite{FR1maths}. In fact, and importantly, their asymptotic expansions $\tilde{f}_{k,N}$ and $\tilde{g}_{k,N}$ do not possess one of the defining properties of modular resurgent series. Namely, the Dirichlet series built from the sequences of Stokes constants $\{ R_n^k \}$ and $\{ S_n^k \}$ are not necessarily $L$-functions. 
Nonetheless, the $q$-Pochhammer symbols $f_{k,N}$ and $g_{k,N}$ are holomorphic quantum modular forms in the sense of Zagier~\cite{zagier-talk}. 
\begin{theorem}[Thm.~\ref{thm:qm}]\label{thm:main-qm}
For every $k\in\IZ_N$, the functions $f_{k,N},g_{k,N}\colon\IH\to\IC$ in Eq.~\eqref{eq:f_kN} are holomorphic quantum modular functions for the group $\Gamma_N\subset\mathsf{SL}_2(\IZ)$ generated by the elements\footnote{Note that $\Gamma_N \subseteq \Gamma_1(N)$, where the equality is realized for $N=2,3,4$.}
\begin{equation} \label{eq: GammaN-generators}
    T=\begin{pmatrix}
        1 & 1\\
        0 & 1
    \end{pmatrix}\,, \quad \gamma_N=\begin{pmatrix}
        1 & 0\\
        N & 1
    \end{pmatrix}\, . 
\end{equation}
\end{theorem}

\medskip 

A pair of modular resurgent series can be constructed by considering suitable linear combinations of the $q$-Pochhammer symbols $f_{k,N}$ and $g_{k,N}$ in Eq.~\eqref{eq:f_kN} over the parameter $k\in\IZ_N$. Specifically, let $\chi_N: \IZ \to \IC$ be a Dirichlet character of modulus $N$. We introduce the functions $\frakf \,, \frakg \colon\IH\to\IC$ as
\be \label{eq:fg-intro}
\frakf(y):=\sum_{k\in\IZ_N}\chi_N(k)f_{k,N}(y)\,, \quad
\frakg(y):=\sum_{k\in\IZ_N}\chi_N(k)g_{k,N}(y) \, .
\ee
If $\chi_N=\chi_3$ denotes the unique primitive Dirichlet character of modulus $3$, then $f_0 \propto \frakf$ and $f_\infty \propto  \frakg$ reproduce the generating functions of the Stokes constants of the logarithm of the spectral trace of local $\IP^2$ in the weak and strong coupling limits, respectively, which have been extensively studied in~\cite{Rella22, FR1phys} and acted as a catalyst for the proposal of~\cite{FR1maths}. In the second part of this paper, 
we discuss the modular resurgent properties of the functions in Eq.~\eqref{eq:fg-intro}. We report our main results below.

\begin{theorem}[Thm.~\ref{thm:resurgence-f0}]\label{thm:main1-0}
Let $\chi_N$ be a primitive Dirichlet character of modulus $N$. Let $\tfrakf(y)$ be the asymptotic expansion of the function $\frakf(y)$ in Eq.~\eqref{eq:fg-intro} in the limit $y \to 0$ with $\Im(y)>0$. Then, $\tfrakf$ is modular resurgent if and only if $\chi_N$ is odd.
\end{theorem}

\begin{theorem}[Thm.~\ref{thm:resurgence-f-inf}]\label{thm:main1-inf}
Let $\chi_N$ be a primitive Dirichlet character of modulus $N$. Let $\tfrakg(y)$ be the asymptotic expansion of the function $\frakg(y)$ in Eq.~\eqref{eq:fg-intro} in the limit $y \to 0$ with $\Im(y)>0$. Then, $\tfrakg$ is modular resurgent if and only if $\chi_N$ is odd.
\end{theorem}

As conjectured in~\cite{FR1maths}, modular resurgent series possess special summability properties and their median resummation gives rise to holomorphic quantum modular forms. While the quantum modularity of $\frakf$ and $\frakg$ is a straightforward consequence of Theorem~\ref{thm:main-qm}, the effectiveness of the median resummation does, instead, only hold for the weighted sums and not for the single $q$-Pochhammer symbols $f_{k,N}$ and $g_{k,N}$. Accordingly, we obtain the following results.\footnote{The analogue of the analytic argument used in the proof of Theorem~\ref{thm:main1} is currently lacking for Conjecture~\ref{conj:main1b}, whose proof is left for future work.}

\begin{theorem}[Thm.~\ref{thm:summability-f0}]\label{thm:main1}
The function $\frakf\colon\IH\to\IC$ defined in Eq.~\eqref{eq:fg-intro} agrees with the median resummation of its asymptotic expansion $\tilde{\frakf}(y)$ if and only if the Dirichlet character $\chi_N$ is odd.
\end{theorem}

\begin{conjecture}[Thm.~\ref{thm:summability-f-inf}]\label{conj:main1b}
The function $\frakg\colon\IH\to\IC$ defined in Eq.~\eqref{eq:fg-intro} agrees with the median resummation of its asymptotic expansion $\tfrakg(y)$ if and only if the Dirichlet character $\chi_N$ is odd.
\end{conjecture}

It is argued on general grounds in~\cite{FR1maths} that pairs of modular resurgent series embed in the so-called modular resurgence paradigm, often represented as an exact commutative diagram, and capturing the mathematical structure underlying the strong-weak resurgent symmetry of~\cite{Rella22, FR1phys}.  
Here we prove that the modular resurgence paradigm of~\cite{FR1maths} holds for the $q$-series $\frakf$ and $\frakg$ in Eq.~\eqref{eq:fg-intro} under suitable constraints on the Dirichlet character $\chi_N$.

\begin{theorem}[Cor.~\ref{thm:mod-res}]\label{thm:main3}
Let $\chi_N$ be a primitive Dirichlet character of modulus $N$. The functions $\frakf$ and $\frakg$ defined in Eq.~\eqref{eq:fg-intro} satisfy the modular resurgence paradigm if and only if $\chi_N$ is odd.
\end{theorem}

\medskip 

Our results on the (modular) resurgent structure of the $q$-Pochhammer symbols in Eq.~\eqref{eq:f_kN} and their weighted sums in Eq.~\eqref{eq:fg-intro} find a straightforward application in the context of the topological string/spectral theory (TS/ST) correspondence of~\cite{GHM, CGM2}. Local mirror symmetry~\cite{KKV, CKYZ} pairs a toric CY threefold $X$ with an algebraic equation of the form
\be \label{eq: mirrorC}
W(\re^x, \re^y)=0 \, , \quad x,y \in \IC \, ,
\ee
describing a Riemann surface $\Sigma \subset \IC^*\times \IC^*$ of genus $g_\Sigma$, called its mirror curve. Following the prescription of~\cite{GHM, CGM2}, the complex variables $x,y$ are promoted to self-adjoint Heisenberg operators $\mx, \my$ satisfying $[\mx, \my] = \ri \hbar$ with $\hbar \in \IR_{>0}$. Then, the Weyl quantization of the algebraic curve in Eq.~\eqref{eq: mirrorC} naturally leads to quantum operators $\mathsf{O}_j(\re^{\mx}, \re^{\my})$, where $j=1, \dots, g_\Sigma$, acting on $L^2(\IR)$. Their inverses
$\rho_j = \mO_j^{-1}$ are conjectured to be positive-definite and trace-class, provided the mass parameters appearing in the operators satisfy certain positivity conditions. This property was proven in~\cite{KM} for a large number of toric del Pezzo CY threefolds.

In the third part of this paper, we consider the total space of the canonical bundle over the weighted projective plane $\IP(1,m,n)$ for $m,n \in \IZ_{>0}$, also known as local $\IP^{m,n}$, which is generally a singular manifold.
Note that some of these geometries appear as degenerations of toric del Pezzo CY threefolds in special limits. 
For example, when $n=1$, local $\IP^{m,1}$ can be constructed as a partial blow-down of the resolved $\IC^{3}/\IZ_{m+2}$ orbifold. 
In particular, local $\IP^{1,1}=\IP^2$ is the total resolution of the orbifold $\IC^3/\IZ_3$, while the mirror curve to local $\IP^{2,1}$ reproduces the so-called massless limit of the mirror curve to local $\IF_2$.
Following~\cite{KM}, we associate to local $\IP^{m,n}$ the quantum-mechanical operator
\be \label{eq:quantum-ops}
\mO_{m,n}= \re^\mx+\re^\my+\re^{-m\mx-n\my} \, ,
\ee
whose inverse
\be \label{eq: rho-mn}
\rho_{m,n}=\mO_{m,n}^{-1}
\ee
exists and is positive-definite and trace-class on $L^2(\IR)$. Therefore, the spectrum of $\rho_{m,n}$ is discrete and positive, and its spectral traces 
\be \label{eq: all-traces}
\mathrm{Tr}(\rho_{m,n}^\ell) < \infty \,  , \quad \ell \in \IZ_{>0} \, ,
\ee
can be expressed as multi-dimensional integrals whose kernels involve combinations of Faddeev's quantum dilogarithm.
For all $m,n \in \IZ_{>0}$, we study the asymptotic behavior, resurgence, summability, and quantum modularity of the first spectral trace $\mathrm{Tr}(\rho_{m,n})$ as a function of $\hbar$, which is known in closed form due to~\cite{KM}. More precisely, up to the addition of simple terms, $\log \mathrm{Tr}(\rho_{m,n})$ can be written as a finite sum of the $q$-Pochhammer symbols in Eq.~\eqref{eq:f_kN}, where we take $y \propto \hbar$ and $N=1+m+n$. We can thus apply our results on the $q$-Pochhammer symbols to determine the exact resurgent structures and corresponding arithmetic properties of $\log \mathrm{Tr}(\rho_{m,n})$ in both limits of $\hbar \to 0$ and $\hbar \to \infty$.

We denote by $\varphi_{m,n}(\hbar)$ and $\psi_{m,n}(\hbar^{-1})$ the formal power series of $\log \mathrm{Tr}(\rho_{m,n})$ at weak and strong coupling, respectively. These are simple resurgent Gevrey-1 asymptotic series whose Borel transform reproduces the same pattern observed for the individual $q$-Pochhammer symbol---namely, a tower of simple poles along the imaginary axis. The Stokes constants $\{S_\ell^{m,n}\}$ and $\{R_\ell^{m,n}\}$, $\ell \in \IZ_{\ne 0}$, are divisor sum functions and their generating functions $f_0^{m,n}, f_\infty^{m,n}: \IC\setminus\IR \to \IC$ are again expressible as finite sums of $q$-Pochhammer symbols.
As in the case of $\frakf$ and $\frakg$, the quantum modularity of $f_0^{m,n}$ and $f_\infty^{m,n}$ is a straightforward consequence of Theorem~\ref{thm:main-qm}. However, a second remarkable statement follows.  
Namely, each of the generating functions is identified with the discontinuity of the asymptotic expansion of the other in the appropriate limit and uniquely determines the asymptotic expansion in the dual regime.
\begin{theorem}[Thm.~\ref{thm:strong-weak}]\label{thm:Strong-weak-intro}
 For every $m,n \in \IZ_{>0}$, the generating functions $f_0^{m,n}$ and $f_\infty^{m,n}$ of the Stokes constants of $\log {\rm Tr}(\rho_{m,n})$ at weak and strong coupling are related by a weaker (yet, still exact) version of the strong-weak resurgent symmetry of~\cite{FR1phys}.  
\end{theorem}
Moreover, and differently from the case of the single $q$-Pochhammer symbol, the median resummation is effective in reconstructing the $q$-series $f_0^{m,n}$ and $f_\infty^{m,n}$ from their asymptotics.\footnote{Proofs of both Conjectures~\ref{conj:main1b} and~\ref{conj:median-Proj-intro-finf} would follow from a proof of the summability of the $q$-Pochhammer symbol $g_{k,N}$ as presented in Conjecture~\ref{conj:summability-g-kN}. }
\begin{theorem}[Thm.~\ref{thm:disc-med-0}]\label{thm:median-Proj-intro-f0}
   The generating function $f_0^{m,n}: \IH \to \IC$ (respectively, $f_0^{m,n}: \IH_- \to \IC$) agrees with the median resummation of its asymptotic expansion in the limit $y \to 0$ with $\Im(y)>0$ (respectively, $\Im(y)<0$).   
\end{theorem}
\begin{conjecture}[Thm.~\ref{thm:disc-med-inf}]\label{conj:median-Proj-intro-finf}
The generating function $f_\infty^{m,n}: \IH \to \IC$ (respectively, $f_\infty^{m,n}: \IH_- \to \IC$) agrees with the median resummation of its asymptotic expansion in the limit $y \to 0$ with $\Im(y)>0$ (respectively, $\Im(y)<0$).     
\end{conjecture}
We observe that, for $N>4$, the sequences of Stokes constants $\{S_\ell^{m,n}\}$ and $\{R_\ell^{m,n}\}$ are not necessarily multiplicative and thus do not define $L$-functions. As a consequence, the asymptotic series $\varphi_{m,n}(\hbar)$ and $\psi_{m,n}(\hbar^{-1})$ are not generally modular resurgent. Yet, the full-fledged, number-theoretic symmetry of~\cite{FR1maths, FR1phys} is realized when $N=3$ and $N=4$ (\emph{i.e.}, for local $\IP^2$ and $\IP^{1,2}$, $\IP^{2,1}$) as the generating functions $f_0^{m,n}, f_\infty^{m,n}$ reduce to the weighted sums $\frakf, \frakg$ in Eq.~\eqref{eq:fg-intro} for the unique primitive Dirichlet characters of modulus $3$ and $4$, respectively. 

More generally, we prove that, for any choice of $N \ge 3$, there exists a linear combination of the generating functions $f_0^{m,n}$ (resp. $f_\infty^{m,n}$) over $m,n \in \IZ_N$ with $1+m+n=N$ that coincides with a sum of $q$-Pochhammer symbols weighted by an odd Dirichlet character of modulus $N$ in the form of the function $\frakf$ (resp. $\frakg$).
When the character is also primitive, then the complete modular resurgence paradigm holds due to Theorem~\ref{thm:main3}.

\medskip 

This paper is organized as follows. 
In Section~\ref{sec:background}, we recall some basic aspects of the theory of resurgence, quantum modularity, and modular resurgent series. 
In Section~\ref{sec: single}, we discuss the resurgent structure, summability properties, and quantum modularity of the $q$-Pochhammer symbols $f_{k,N}$, $g_{k,N}$ in Eq.~\eqref{eq:f_kN}. 
Then, in Section~\ref{sec:weighted-sum}, we consider the weighted sums of $q$-Pochhammer symbols $\frakf$, $\frakg$ in Eq.~\eqref{eq:fg-intro}, which represent a novel class of examples of modular resurgent series, and prove our main theorems.
Finally, in Section~\ref{sec:weighted}, we bring our attention to the spectral theory of local weighted projective planes, where the $q$-Pochhammer symbols play a central role. Here, we apply the results of the previous sections to study the resurgence of the spectral traces of the inverses of the three-term operators $\mO_{m,n}$ in Eq.~\eqref{eq:quantum-ops} and discuss a generalization of the strong-weak resurgent symmetry of local $\IP^2$ to all local $\IP^{m,n}$.
There is one appendix containing the technical details of the main proofs.

\section{Background}\label{sec:background}

In this section, we review the basic tools of the theory of resurgence of \'Ecalle~\cite{EcalleI} and describe a particular class of quantum modular forms introduced by Zagier~\cite{zagier-talk}. These notions merge into the recent proposal of~\cite{FR1maths} for a novel perspective on the resurgence of those asymptotic series whose coefficients can be suitably expressed in terms of an $L$-function.
Such modular resurgent series are conjectured to possess specific summability and quantum modularity properties, while their resurgent structure is ultimately dictated by the analytic number-theoretic qualities of the corresponding $L$-function via the paradigm of modular resurgence.

\subsection{Resurgence of asymptotic series}

Let $y$ be a formal variable and $\phi(y)$ a Gevrey-1 asymptotic series. Explicitly, 
\be \label{eq: phi}
\phi(y) = \sum_{n=0}^{\infty} a_n y^{n+1} \in \IC[\![y]\!] \, , 
\ee
where the coefficients have the form
\be
|a_n|\le \CA^{-n} n! \, , \quad n \geq 0 \, ,
\ee
for some constant $\CA \in \IR_{>0}$. Its Borel transform is defined as
\be \label{eq: phihat}
\borel[\phi](\zeta) := \sum_{n=0}^{\infty} \frac{a_n}{n!} \zeta^{n} \in  \IC\{\zeta\} \, ,
\ee
where $\zeta$ is a new formal variable. Note that $\borel[\phi](\zeta)$ is a holomorphic function in an open neighborhood of $\zeta=0$ of radius $\CA$ and may be analytically continued to a Riemann surface, known as the Borel plane. When this is the case, the following definition applies.
\begin{definition}
A Gevrey-1 asymptotic series $\phi(y)$ is called \emph{resurgent} if its Borel transform $\borel[\phi](\zeta)$ can be endlessly analytically continued. Namely, for every $L>0$, there is a finite set of points $\Omega_L$ in the Borel plane such that $\borel[\phi](\zeta)$ can be analytically continued along any path of length at most $L$ that starts from a point 
$\eta_0$ with $|\eta_0|<\CA$ and avoids $\Omega_L$. 
The set of singularities is labeled by $\Omega:=\bigcup_{L>0}\Omega_L$. 
If, additionally, $\borel[\phi](\zeta)$ has only simple poles and logarithmic branch points, then $\phi(y)$ is called \emph{simple resurgent}.
\end{definition}

We denote by $\rho_{\omega} \in \IC$, $\omega \in \Omega$, the locations of the singularities of the Borel transform. 
A ray in the Borel plane that starts at the origin and passes through a singular point $\rho_\omega$ is called a Stokes ray.
We assume that the formal power series $\phi(y)$ in Eq.~\eqref{eq: phi} is simple resurgent---as it will be the case in the rest of this work. In particular, if the singularity at $\rho_{\omega}$ with $\omega\in\Omega$ is a simple pole, the local expansion of the Borel transform in Eq.~\eqref{eq: phihat} around it has the form 
\be \label{eq: Stokes0}
\borel[\phi](\zeta) = - \frac{S_{\omega}}{2 \pi \ri (\zeta - \rho_{\omega})} + \text{regular in $\zeta-\rho_\omega$} \, ,
\ee
where $S_{\omega} \in \IC$ is the Stokes constant at $\rho_{\omega}$.

Let us now fix an angle $0 \le \theta < 2 \pi$ and denote by $\CC_\theta= \re^{\ri \theta} \IR_{\ge 0}$ the corresponding ray in the complex $\zeta$-plane. 
The Laplace transform of $\borel[\phi](\zeta)$ along $\CC_\theta$ is defined as 
\be \label{eq: Laplace}
s_{\theta}(\phi)(y) = \int_{\CC_\theta} \re^{-\zeta/y} \borel[\phi](\zeta) \, d \zeta \, , 
\ee
which is analytic in the half-plane $\Re(e^{-\ri\theta}y)>0$ and whose asymptotic expansion at the origin reconstructs the factorially divergent formal power series $\phi(y)$ (see~\cite[Thm.~5.20]{diver-book}). We call $s_{\theta}(\phi)(y)$ the Borel--Laplace sum of $\phi(y)$ at angle $\theta$. Note that this is a well-defined function only if the analytically continued Borel transform grows at most exponentially in an open sector of the Borel plane containing $\CC_\theta$. Indeed, the Borel--Laplace sum in Eq.~\eqref{eq: Laplace} is discontinuous across the rays
\be
\arg(y)=\arg(\rho_{\omega}) \, , \quad \omega \in \Omega \, ,
\ee
corresponding to the Stokes rays in the Borel plane, and converges to a generally different holomorphic function in each sector bounded by these rays. Adjacent sectorial solutions differ by a correction that is invisible in the asymptotic limit $y \to 0$. 
Specifically, the discontinuity of $\phi(y)$ across a given ray $\CC_{\theta}$ in the complex $y$-plane is defined as
\be \label{eq: disc}
\mathrm{disc}_{\theta}\phi(y) = s_{\theta_+}(\phi)(y) - s_{\theta_-}(\phi)(y) = \int_{\mathcal{C}_{\theta_+} - \, \mathcal{C}_{\theta_-}} \re^{-\zeta/y} \borel[\phi](\zeta) \,  d \zeta \, , 
\ee
where $\theta_{\pm}= \theta \pm \epsilon$ for some $0 < \epsilon \ll 1$.
A standard contour deformation argument shows that such a difference collects an exponentially small contribution from each singularity of the Borel transform that lies on the ray $\CC_{\theta}$.
In particular, if $\borel[\phi](\zeta)$ has only simple poles, then we obtain
\be \label{eq: Stokes1-poles}
\mathrm{disc}_{\theta}\phi(y) = \sum_{\omega  \in \Omega_{\theta}} S_{\omega} \re^{-\rho_{\omega}/y}  \, ,
\ee
where the index $\omega \in \Omega_{\theta}$ labels the singular points $\rho_{\omega}$ such that $\arg (\rho_{\omega}) = \theta$ and the complex numbers $S_{\omega}$ are the same Stokes constants that appear in Eq.~\eqref{eq: Stokes0}. In the presence of more general types of singularities, the term of order $\re^{-\rho_\omega/y}$ is encoded in the specific form of the local behavior of the Borel transform near $\rho_\omega$ and computed by analytic continuation.

Finally, we recall that the median resummation of the Gevrey-1 asymptotic series $\phi(y)$ in Eq.~\eqref{eq: phi} at an arbitrary angle $\theta$ in the complex $y$-plane is the average of the two lateral Borel--Laplace sums $s_{\theta_\pm}(\phi)(y)$, that is,
\be \label{eq: median}
\mathcal{S}^{\mathrm{med}}_{\theta}\phi(y) = \frac{s_{\theta_+}(\phi)(y) + s_{\theta_-}(\phi)(y)}{2} \, ,
\ee
which is an analytic function for $\Re \left( \re^{-\ri\theta} y \right)>0$. Equivalently, by suitably varying the contour of integration of the Laplace transform, we can write
\begin{equation} \label{eq: median2}
\mathcal{S}^{\mathrm{med}}_{\theta}\phi(y) =\begin{cases}
        s_{\theta_-}(\phi)(y)+\frac{1}{2}\,\mathrm{disc}_{\theta}\phi(y) \, , \quad & \Re \left( \re^{-\ri\theta_{-}}y \right)>0 \,  \\
        & \\
        s_{\theta_+}(\phi)(y)-\frac{1}{2}\,\mathrm{disc}_{\theta}\phi(y) \, , \quad & \Re \left( \re^{-\ri\theta_{+}}y \right)>0 \, 
    \end{cases} \, .
\end{equation}

\subsection{Modular resurgent series}

Following the original proposal by Zagier~\cite{zagier_modular}, a function $f\colon\IQ\to\IC$ is called a weight-$\omega$ quantum modular form with respect to a subgroup $\Gamma\subseteq\mathsf{SL}_2(\IZ)$, where $\omega$ is a fixed integer or half-integer, if the cocycle $h\colon\Gamma\times\IQ\to\IC$ defined by
\begin{equation}\label{cocycle}
     h_\gamma[f](y):=(cy+d)^{-\omega} f\left(\frac{ay+b}{cy+d}\right) - f(y) 
\end{equation} 
has better analyticity properties than the function $f$ itself for every $\gamma=\left( \begin{smallmatrix}
    a & b\\
    c & d
\end{smallmatrix} \right)\in\Gamma$---\emph{e.g.}, it is real analytic over $\IR\setminus\{\gamma^{-1}(\infty)\}$.
Note that the cocycle $h$ of a quantum modular form $f$ measures the failure of invariance under the action of the modular group $\Gamma$; in fact, $h_\gamma[f]=0$ for every $\gamma\in\Gamma$ if and only if $f$ is a modular form.
 
Similarly, a notion of quantum modularity for functions that are holomorphic in the upper half of the complex plane, which we denote by $\IH$, is obtained by requiring that the cocycle in Eq.~\eqref{cocycle} is analytic in a domain larger than $\IH$. The quantum modular forms that belong to this class are called holomorphic quantum modular forms, and first introduced by Zagier~\cite[min 21:45]{zagier-talk}.
\begin{definition} \label{def: holoQM}
    A holomorphic function $f\colon\IH\to\IC$ is a weight-$\omega$ \emph{holomorphic quantum modular form} for a subgroup $\Gamma\subseteq\mathsf{SL}_2(\IZ)$, where $\omega$ is integer or half-integer, if the function $h_\gamma[f]\colon\IH\to\IC$ in Eq.~\eqref{cocycle} extends analytically to\footnote{Note that $\IC_\gamma = \IC \setminus \left(-\infty; \, -d/c\right]$ when $c>0$ and $\IC_\gamma = \IC \setminus \left[-d/c; \, +\infty \right)$ when $c<0$.}  
    \be
    \IC_\gamma:=\{y\in\IC\colon cy+d\in\IC\setminus\IR_{\leq 0}\}
    \ee
    for every $\gamma=\left( \begin{smallmatrix}
    a & b\\
    c & d
\end{smallmatrix} \right) \in\Gamma$.
\end{definition}
Note that an analog of Definition~\ref{def: holoQM} can be written for a holomorphic function $f\colon\IH_{-}\to\IC$, where $\IH_{-}$ denotes the lower half of the complex plane. (Holomorphic) quantum modular forms of weight zero are also known as (holomorphic) quantum modular functions. 

Quantum modular forms constitute a large class of functions, and characterizing them is a rather challenging endeavor. 
Yet, when the failure of modularity is measured by a function with factorially divergent asymptotics, resurgence proves to be remarkably useful. In~\cite{FR1maths}, we introduced the theory of modular resurgence, which identifies certain $q$-series as holomorphic quantum modular forms via the common resurgent structure of their divergent asymptotic expansions. 
As we will now briefly revisit, our approach fundamentally relies on the interplay between $q$-series and $L$-functions and their analytic and number-theoretic properties.

We start by recalling the definition and rich arithmetic structure of modular resurgent series as proposed in~\cite{FR1maths}.
\begin{definition}[Def.~3.2 in~\cite{FR1maths}]\label{def:modular_res_struct}
A Gevrey-1 asymptotic series $\tilde{\frakg}(y)$ is a \emph{modular resurgent series (MRS)} if the following conditions hold.
\begin{enumerate}
    \item The Borel transform $\CB[\tilde{\frakg}](\zeta)$ has a tower of singularities at the locations $\rho_m= m \CA $, $m\in\mathbb{Z}_{\ne 0}$, for some constant $\CA \in \IC$, in the complex $\zeta$-plane.
    \item For every $m\in\IZ_{\ne 0}$, the resurgent series at $\zeta=\rho_m$ is the constant function $S_m\in\IC$, \emph{i.e.}, the Stokes constant.
    \item The Stokes constants are the coefficients of two $L$-functions\footnote{An $L$-series is a Dirichlet series $L(s)=\sum_{m >0} \frac{A_{m}}{m^s}$ that converges in some right half-plane $\{\Re (s)>C\ge 0\}\subset\IC$ where it admits an Euler product expansion $L(s)=\prod_{p \in \mathsf{P}} \sum_{k >0} \frac{A_{p^k}}{p^{k s}}$ indexed by the set of prime numbers $\mathsf{P}$. When it exists, its meromorphic continuation to $\{\Re (s)< 0\}\subset\IC$ defines an $L$-function.}
    \be
		L_{+}(s)=\sum_{m >0} \frac{S_{m}}{m^s} \,, \quad L_{-}(s)=-\sum_{m >0} \frac{S_{- m}}{m^s} \, .
     \ee
\end{enumerate}
\end{definition} 
The MRS above is equivalently characterized by the generating series of the Stokes constants, that is, the $q$-series
\begin{equation} \label{eq:f_Am}
\frakf(y)=\begin{cases}
\displaystyle\sum_{m>0}S_m q^m & \mbox{if} \quad \Im(y)>0 \, ,  \\
\\
-\displaystyle\sum_{m>0}S_{-m} q^{-m} & \mbox{if} \quad \Im(y)<0 \, , 
\end{cases} \quad q= \re^{2\pi \ri y} \, , \quad y \in \IC\setminus\IR \, .
\end{equation} 
As emphasized in~\cite{FR1maths}, there is a canonical correspondence between the $L$-functions $L_\pm$ and the $q$-series $\frakf|_{\IH_\pm}$ based on the action of the Mellin transform and its inverse.

Notice that part~$3$ of Definition~\ref{def:modular_res_struct} does not specify the shape of the functional equation dictating the analytic continuation of the $L$-functions $L_{\pm}(s)$ to $\{\Re (s)< 0\}\subset\IC$, although its existence is required. Indeed, as first observed in the context of the spectral theory of the toric CY threefold known as local $\IP^2$ in~\cite{Rella22, FR1phys}, when an MRS comes from the asymptotic expansion of a $q$-series, then the functional equation for the corresponding $L$-functions must satisfy a certain constraint, which in turn guarantees the existence of a second, canonically paired MRS. The two MRSs constructed this way fit together into what we call \emph{modular resurgence paradigm}. Let us review the various ingredients and their connections step-by-step.

Assume that $\tilde{\frakg}(y)$ is an MRS whose Borel transform has only simple poles. Further assume that $\tilde{\frakg}(y)$ is the asymptotic expansion of a given $q$-series $\frakg(y)$ of the same form as in Eq.~\eqref{eq:f_Am} with coefficients given by $R_m \in \IC$, $m \in \IZ_{\ne 0}$. 
By definition, its resurgent structure leads to a pair of $L$-functions $L_\pm(s)$, whose inverse Mellin transform~\cite[Prop.~3.3]{FR1maths} is the $q$-series $\frakf(y)$ generating the Stokes constants $S_m \in \IC$, $m \in \IZ_{\ne 0}$, of $\tilde{\frakg}(y)$. 
Denote by $\tilde{\frakf}(y)$ its asymptotic expansion in the limit $y \to 0$. Then, crucially, the resurgent behavior of $\tilde{\frakf}(y)$ is entirely determined by the meromorphic continuation of $L_\pm(s)$ to $\Re(s)<0$, as we proved in~\cite[Props.~3.1 and~3.4]{FR1maths}. In particular, $\tilde{\frakf}(y)$ is also an MRS and its Stokes constants $R_m$, coefficients of a new pair of $L$-functions $L'_{\pm}(s)$, are generated by the $q$-series $\frakg(y)$ that we started with. The two pairs of $L$-functions $L_\pm$ and $L'_\pm$ provide the analytic continuation of each other, that is, they are connected by means of their functional equation.\footnote{Sometimes the two $L$-functions $L_\pm(s)$ and $L'_\pm(s)$ might be equal, as observed for Maass cusp forms in~\cite{FR1maths}.}
We illustrate the modular resurgence paradigm schematically by the following commutative diagram. 
We refer the reader to~\cite{FR1maths} for more details.
\begin{equation}\label{diag:resurgence-L funct}
\begin{tikzcd}[column sep=2.5em, row sep=2.8em]
       L'_{\pm}(s) \arrow[rrr,"\text{inverse Mellin}"]\arrow[ddrrrrrrrr,sloped,"\text{functional}"] & & & \mathfrak{g}(y) \arrow[rr,"y \rightarrow 0"] &  & \tilde{\mathfrak{g}}(y) \arrow[rrr,"\text{resurgence}"] & & & \{ S_m \}\arrow[dd,sloped,"\text{$L$-function}"] \\  \\
       \{ R_m \} \arrow[uu,sloped,"\text{$L$-function}"] & & &\tilde{\mathfrak{f}}(y) \arrow[lll,"\text{resurgence}"] 
       &  & \mathfrak{f}(y)  \arrow[ll,"y \rightarrow 0"] & & & L_{\pm}(s) \arrow[lll,"\text{inverse Mellin}"]\arrow[uullllllll,sloped,swap,"\text{equation}"]
\end{tikzcd}
\end{equation}

\begin{rmk}
    By construction, the MRSs $\tilde{\frakf}(y)$ and $\tilde{\frakg}(y)$ above are such that one equals the asymptotic expansion of the discontinuity of the other in the appropriate variables. In fact, the theory of modular resurgence was originally developed to make sense of the striking number-theoretic symmetry between the perturbative and non-perturbative contributions to the spectral trace of local $\IP^2$ at strong and weak coupling, which was discovered in~\cite{Rella22} and further formalized in~\cite{FR1phys}. Hence, the modular resurgence paradigm is referred to as \emph{strong-weak resurgent symmetry} in those contexts in which such an interpretation is present. It is then natural to rename the $q$-series as $f_0$ and $f_\infty$, their asymptotic expansions as $\tilde{f}_0$ and $\tilde{f}_\infty$, and their $L$-functions as $L_{0, \pm}$ and $L_{\infty, \pm}$, as we do in~\cite{FR1phys}.
\end{rmk}

Finally, we recall the conjectural summability and quantum modularity properties of MRSs. 
\begin{conjecture}[Conj.~1 in~\cite{FR1maths}]\label{conj:quantum_modular1-intro}
If $\frakg\colon\IH\to\IC$ is a $q$-series, where $q=\re^{2 \pi \ri y}$, whose asymptotic expansion $\tilde{\frakg}(y)$ as $y\to 0$ with $\Im(y)>0$ is modular resurgent, then the \emph{median resummation} of $\tilde{\frakg}$ reconstructs the original function $\frakg$. 
\end{conjecture}
\begin{conjecture}[Conj.~2 in~\cite{FR1maths}]\label{conj:quantum_modular2-intro}
If $\frakg\colon\IH\to\IC$ is a $q$-series, where $q=\re^{2 \pi \ri y}$, whose asymptotic expansion $\tilde{\frakg}(y)$ as $y\to 0$ with $\Im(y)>0$ is modular resurgent, then the function $\frakg$ is a \emph{holomorphic quantum modular form} for a subgroup $\Gamma\subseteq\mathsf{SL}_2(\mathbb{Z})$.
\end{conjecture}

In the following sections of this paper, we will construct a large class of examples of MRSs intimately tied to the $q$-Pochhammer symbols and prominently appearing in the spectral theory of local weighted projective planes. These examples will also provide evidence of our conjectures.

\section{The \texorpdfstring{$q$}{q}-Pochhammer symbols} \label{sec: single}

In this section, we consider the functions $f_{k,N}(y)$, $g_{k,N}(y)$, $y \in \IH$, introduced in Eq.~\eqref{eq:f_kN} as logarithms of appropriate $q$-Pochhammer symbols after fixing $N\in\IZ_{\geq 2}$ and $k\in\IZ_N$. We study the resurgent structure and summability properties of their asymptotic expansions in the limit $y\to0$ with $\Im(y)>0$. The Borel transform displays a single tower of repeated simple poles along the imaginary axis with Stokes constants given by an explicit arithmetic function. Meanwhile, each $q$-Pochhammer symbol is reproduced by the Borel--Laplace sums at angles $0$ and $\pi$ of its asymptotics up to a correction depending on the other $q$-Pochhammer symbol. Moreover, the median resummation at angles $\pm\pi/2$ is not effective in reconstructing either $f_{k,N}(y)$ or $g_{k,N}(y)$. Finally, we prove that both $q$-Pochhammer symbols are holomorphic quantum modular functions for the group~$\Gamma_N \subset \mathsf{SL}_2(\IZ)$ generated by the elements in Eq.~\eqref{eq: GammaN-generators}.

To simplify our notation, we will denote $\underline{k}=N-k$ for every $k\in\IZ_N$.

\subsection{Resurgence and summability}

We begin by computing the resurgent structures of the asymptotic series $\tilde{f}_{k,N}(y)$, $\tilde{g}_{k,N}(y)$ of the functions $f_{k,N}(y)$, $g_{k,N}(y)$, respectively, in the limit $y \rightarrow 0$ with $\Im(y)>0$.

Recall that for $y \rightarrow 0$ with $\Im(y) >0$ and $x$ fixed, the following $q$-Pochhammer symbols, where $q=\re^{2 \pi \ri y}$, have the known asymptotic expansions~\cite{Katsurada, Shubho}
\begin{subequations}
\begin{align}
\log (x ; \, q)_{\infty} \sim & \frac{1}{2} \log(1-x) + \sum_{n=0}^{\infty} (2 \pi \ri y)^{2n-1} \frac{B_{2n}}{(2n)!} \mathrm{Li}_{2-2n}(x) \, , \label{eq: logPhiNC} \\
\log (q^{\alpha} ; \, q)_{\infty} \sim & - \frac{\pi \ri}{12 y} - B_1(\alpha) \log(- 2 \pi \ri y) - \log \frac{\Gamma(\alpha)}{\sqrt{2 \pi}}  \label{eq: logPhiK} \\
& - B_2(\alpha) \frac{\pi \ri y}{2} - \sum_{n=2}^{\infty} (2 \pi \ri y)^n \frac{B_n B_{n+1}(\alpha)}{n (n+1)!} \, , \quad \alpha > 0 \, , \nonumber
\end{align}
\end{subequations}
where $B_n(\alpha)$ is the $n$-th Bernoulli polynomial, $B_n = B_n(0)$ is the $n$-th Bernoulli number, $\Gamma(\alpha)$ is the gamma function, and $\mathrm{Li}_n(x)$ is the polylogarithm of order $n$.
Applying the formula in Eq.~\eqref{eq: logPhiNC} with $x=\zeta_N^k$ to the expression for $f_{k,N}(y)$ in Eq.~\eqref{eq:f_kN} and explicitly evaluating the special functions that appear, we obtain that 
\be \label{eq: expansion-fkN}
\tilde{f}_{k,N}(y)=\frac{1}{2} \log(1-\zeta_N^k)+\frac{1}{2\pi \ri y} \mathrm{Li}_2(\zeta_N^k)+\psi_k(y)\, ,
\ee
where $\psi_k(y)$ is the formal power series
\be\label{eq: tilde-psi}
\psi_k(y):=\sum_{n=1}^\infty (2\pi \ri y)^{2n-1}\frac{B_{2n}}{(2n)!}\mathrm{Li}_{2-2n}(\zeta_N^k)\, . 
\ee
Analogously, if we use the formula in Eq.~\eqref{eq: logPhiK} with $\alpha=k/N$ to expand the function $g_{k,N}(y)$ in Eq.~\eqref{eq:f_kN}, we find that
\be \label{eq: expansion-gkN}
\tilde{g}_{k,N}(y/N)=- \frac{\pi \ri }{12 y} - B_1\left(\tfrac{k}{N}\right) \log(- 2 \pi \ri y) - \log \frac{\Gamma\big(\tfrac{k}{N}\big)}{\sqrt{2 \pi}} - B_2\big(\tfrac{k}{N}\big) \frac{\pi \ri y}{2} - \varphi_k(y)\, ,
\ee
where $\varphi_k(y)$ is the formal power series
\be\label{eq: tilde-phi}
\varphi_k(y):= \sum_{n=1}^{\infty} (2 \pi \ri y)^{2n} \frac{B_{2n} B_{2n+1}\big(\tfrac{k}{N}\big)}{2n (2n+1)!}\, . 
\ee

\subsubsection{The resurgent structure of \texorpdfstring{$\tilde{f}_{k,N}$}{fkN}}

\begin{prop}\label{prop:Borel-f_kN}
The formal power series $\psi_k(y)$ in Eq.~\eqref{eq: tilde-psi} governing the asymptotic expansion $\tilde{f}_{k,N}(y)$ in Eq.~\eqref{eq: expansion-fkN} is Gevrey-1 and simple resurgent.
\end{prop}

\begin{proof}
By the Hadamard product formula, the Borel transform of $\psi_k(y)$ can be resummed\footnote{The convergence of the infinite sum in the right-hand side (RHS) of Eq.~\eqref{eq: intBorel2infty-proof} can be easily verified by, \emph{e.g.}, the limit comparison test. Indeed, where defined, the generic term of the series is dominated by $1/m^2$.} into 
\be \label{eq: intBorel2infty-proof}
\begin{aligned}
    \borel [\psi_k](\zeta)&=\frac{1}{2 \pi \ri} \sum_{m=1}^{\infty}\frac{1}{m^2}\left(\frac{1}{1-\zeta_N^{-k}\re^{-\frac{\zeta}{m}}}+\frac{1}{1-\zeta_N^{-k}\re^{\frac{\zeta}{m}}}\right) \, ,
\end{aligned}
\ee
which is a well-defined function of $\zeta$ for $|\zeta| < 2 \pi k /N$ and extends meromorphically to the whole complex $\zeta$-plane. See Appendix~\ref{app: proofs} for the full derivation.
In fact, $\borel [\psi_k](\zeta)$ has simple poles at $\zeta = \eta_{m, \ell}$ where   
\be \label{eq: eta-ml}
\eta_{m, \ell} = 2 \pi \ri m \left(\pm\frac{k}{N} + \ell\right) \, , \quad m \in \IZ_{\geq 1} \, , \quad \ell\in\IZ \, .
\ee
\end{proof}

\begin{cor}\label{cor:stokes-f}
    The singularities of the Borel transform of the asymptotic series $\psi_k(y)$ are simple poles located along the imaginary axis at the points
    \be \label{eq: eta-n}
    \eta_n=2\pi\ri \frac{n}{N} \, , \quad n\in\IZ_{\neq 0} \, ,
    \ee 
    while the corresponding Stokes constants are given by the arithmetic function
   \be\label{eq:Stokes-R-kN}
   R^k_n=\sum_{\substack{d\vert n\\ d\underset{N}{\equiv} k}}\frac{d}{n}-\sum_{\substack{d\vert n\\ d\underset{N}{\equiv} -k}}\frac{d}{n}\, .
\ee
The sums run over the positive integer divisors of $n$. Besides, $R^k_n=-R^k_{-n} \in \IQ$, while $n \, R^k_n \in \IZ$.
\end{cor}

\begin{proof}
It follows from Proposition~\ref{prop:Borel-f_kN} that the singularities of $\borel[\psi_k](\zeta)$ are simple poles located at the points $\eta_{m,\ell}$ in Eq.~\eqref{eq: eta-ml}---equivalently, at the points $\eta_n$ in Eq.~\eqref{eq: eta-n}---and the corresponding local expansions are obtained via the formula in Eq.~\eqref{eq: Stokes0}. Thus, the Stokes constant $R_n^k$, $n \in \IZ_{\ne 0}$, associated with the singularity at $\eta_n$ is given by the sum of residues from each choice of $m \in \IZ_{\ge 1}$ and $\ell \in \IZ$ such that
\be
n = m (\pm k + \ell N) \, .
\ee
Therefore, we find that
\be
R_n^k=-2 \pi \ri \, \underset{\zeta=\eta_n}{\text{Res}}\,\borel[\psi_k](\zeta) =\sum_{\substack{m,\, \ell\\ \tfrac{n}{m}=k+N\ell}} \frac{1}{m}-\sum_{\substack{m,\, \ell\\ \tfrac{n}{m}=-k+N\ell}} \frac{1}{m} \, ,
\ee
which gives the desired formula in Eq.~\eqref{eq:Stokes-R-kN}.
\end{proof}

\begin{cor} \label{cor: disc-Rn}
The discontinuity of the asymptotic series $\psi_k(y)$ across the Stokes ray at angle $\pi/2$, that is, the generating series of the Stokes constants $R_n^k$, $n \in \IZ_{>0}$, is given by
\be \label{eq:disc-psi}
\mathrm{disc}_{\frac{\pi}{2}}\psi_k(y)= \sum_{n=1}^\infty R_n^k \re^{-\eta_n/y} = -g_{k,N}\big(-\tfrac{1}{Ny}\big)+g_{\underline{k}
,N}\big(-\tfrac{1}{Ny}\big) \, .
\ee
\end{cor}
\begin{proof}
Applying the definition of the $q$-Pochhammer symbol in Eq.~\eqref{eq:f_kN} and Taylor expanding the logarithm, we write the RHS of Eq.~\eqref{eq:disc-psi} as
\be
\ba
\log \frac{(x^{k/N} ;x)_\infty}{(x^{(N-k)/N} ;x)_\infty} 
&= - \sum_{m=0}^{\infty} \sum_{\ell =1}^{\infty} \frac{1}{\ell} \left( x^{(k+ m N)\ell/N } - x^{(\underline{k}+ m N) \ell/N } \right) \\
&= - \sum_{n =1}^{\infty} x^{n/N} \Bigg[ \sum_{\substack{d\vert n\\ d\underset{N}{\equiv} k}}\frac{d}{n}-\sum_{\substack{d\vert n\\ d\underset{N}{\equiv} -k}}\frac{d}{n} \Bigg] \, ,
\ea
\ee
where $x = \re^{-2 \pi \ri/y}$.
\end{proof}
As a consequence of Corollary~\ref{cor: disc-Rn}, for any choice of $k \in \IZ_{N}$, the discontinuity of the asymptotic series obtained from the $q$-Pochhammer symbol $f_{k,N}(y)$ is uniquely determined by the $q$-Pochhammer symbols $g_{k,N}(-\tfrac{1}{Ny})$ and $g_{\underline{k},N}(-\tfrac{1}{Ny})$. As we will prove in Corollary~\ref{cor: disc-Sn}, the converse is also true. This two-way exchange of perturbative/non-perturbative
content between the two $q$-Pochhammer symbols in Eq.~\eqref{eq:f_kN} is reminiscent of the exact strong-weak resurgent symmetry of~\cite{Rella22, FR1phys}. We will return to this point in Remark~\ref{rmk: not-L-functs} and later in Section~\ref{sec:strong-weak}.

\subsubsection{The resurgent structure of \texorpdfstring{$\tilde{g}_{k,N}$}{gkN}}

\begin{prop}\label{prop:Borel-g_kN}
The formal power series $\varphi_k(y)$ in Eq.~\eqref{eq: tilde-phi} governing the asymptotic expansion $\tilde{g}_{k,N}(y/N)$ in Eq.~\eqref{eq: expansion-gkN} is Gevrey-1 and simple resurgent.
\end{prop}

\begin{proof}
By the Hadamard product formula, the Borel transform of $\varphi_k(y)$ can be resummed\footnote{The convergence of the infinite sum in the RHS of Eq.~\eqref{eq: intBorel20-proof} can be easily verified by, \emph{e.g.}, the limit comparison test. Indeed, where defined, the generic term of the series is dominated by $1/\ell^2$.} into
\be \label{eq: intBorel20-proof}
    \borel [\varphi_k](\zeta)=-\sum_{\ell=1}^\infty \frac{\sin\big(\tfrac{2 \pi k \ell}{N}\big)}{\pi\ell}\left(-\frac{1}{\zeta}+\frac{1}{2\ell}\coth\left(\frac{\zeta}{2\ell}\right)\right) \, ,
\ee
which is a well-defined function of $\zeta$ for $|\zeta| < 2 \pi$ and extends meromorphically to the whole complex $\zeta$-plane. See Appendix~\ref{app: proofs} for the full derivation. 
In fact, $\borel [\varphi_k](\zeta)$ has simple poles at $\zeta = \rho_{m,\ell}$ where
\be \label{eq: zeta-ml}
\rho_{m,\ell} = 2\pi\ri m \ell \, , \quad m\in\IZ_{\neq 0} \, , \quad \ell\in \IZ_{\geq 1} \, .
\ee
\end{proof}

\begin{cor}\label{cor:stokes-g}
   The singularities of the Borel transform of the asymptotic series $\varphi_k(y)$ are simple poles located along the imaginary axis at the points
    \be \label{eq: zeta-n}
    \rho_n=2\pi\ri n \, , \quad n\in\IZ_{\neq 0} \, ,
    \ee 
    while the corresponding Stokes constants are given by the arithmetic function
   \be\label{eq:Stokes-S-kN}
   S_n^k= \ri \sum_{j \in \IZ_N} \sin\left( \frac{2 \pi k j}{N} \right) \Bigg[ \sum_{\substack{d\vert n\\ d\underset{N}{\equiv} j}} \frac{1}{d} -\sum_{\substack{d\vert n\\ d\underset{N}{\equiv} -j}} \frac{1}{d}\Bigg] \, .
\ee
The inner sums run over the positive integer divisors of $n$. Besides, $S_n^k=S^k_{-n} \in \ri \IR$.
\end{cor}

\begin{proof}
It follows from Proposition~\ref{prop:Borel-g_kN} that the singularities of $\borel[\varphi_k](\zeta)$ are simple poles located at the points $\rho_{m,\ell}$ in Eq.~\eqref{eq: zeta-ml}---equivalently, at the points $\rho_n$ in Eq.~\eqref{eq: zeta-n}---and the corresponding local expansions are obtained via the formula in Eq.~\eqref{eq: Stokes0}. Thus, the Stokes constant $S_n^k$, $n \in \IZ_{\ne 0}$, associated with the singularity at $\rho_n$ is given by the sum of residues from each choice of $m \in \IZ_{\ne 0}$ and $\ell \in \IZ_{\ge 1}$ such that
\be
n = m \ell \, .
\ee
Therefore, we find that
\be
S_n^k=-2 \pi \ri \, \underset{\zeta=\rho_n}{\text{Res}}\,\borel[\varphi_k](\zeta) = 2 \ri \sum_{\substack{m,\, \ell\\ \tfrac{n}{\ell}=m}}  \frac{1}{\ell} \sin\left(\frac{2 \pi k \ell }{N}\right) =  2 \ri  \sum_{d\vert n} \frac{1}{d} \sin\left( \frac{2 \pi k d}{N} \right)\, ,
\ee 
which gives the desired formula in Eq.~\eqref{eq:Stokes-S-kN}. 
\end{proof}

Let us stress that the formula above for the Stokes constants $S_n^k$, $n \in \IZ_{\ne 0}$, simplifies considerably in the cases of $N=3,4$. Specifically, 
\begin{subequations}
    \begin{align}
       N&=3 \quad : \quad - \frac{\ri}{\sqrt{3}} S_n^1 = \sum_{\substack{d\vert n\\ d\underset{3}{\equiv} 1}} \frac{1}{d} -\sum_{\substack{d\vert n\\ d\underset{3}{\equiv} -1}} \frac{1}{d} = \sum_{d|n} \frac{1}{d} \chi_{3}(d) \in \IQ_{>0} \, ,\label{eq:sn-3} \\
       N&=4 \quad : \quad -  \frac{\ri}{2} S_n^1 = \sum_{\substack{d\vert n\\ d\underset{4}{\equiv} 1}} \frac{1}{d} -\sum_{\substack{d\vert n\\ d\underset{4}{\equiv} -1}} \frac{1}{d} = \sum_{d|n} \frac{1}{d} \chi_{4}(d) \in \IQ_{>0} \label{eq:sn-4}\, ,
    \end{align}
\end{subequations}
where $\chi_{3}$, $\chi_{4}$ are the unique non-principal Dirichlet characters\footnote{The Dirichlet characters $\chi_{3}(\bullet)$ and $\chi_{4}(\bullet)$ are also known as the Kronecker symbols $\big(\frac{-3}{\bullet}\big)$ and $\big(\frac{-4}{\bullet}\big)$, respectively.} of modulus 3 and 4, respectively. In these cases, we also have that
\begin{subequations}
    \begin{align}
       N&=3 \quad : \quad R_n^1 = \sum_{d|n} \frac{d}{n} \chi_{3}(d) \in \IQ_{ \ne 0} \, ,\label{eq:rn-3} \\
       N&=4 \quad : \quad R_n^1 = \sum_{d|n} \frac{d}{n} \chi_{4}(d) \in \IQ_{\ne 0} \, . \label{eq:rn-4}
    \end{align}
\end{subequations}

\begin{cor} \label{cor: disc-Sn} 
The discontinuity of the asymptotic series $\varphi_k(y)$ across the Stokes ray at angle $\pi/2$, that is, the generating series of the Stokes constants $S_n^k$, $n \in \IZ_{>0}$, is given by 
\be \label{eq:disc-phi}
\mathrm{disc}_{\frac{\pi}{2}}\varphi_k(y)= \sum_{n=1}^\infty S_n^k \re^{-\rho_n/y} = - f_{k,N}\Big( -\tfrac{1}{y} \Big)+f_{\underline{k},N}\Big( -\tfrac{1}{y} \Big) + \frac{\pi \ri}{N} (2k-N) \, .
\ee
\end{cor}
\begin{proof}
Applying the definition of the $q$-Pochhammer symbol in Eq.~\eqref{eq:f_kN} and Taylor expanding the logarithm, we write the RHS of Eq.~\eqref{eq:disc-phi} as
\be
\ba
\log \frac{(x \zeta_N^k ;x)_\infty}{(x \zeta_N^{-k} ;x)_\infty} 
&= - \sum_{m, \ell =1}^{\infty} \frac{x^{m \ell}}{\ell} \left( \re^{2 \pi \ri k \ell/N} - \re^{-2 \pi \ri k \ell/N} \right) \\
&= - 2 \ri \sum_{n =1}^{\infty} x^n \Bigg[ \sum_{d|n} \frac{1}{d} \sin \left( \frac{2 \pi k d}{N} \right) \Bigg] \, ,
\ea
\ee
where $x=\re^{- 2 \pi \ri /y}$. Note that we have also used that
\be
\frac{\pi \ri}{N} (2k-N) = \log \frac{(1-\zeta_N^k)}{(1-\zeta_N^{-k})} \, .
\ee
\end{proof}

\begin{rmk} \label{rmk: not-L-functs}
For fixed $N \in \IZ_{\ge 2}$ and $k \in \IZ_N$, the Dirichlet series obtained from the sequences of Stokes constants $\{R_n^k\}$ and $\{S_n^k\}$, $n \in \IZ_{>0}$, that is, 
\be \label{eq: dir-LL'}
L'_{k,N}(s) = \sum_{n=1}^\infty \frac{R_n^k}{n^s} \, , \quad L_{k,N}(s) = \sum_{n=1}^\infty \frac{S_n^k}{n^s} \, , \quad s \in \IC \, ,
\ee
are not necessarily $L$-series. More precisely, for $N>4$, the arithmetic functions in Eqs.~\eqref{eq:Stokes-R-kN} and~\eqref{eq:Stokes-S-kN} are not multiplicative. Equivalently, they do not satisfy an expansion as an Euler product indexed by the set of prime numbers. It follows that the formal power series $\psi_k(y)$ and $\varphi_k(y)$ in Eqs.~\eqref{eq: tilde-psi} and~\eqref{eq: tilde-phi} are not modular resurgent, \emph{i.e.}, Definition~\ref{def:modular_res_struct} does not generally apply.
Nonetheless, a weaker formulation of the modular resurgence paradigm holds. 
If we introduce the functions
\begin{subequations} \label{eq: akbk}
\begin{align}
\mathcal{F}_{k,N}(y) &:= f_{k,N}(y) - f_{\underline{k},N}(y) - \frac{\pi \ri}{N}(2k-N) \, ,  \label{eq: ak} \\
\mathcal{G}_{k,N}(y) &:= -g_{k,N}(y) + g_{\underline{k},N}(y) \, , \label{eq: bk}
\end{align}
\end{subequations}
and denote by $\tilde{\mathcal{F}}_{k,N}(y)$ and $\tilde{\mathcal{G}}_{k,N}(y)$ their asymptotic expansions for $y \to 0$ with $\Im(y)>0$, then Corollaries~\ref{cor: disc-Rn} and~\ref{cor: disc-Sn} imply that
\be \label{eq: disc-akbk}
\mathrm{disc}_{\frac{\pi}{2}}\tilde{\mathcal{F}}_{k,N}(y) = 2\mathcal{G}_{k,N}\big(-\tfrac{1}{Ny}\big) \, , \quad
    \mathrm{disc}_{\frac{\pi}{2}}\tilde{\mathcal{G}}_{k,N}(y) = -2\mathcal{F}_{k,N}\big(-\tfrac{1}{Ny}\big) \, .
\ee
As we will see in Section~\ref{sec:strong-weak}, this simple exchange of perturbative/non-perturbative content between the $q$-Pochhammer symbols $\mathcal{F}_{k,N}$ and $\mathcal{G}_{k,N}$ finds an equivalent realization in the strong-weak symmetry relating the dual resurgent structures of the spectral trace of a local weighted projective plane. 
\end{rmk}

\subsubsection{Borel--Laplace sums}

We denote by $\e\colon\IC\setminus \ri\IR_{\geq 0}\to\IC$ the function defined as
\begin{equation} \label{eq: e1-def}
    \e(y):=\frac{1}{2\pi \ri}\int_0^{\infty} \re^{-2\pi t} \frac{dt}{t+\ri y} =\frac{1}{2\pi \ri} \re^{2\pi \ri y} \, \Gamma(0,2\pi \ri y) \, ,
\end{equation}
where $\Gamma(s, 2\pi \ri y)$ is the upper incomplete gamma function. The exponential integral $\e(y)$ is analytic for $y \in \IC\setminus \ri\IR_{\geq 0}$, and its jump across the branch cut along the positive imaginary axis is given by 
\be
\lim_{\delta \rightarrow 0^+} \e(\ri x - \delta) - \e(\ri x+ \delta) = \re^{-2 \pi x} \, ,
\ee
where we have fixed $x \in \IR_{\ge 0}$.
Let us now go back to the asymptotic expansions of the $q$-Pochhammer symbols $f_{k,N}$ and $g_{k,N}$ in Eq.~\eqref{eq:f_kN}. Their Borel--Laplace sum along the real axis follows from the particularly simple singularity structure of their Borel transforms, and we can explicitly compute it at angles $0$ and $\pi$ in terms of the function $\e$.

\begin{prop} \label{prop: BL-e1}
    The asymptotic expansions $\tilde{f}_{k,N}(y)$ and $\tilde{g}_{k,N}(y)$ in Eqs.~\eqref{eq: expansion-fkN} and~\eqref{eq: expansion-gkN} are Borel--Laplace summable along the positive and negative real axes. Their Borel--Laplace sums at angles $\alpha = 0, \pi$ are 
    \begin{subequations}
    \begin{align}
        s_\alpha(\tilde{f}_{k,N})(y) =&\frac{1}{2} \log(1-\zeta_N^k)+\frac{1}{2\pi \ri y} \mathrm{Li}_2(\zeta_N^k) - \sum_{n \in \IZ_{\ne 0}} R_n^k\, \e \big(-\tfrac{n}{N y} \big) \, , \\
        s_\alpha(\tilde{g}_{k,N})(y) =&- \frac{\pi \ri }{12 N y} - B_1\left(\tfrac{k}{N}\right) \log(- 2 \pi \ri N y) - \log \frac{\Gamma\big(\tfrac{k}{N}\big)}{\sqrt{2 \pi}} \nonumber \\
        &- B_2\big(\tfrac{k}{N}\big) \frac{\pi \ri N y}{2} + \sum_{n \in \IZ_{\ne 0}} S_n^k\, \e \big(-\tfrac{n}{N y} \big) \, ,
    \end{align}
    \end{subequations}
    which are analytic if $\Re(y)>0$ for $\alpha=0$ and if $\Re(y)<0$ for $\alpha=\pi$. 
\end{prop}
\begin{proof}
As a consequence of Corollaries~\ref{cor:stokes-f} and~\ref{cor:stokes-g}, the Borel transforms of the Gevrey-1 asymptotic series $\psi_k(y)$ and $\varphi_k(y)$ in Eqs.~\eqref{eq: tilde-psi} and~\eqref{eq: tilde-phi} can be written as
\begin{subequations}
    \begin{align}
\borel[\psi_k](\zeta) &= -\frac{1}{2 \pi \ri} \sum_{n \in \IZ_{\ne 0}} \frac{R_n^k}{\zeta-\eta_n}\, , \\
\borel[\varphi_k](\zeta) &= -\frac{1}{2 \pi \ri} \sum_{n \in \IZ_{\ne 0}} \frac{S_n^k}{\zeta-\rho_n}\, , 
    \end{align}
\end{subequations}
where $R_n^k$, $S_n^k$ are the Stokes constants in Eqs.~\eqref{eq:Stokes-R-kN} and~\eqref{eq:Stokes-S-kN}, while $\eta_n$, $\rho_n$ are the locations of the simple poles in Eqs.~\eqref{eq: eta-n} and~\eqref{eq: zeta-n}.
Assuming that $y\in\IR_{\ge 0}$, we compute the Laplace transforms along the positive real axis as
\begin{subequations}
\begin{align}
s_0(\psi_k)(y) &= \int_0^{\infty} d\zeta \re^{-\zeta/y} \borel[\psi_k](\zeta) = - \sum_{n \in \IZ_{\ne 0}} R_n^k \e \left(\frac{\ri \eta_n}{2 \pi y} \right) \, ,    \label{eq: LT-psi-e1} \\
s_0(\varphi_k)(y) &= \int_0^{\infty} d\zeta \re^{-\zeta/y} \borel[\varphi_k](\zeta) = - \sum_{n \in \IZ_{\ne 0}} S_n^k \e \left(\frac{\ri \rho_n}{2 \pi y} \right) \, ,    \label{eq: LT-phi-e1}
\end{align}
\end{subequations}
where we have permuted sum and integral due to absolute convergence and used the notation introduced in Eq.~\eqref{eq: e1-def}. Finally, applying Eqs.~\eqref{eq: expansion-fkN} and~\eqref{eq: expansion-gkN}, we find the desired expressions for the Borel--Laplace sums at $0$ of $\tilde{f}_{k,N}(y)$ and $\tilde{g}_{k,N}(y)$, which is analytic for $\Re(y)>0$ by the properties of the function $\e$. The Borel--Laplace sums at $\pi$ are computed~analogously.
\end{proof}

We will now prove that the Borel–Laplace sums of $\tilde{f}_{k,N}$ reconstruct the original $q$-Pochhammer symbol $f_{k,N}$ up to corrections that are suitably encoded in the functions $g_{k,N}$ and $g_{\underline{k},N}$.
\begin{lemma}\label{lemma:summability-f-kN}
    The Borel--Laplace sums at angles $0$ and $\pi$ of the asymptotic expansion $\tilde{f}_{k,N}(y)$ in Eq.~\eqref{eq: expansion-fkN} satisfy the relations
\begin{subequations} \label{eq: BL-fkN}
\begin{align}
    s_0(\tilde{f}_{k,N})(y)&=f_{k,N}(y)-g_{\underline{k},N}\big(-\tfrac{1}{Ny}\big) \,, \quad \Re(y)>0 \, , \label{eq: BL-0-fkN} \\
    s_\pi(\tilde{f}_{k,N})(y)&=f_{k,N}(y)-g_{k,N}\big(-\tfrac{1}{Ny}\big) \, , \quad \Re(y)<0 \, .  \label{eq: BL-pi-fkN}
\end{align}
\end{subequations}
\end{lemma}

\begin{proof}
Recall from Eq.~\eqref{eq: expansion-fkN} and Proposition~\ref{prop:Borel-f_kN} that the asymptotic expansion $\tilde{f}_{k,N}(y)$ is governed by a Gevrey-$1$ formal power series $\psi_k(y)$ whose Borel transform has a single tower of simple poles spaced by integer multiples of the constant $2 \pi \ri /N$ along the imaginary axis. Applying the definition of the Laplace transform in Eq.~\eqref{eq: Laplace} to the resummed expression for $\borel[\psi_k](\zeta)$ in Eq.~\eqref{eq: intBorel2infty-proof}, we compute explicitly the Borel--Laplace sums of $\psi_k(y)$ along the positive and negative real axis. The details of the computation are given in Appendix~\ref{app: proofs}. 
\end{proof}

It follows that the $q$-Pochhammer symbol $f_{k,N}$ in Eq.~\eqref{eq:f_kN} cannot be recovered via the median resummation of its asymptotic expansion along the positive imaginary axis. 
\begin{theorem}
The median resummation at angle $\pi/2$ of the asymptotic expansion $\tilde{f}_{k,N}(y)$ in Eq.~\eqref{eq: expansion-fkN} is given by
\be \label{eq: median-fkN}
\mathcal{S}_{\frac{\pi}{2}}^{\mathrm{med}}\tilde{f}_{k,N}(y)= f_{k,N}(y)-\frac{1}{2}\left[ g_{k,N}\big(-\tfrac{1}{Ny}\big)+g_{\underline{k},N}\big(-\tfrac{1}{Ny}\big) \right]\,,
\ee
where we take $\Im(y)>0$.
\end{theorem}
\begin{proof}
Note that Lemma~\ref{lemma:summability-f-kN} implies that the discontinuity of $\tilde{f}_{k,N}(y)$ across the Stokes ray at angle $\pi/2$ can be written as
\be
\begin{aligned} \label{eq: disc-fkN}
\mathrm{disc}_{\frac{\pi}{2}}\tilde{f}_{k,N}(y)&= s_\pi(\tilde{f}_{k,N})(y) - s_0(\tilde{f}_{k,N})(y) \\
&=-g_{k,N}\big(-\tfrac{1}{Ny}\big)+g_{\underline{k},N}\big(-\tfrac{1}{Ny}\big) \,,
\end{aligned}
\ee
as expected from Corollary~\ref{cor: disc-Rn}.
Applying the definition in Eq.~\eqref{eq: median2} together with Eqs.~\eqref{eq: BL-fkN} and~\eqref{eq: disc-fkN} then yields the desired result.
\end{proof}
Mirroring the previous statements and relying on the support of numerical tests, we conjecture that the Borel–Laplace sums of $\tilde{g}_{k,N}$ reconstruct the original $q$-Pochhammer symbol $g_{k,N}$ up to corrections that are suitably encoded in the functions $f_{k,N}$ and $f_{\underline{k},N}$.

\begin{conjecture}\label{conj:summability-g-kN}
    The Borel--Laplace sums at angles $0$ and $\pi$ of the asymptotic expansion $\tilde{g}_{k,N}(y)$ in Eq.~\eqref{eq: expansion-gkN} satisfy the relations
\begin{subequations} \label{eq: BL-gkN}
\begin{align}
    s_0(\tilde{g}_{k,N})(y)&=g_{k,N}(y)-f_{k,N}\big(-\tfrac{1}{Ny}\big) +\log(1-\re^{2\pi\ri k/N}) \,, \quad \Re(y)>0 \, , \label{eq: BL-0-gkN} \\
    s_\pi(\tilde{g}_{k,N})(y)&=g_{k,N}(y)-f_{\underline{k},N}\big(-\tfrac{1}{Ny}\big) +\log(1-\re^{-2\pi\ri k/N}) \, , \quad \Re(y)<0 \, .  \label{eq: BL-pi-gkN}
\end{align}
\end{subequations}
\end{conjecture}
Despite the formal similarities with $\tilde{f}_{k,N}$, the analog to the analytic argument of Lemma~\ref{lemma:summability-f-kN} for the asymptotic expansion $\tilde{g}_{k,N}$ is still missing. However, assuming Conjecture~\ref{conj:summability-g-kN} holds, then the discontinuity of $\tilde{g}_{k,N}(y)$ across the Stokes ray at angle $\pi/2$ can be written as
\be
\begin{aligned} \label{eq: disc-gkN}
\mathrm{disc}_{\frac{\pi}{2}}\tilde{g}_{k,N}(y)&= s_\pi(\tilde{g}_{k,N})(y) - s_0(\tilde{g}_{k,N})(y) \\
&=f_{k,N}\big(-\tfrac{1}{Ny}\big)-f_{\underline{k},N}\big(-\tfrac{1}{Ny}\big) - \frac{\pi \ri}{N} (2k-N) \,,
\end{aligned}
\ee
as expected from Corollary~\ref{cor: disc-Sn}.
Importantly, as observed before for the dual $q$-Pochhammer symbol $f_{k,N}(y)$, also the function $g_{k,N}(y)$ in Eq.~\eqref{eq:f_kN} cannot be recovered via the median resummation of its asymptotic expansion $\tilde{g}_{k,N}(y)$ along the positive imaginary axis. Indeed, applying the definition in Eq.~\eqref{eq: median2} together with Eqs.~\eqref{eq: BL-gkN} and~\eqref{eq: disc-gkN} yields
\be \label{eq: median-gkN}
\mathcal{S}_{\frac{\pi}{2}}^{\mathrm{med}}\tilde{g}_{k,N}(y)= g_{k,N}(y)-\frac{1}{2}\left[ f_{k,N}\big(-\tfrac{1}{Ny}\big)-\log(1-\re^{2\pi\ri k/N})+f_{\underline{k},N}\big(-\tfrac{1}{Ny}\big)-\log(1-\re^{-2\pi\ri k/N}) \right]\, ,
\ee
where we take $\Im(y)>0$. 

\begin{rmk}
The failure of the median resummation stated in Eqs.~\eqref{eq: median-fkN} and~\eqref{eq: median-gkN} is not in conflict with the modular resurgence conjectures of~\cite{FR1maths} (here, Conjectures~\ref{conj:quantum_modular1-intro} and~\ref{conj:quantum_modular2-intro}) because the asymptotic series $\psi_k(y)$ and $\varphi_k(y)$ in Eqs.~\eqref{eq: tilde-psi} and~\eqref{eq: tilde-phi} do not fit the definition of an MRS~\cite[Defs.~3.1 and~3.2]{FR1maths} (here, Definition~\ref{def:modular_res_struct}), as stressed in Remark~\ref{rmk: not-L-functs}.
Meanwhile, generalizing the construction of~\cite{FR1phys}, we will prove in Section~\ref{sec:weighted-sum} that suitable weighted sums of the functions $f_{k,N}(y)$ and $g_{k,N}(y)$ over the integer $k \in \IZ_N$ realize the modular resurgence paradigm of~\cite{FR1maths} and are effectively reconstructed through the median resummation.  
\end{rmk}

\subsection{Quantum modularity}
We begin this section by recalling the definition and fundamental properties of Faddeev's quantum dilogarithm~\cite{Faddeev2, AK, GK_faddev}. 
%We also refer to~\cite[Appendix~A]{AK} for a detailed review. 

\begin{definition}
The Faddeev's quantum dilogarithm $\Phi_{\mb}(x)$ is defined for $| \Im (x) | < | \Im (c_{\mb}) |$, where 
\be \label{eq: cb-def}
c_{\mb} = \frac{\ri}{2} (\mb + \mb^{-1}) \, , 
\ee
by the integral representation
\be \label{eq: intPhib}
\Phi_{\mb}(x) = \exp \left( \int_{\IR + \ri \epsilon} \frac{\re^{-2 \ri x z}}{4 \sinh(z \mb ) \sinh(z \mb^{-1})} \frac{d z}{z} \right) \, ,
\ee
which implies the symmetry properties
\be \label{eq: symmPhib}
\Phi_{\mb}(x) = \Phi_{-\mb}(x) = \Phi_{\mb^{-1}}(x) \, .
\ee
\end{definition}
By means of Eq.~\eqref{eq: symmPhib}, it can be extended to the region $\Im(\mb^2) < 0$ and it further admits analytic continuation to all values of $\mb$ such that $\mb^2 \notin \IR_{\le 0}$. Moreover, $\Phi_{\mb}(x)$ can be extended to the whole complex $x$-plane as a meromorphic function with an essential singularity at infinity, poles at the points 
\be
x = c_{\mb} + \ri m \mb + \ri n \mb^{-1} \, ,
\ee
and zeros at the points 
\be
x = -c_{\mb} - \ri m \mb - \ri n \mb^{-1} \, ,
\ee
for $m,n \in \IN$. 
When $\Im(\mb^2) > 0$, the formula in Eq.~\eqref{eq: intPhib} can be equivalently expressed in terms of $q$-Pochhammer symbols as
\be \label{eq: seriesPhib}
\Phi_{\mb}(x) = \frac{( \re^{2 \pi \mb (x + c_{\mb})}; \, q)_{\infty}}{(  \re^{2 \pi \mb^{-1} (x - c_{\mb})}; \, \tilde{q})_{\infty}} \, ,
\ee
where $q = \re^{2 \pi \ri \mb^2}$ and $\tilde{q} = \re^{- 2 \pi \ri \mb^{-2}}$. 

We will now prove that the $q$-Pochhammer symbols $f_{k,N}$ and $g_{k,N}$ in Eq.~\eqref{eq:f_kN} are holomorphic quantum modular functions for which the Faddeev's quantum dilogarithm plays the role of the cocyle.
Let us start by introducing the holomorphic functions $F_k, G_k \colon \IH\to\IC$ as
\begin{subequations} \label{eq: FG-def}
    \begin{align}
        F_k(y)&:=f_{k,N}(y)-g_{\underline{k},N}\big(-\tfrac{1}{Ny}\big)\, , \label{eq: F-def}\\
        G_k(y)&:=g_{k,N}(y)-f_{k,N}\big(-\tfrac{1}{Ny}\big)\, . \label{eq: G-def}
    \end{align}
\end{subequations}
These can be written in terms of the logarithm of Faddeev's quantum dilogarithm. 
Indeed, setting $\mb^2=y$, it follows from the infinite product representations in Eqs.~\eqref{eq:f_kN} and~\eqref{eq: seriesPhib} that
\begin{subequations}
\begin{align}
F_k(y)&=\log\Phi_{\mathsf{b}}\Big(\frac{k \ri}{N\mathsf{b}}-c_\mathsf{b}\Big)\label{eq: F-phib} \, , \\
G_k(y/N)&=\log\Phi_{\mathsf{b}}\Big(\frac{k\ri\mathsf{b}}{N}-c_\mathsf{b}\Big) - \log(1-\zeta_N^k)\label{eq: G-phib} \, ,
\end{align}
\end{subequations} 
which can be analytically continued to the cut complex plane $\IC':= \IC\setminus\IR_{\le0}$ due to the analytic properties of $\Phi_\mb$~\cite{GK_faddev}. 
Furthermore, recalling the expressions for the Borel--Laplace sums of the asymptotic expansions $\tilde{f}_{k,N}$ and $\tilde{g}_{k,N}$ in Eqs.~\eqref{eq: BL-fkN} and~\eqref{eq: BL-gkN}, we obtain the equalities
\begin{subequations} \label{eq: FG-BL}
    \begin{align}
        F_k(y) = s_0(\tilde{f}_{k,N})(y) \, , \quad &G_k(y) = -s_\pi(\tilde{f}_{k,N})\big(- \tfrac{1}{Ny} \big) \, , \quad \Re(y)>0 \, ,  \label{eq: FG-BL-p}\\
        F_k(y) = -s_0(\tilde{g}_{\underline{k},N})\big(- \tfrac{1}{Ny} \big) +\log(1-\re^{2\pi\ri\frac{k}{N}})  \, , \quad &G_k(y) =s_\pi(\tilde{g}_{k,N})(y) -\log(1-\re^{2\pi\ri \frac{k}{N}}) \, , \quad \Re(y)<0 \, . \label{eq: FG-BL-m}
    \end{align}
\end{subequations}
Note that each of the functions $F_k$ and $G_k$ contains the information of both sequences of Stokes constants $\{S_n^k\}$ and $\{R_n^k\}$ by means of Proposition~\ref{prop: BL-e1}.

\begin{theorem}\label{thm:qm}
For every $k\in\IZ_N$, the functions $f_{k,N},g_{k,N}\colon\IH\to\IC$ in Eq.~\eqref{eq:f_kN} are holomorphic quantum modular functions for the group $\Gamma_N \subset \mathsf{SL}_2(\IZ)$ generated by the elements in Eq.~\eqref{eq: GammaN-generators}.
\end{theorem}
\begin{proof}
Applying Definition~\ref{def: holoQM}, it is sufficient to prove the analyticity in $\IC'$ of the cocycles for the two generators $T$ and $\gamma_N$, that is,
\begin{subequations}
    \begin{align}
        h_T[f](y)&=f(y+1)-f(y) \, , \\
        h_{\gamma_N}[f](y)&=f\big(\tfrac{y}{Ny+1}\big)-f(y)\, ,
    \end{align}
\end{subequations}
for $f=f_{k,N}, g_{k,N}$.
The periodicity of the expressions in Eq.~\eqref{eq:f_kN} directly imply that the cocycles for $T$ are trivial, \emph{i.e.}, 
$h_T[f_{k,N}](y) = 0$ and $h_T[g_{k,N}](y)=0$.
Meanwhile, it follows from Eq.~\eqref{eq: F-def} that the cocycle of $f_{k,N}$ for $\gamma_N$ is
\be \label{eq: fkN-cocycle}
\begin{aligned}
h_{\gamma_N}[f_{k,N}](y)&=F_k\big(\tfrac{y}{Ny+1}\big)+g_{\underline{k},N}\big(-\tfrac{1}{Ny}\big)-f_{k,N}(y)\\
&=F_k\big(\tfrac{y}{Ny+1}\big)-F_k(y)\, ,
\end{aligned}
\ee
which is analytic in $\IC'$ because of Eq.~\eqref{eq: F-phib}.
Analogously, Eq.~\eqref{eq: G-def} implies that the cocycle of $g_{k,N}$ for $\gamma_N$ is
\be \label{eq: gkN-cocycle}
    \begin{aligned}
        h_{\gamma_N}[g_{k,N}](y)&=
         G_k\big(\tfrac{y}{Ny+1}\big)+f_{k,N}\big(-\tfrac{1}{Ny}\big)-g_{k,N}(y)\\
        &=G_k\big(\tfrac{y}{Ny+1}\big)-G_k(y) \, ,
    \end{aligned}
\ee
which is analytic on $\IC'$ because of Eq.~\eqref{eq: G-phib}. 
\end{proof}

We stress that the cocycles of the $q$-Pochhammer symbols $f_{k,N}$ and $g_{k,N}$ for the generator $\gamma_N$ in Eqs.~\eqref{eq: fkN-cocycle} and~\eqref{eq: gkN-cocycle} are simply determined by the functions $F_k$ and $G_k$ as
\be
h_{\gamma_N}[f_{k,N}](y)=h_{\gamma_N}[F_k](y) \, , \quad h_{\gamma_N}[g_{k,N}](y)=h_{\gamma_N}[G_k](y) \, , 
\ee
which are themselves determined by the Borel--Laplace sums of the asymptotic expansions $\tilde{f}_{k,N}$ and $\tilde{g}_{k,N}$, as observed in Eqs.~\eqref{eq: FG-BL-p} and~\eqref{eq: FG-BL-m}. 

\subsubsection{\texorpdfstring{$q$}{q}-Pochhammer symbols as quantum Jacobi forms}

In~\cite[Thm.~A-41]{wheeler-thesis}, the author showed that $q$-Pochhammer symbols of the form 
\be \label{eq: qPochh-Jacobi}
f_\lambda(y):=\log(\re^{2\pi\ri\lambda}; \, q)_\infty \,, \quad \lambda\in\IC \,  , 
\quad q=\re^{2 \pi \ri y} \, , 
\ee
are quantum Jacobi functions\footnote{Quantum Jacobi forms were first defined by Bringmann and Folsom in~\cite{Folsom-Brig}.} for $\mathsf{SL}_2(\IZ)$. 
In particular, for every $\gamma=\left( \begin{smallmatrix}
    a & b\\
    c & d
\end{smallmatrix} \right)\in\mathsf{SL}_2(\IZ)$, the corresponding cocycle is given by the ratio of $q$-Pochhammer symbols
\be
    f_{\frac{\lambda}{cy+d}}\left(\frac{ay+b}{cy+d}\right)-f_\lambda(y)=\log\frac{(\re^{\frac{2\pi\ri\lambda}{cy+d}}; \, \tilde{q}_\lambda)_\infty}{ (\re^{2\pi\ri\lambda}; \, q)_\infty} \, , 
\ee
where we have introduced the notation
\be
\tilde{q}_\gamma:=\re^{2\pi\ri\gamma y}=\re^{2\pi\ri\frac{ay+b}{cy+d}} \, ,
\ee
and the function 
\be \label{eq: Phi-lambda}
\Phi_\gamma(\lambda,y):=\frac{(\re^{2\pi\ri\lambda}; \, q)_\infty}{(\re^{\frac{2\pi\ri\lambda}{cy+d}}\, \tilde{q}_\gamma ; \, \tilde{q}_\gamma)_\infty}
\ee
is holomorphic in $\IC_\gamma$. Note that the function in Eq.~\eqref{eq: Phi-lambda} agrees with the Faddeev's quantum dilogarithm when $\gamma=S\colon y\to -1/y$. More precisely, 
\begin{equation}
   \Phi_S(\lambda,y)=\Phi_{\sf b}(\ri\lambda{\sf b^{-1}}-c_{\sf b})\,. 
\end{equation} 

Our previous results then prove that the quantum Jacobi function $f_\lambda(y)$ in Eq.~\eqref{eq: qPochh-Jacobi} is indeed a holomorphic quantum modular form when $\lambda \in \IQ$. Specifically, if $\re^{2\pi\ri\lambda}$ is an $N$-th root of unity, $f_\lambda(y)$ is a holomorphic quantum modular function for the group $\Gamma_N \subset \mathsf{SL}_2(\IZ)$ generated by the elements $T$ and $\gamma_N$ in Eq.~\eqref{eq: GammaN-generators}---as stated in Theorem~\ref{thm:qm}.
In fact, we will now show that the converse is also true.
\begin{prop}
The $q$-Pochhammer symbol $f_\lambda\colon\IH\to\IC$ in Eq.~\eqref{eq: qPochh-Jacobi} is a holomorphic quantum modular function for $\Gamma_N$ if and only if $\lambda=\frac{k}{N}$ for $N \in \IZ_{\ge 2}$ and $k\in\IZ_N$.    
\end{prop}
\begin{proof}
Let us fix $N \in \IZ_{\ge 2}$. By construction, the cocycle for $T$ is trivial, that is, $h_T[f_\lambda](y)=0$. Meanwhile, the cocycle for $\gamma_N$ can be written as 
\be \label{eq: cocycle-Jacobi-QM}
\begin{aligned}
h_{\gamma_N}[f_\lambda](y) &= \log\frac{(\re^{2\pi\ri\lambda};\tilde{q}_{\gamma_N})_\infty}{(\re^{2\pi\ri\lambda};q)_\infty}  \\
&= \log\bigg[\frac{(\re^{2\pi\ri\lambda (Ny+1)};q)_\infty}{(\re^{2\pi\ri\lambda};q)_\infty}\Phi_{\gamma_N}(\lambda(Ny+1),y)\bigg]-\log(1-\re^{2\pi\ri\lambda})\\
&=\log\bigg[\frac{(\re^{2\pi\ri\lambda}q^{\lambda N};q)_\infty}{(\re^{2\pi\ri\lambda};q)_\infty}\Phi_{\gamma_N}(\lambda(Ny+1),y)\bigg]-\log(1-\re^{2\pi\ri\lambda})\, . 
\end{aligned}
\ee
Since the Faddeev's quantum dilogarithm $\Phi_{\gamma_N}(\lambda(Ny+1),y)$ extends to an analytic function on $\IC_{\gamma_N}$, we are left to study the analytic properties of the ratio of $q$-Pochhammer symbols
\begin{equation}\label{eq:ratio}
    \frac{(\re^{2\pi\ri\lambda}q^{\lambda N};q)_\infty}{(\re^{2\pi\ri\lambda};q)_\infty}\,.
\end{equation}
By the $q$-binomial theorem, we can write it as the infinite sum
\begin{equation}
    \frac{(\re^{2\pi\ri\lambda}q^{\lambda N};q)_\infty}{(\re^{2\pi\ri\lambda};q)_\infty}=\sum_{n=0}^\infty\frac{(q^{\lambda N};q)_n}{(q;q)_n}\re^{2\pi\ri\lambda n}\,,
\end{equation}
which is convergent for $y\in\IH$. However, if $\lambda N\notin\IZ$, the term $(q^{\lambda N};q)_n$ will not vanish for $y\in\IZ$, while the term $(q;q)_n$ will. It follows that the function in Eq.~\eqref{eq:ratio} with $\lambda N\notin\IZ$ cannot be extended to $y\in\IZ$---hence, to $y \in \IC_{\gamma_N}$. On the other hand, $\lambda N\in\IZ$ is equivalent to $\lambda N \in\IZ_N$ by definition of the $q$-Pochhammer symbol in Eq.~\eqref{eq: qPochh-Jacobi}. In this case, $f_\lambda= f_{k,N}$, and Theorem~\ref{thm:qm} applies.
\end{proof}

\subsubsection{Fricke involution} \label{sec:fricke}

Along the lines of~\cite{FR1phys}, we construct a second family of pairs of holomorphic quantum modular functions for $\Gamma_N$ by acting on the $q$-Pochhammer symbols $f_{k,N}(y)$ and $g_{k,N}(y)$ with the Fricke involution of $\IH/\Gamma_1(N)$, that is, the transformation $y \mapsto -\frac{1}{Ny}$.
In particular, let us define the holomorphic functions $f_{k,N}^{\star}, g_{k,N}^{\star}\colon \IH \to \IC$ as
\be \label{eq: fricke-fs}
f_{k,N}^{\star}(y):=f_{k,N}\big(-\tfrac{1}{N y}\big) \, , \quad  g_{k,N}^{\star}(y):=g_{k,N}\big(-\tfrac{1}{N y}\big) \, ,
\ee
where the symbol $^{\star}$ denoted the action of the Fricke involution.
\begin{theorem}\label{thm:fricke-qPochh}
The functions $f_{k,N}^{\star}, g_{k,N}^{\star}\colon \IH\to\IC$, defined in Eq.~\eqref{eq: fricke-fs} as the images of the $q$-Pochhammer symbols in Eq.~\eqref{eq:f_kN} under Fricke involution, are holomorphic quantum modular functions for the group $\Gamma_N$.
\end{theorem}
\begin{proof}
Recall that the generators of $\Gamma_N$ are the matrices $T$ and $\gamma_N$ in Eq.~\eqref{eq: GammaN-generators} and the corresponding cocycles are
\begin{subequations}
    \begin{align}
        h_T[f^\star](y)&=f^{\star}(y+1)-f^{\star}(y) = f\big(-\tfrac{1}{Ny+N}\big) - f\big(-\tfrac{1}{Ny}\big)\, , \\
        h_{\gamma_N}[f^\star](y)&=f^{\star}\big(\tfrac{y}{Ny+1}\big)-f^{\star}(y) =f\big(-\tfrac{1}{Ny}-1\big) - f\big(-\tfrac{1}{Ny}\big) \, ,
    \end{align}
\end{subequations}
where we take $f=f_{k,N},\, g_{k,N}$.
It follows from the periodicity properties of the $q$-Pochhammer symbols that the cocycles for $\gamma_N$ are trivial, \emph{i.e.}, $h_{\gamma_N}[f_{k,N}^{\star}](y) = 0$ and $h_{\gamma_N}[g_{k,N}^{\star}](y) = 0$.
Meanwhile, it follows from Eq.~\eqref{eq: G-def} that the cocycle of $f_{k,N}^{\star}$ for $T$ is
\be \label{eq: fstar-cocycle}
    \begin{aligned}
       h_T[f_{k,N}^{\star}](y)&=-G_k(y+1)+g_{k,N}(y+1)-f_{k,N}\big(-\tfrac{1}{Ny}\big)\\
        &=G_k(y)-G_k(y+1)\, ,
    \end{aligned}
\ee
which is analytic in $\IC'$ because of Eq.~\eqref{eq: G-phib}.
Analogously, it follows from Eq.~\eqref{eq: F-def} that the cocycle of $g_{k,N}^{\star}$ for $T$ is
\be \label{eq: gstar-cocycle}
    \begin{aligned}
        h_T[g_{k,N}^{\star}](y)&=-F_{\uk}(y+1)+f_{\uk,N}(y+1)-g_{k,N}\big(-\tfrac{1}{Ny}\big)\\
        &=F_{\uk}(y)-F_{\uk}(y+1)\, ,
    \end{aligned}
\ee
which is analytic in $\IC'$ because of Eq.~\eqref{eq: F-phib}.
\end{proof}

Note that the cocycles of the dual functions $f_{k,N}^{\star}$ and $g_{k,N}^{\star}$ for the generator $T$ in Eqs.~\eqref{eq: fstar-cocycle} and~\eqref{eq: gstar-cocycle} are again determined by Faddeev's quantum dilogarithm via the functions $F_k$ and $G_k$. Namely, 
\be
h_T[f_{k,N}^{\star}](y)=-h_T[G_k](y) \, , \quad h_T[g_{k,N}^{\star}](y)=-h_T[F_{\uk}](y) \, , 
\ee
which in turn are determined by the Stokes constants $S^k_n$, $R^k_n$, $n \in \IZ_{\ne 0}$, of the original $q$-Pochhammer symbols through the Borel--Laplace sums in Eqs.~\eqref{eq: FG-BL-p} and~\eqref{eq: FG-BL-m}. 

\section{Weighted sums}\label{sec:weighted-sum}

In this section, after fixing $N \in \IZ_{\ge 2}$, we consider the functions $\frakf(y)$, $\frakg(y)$, $y \in \IH$, introduced in Eq.~\eqref{eq:fg-intro} as sums of the $q$-Pochhammer symbols $f_{k,N}(y)$, $g_{k,N}(y)$ over $k \in \IZ_N$ weighted by a Dirichlet character of modulus $N$. We first study the summability and quantum modularity properties of their asymptotic expansions in the limit $y \to 0$ with $\Im(y)>0$ and then delve into the details of their resurgent structures. Notably, when $\chi_N$ is primitive and odd, the weighted sums $\frakf$ and $\frakg$ produce a pair of MRSs that together satisfy the modular resurgence paradigm of~\cite{FR1maths} and provide new evidence of Conjectures~\ref{conj:quantum_modular1-intro} and~\ref{conj:quantum_modular2-intro}.

\subsection{Resurgence, summability, and quantum modularity}\label{sec:sum-weight}

We begin by considering the asymptotic series $\tilde{\frakf}(y)$, $\tilde{\frakg}(y)$ of the functions $\frakf(y)$, $\frakg(y)$, respectively, in the limit $y \rightarrow 0$ with $\Im(y)>0$. We will then show how the summability and quantum modularity properties of $\tilde{\frakf}(y)$ and $\tilde{\frakg}(y)$ are straightforwardly derived from the results of Section~\ref{sec: single}. Remarkably, and contrary to the case of the individual $q$-Pochhammer symbols $f_{k,N}$ and $g_{k,N}$, the median resummation proves to be effective in reconstructing the weighted sums $\frakf$ and $\frakg$ under the assumption that the Dirichlet character $\chi_N$ appearing in Eq.~\eqref{eq:fg-intro} is odd.

\subsubsection{Resurgence and summability of \texorpdfstring{$\mathfrak{f}$}{f}}

The asymptotic expansion $\tilde{\frakf}(y)$  of the function $\frakf(y)$ in Eq.~\eqref{eq:fg-intro} for $y \rightarrow 0$ with $\Im(y)>0$ is given by 
\begin{equation}\label{eq:formal-f0}
    \tilde{\frakf}(y)=\sum_{k\in\IZ_N}\chi_N(k)\tilde{f}_{k,N}(y)\,,
\end{equation}
where $\tilde{f}_{k,N}(y)$ is defined in Eq.~\eqref{eq: expansion-fkN}. 

\begin{prop}\label{prop:resurgence-f0}
The asymptotic expansion $\tilde{\frakf}(y)$ in Eq.~\eqref{eq:formal-f0} is Gevrey-1 and simple resurgent.   
\end{prop}
\begin{proof}
The statement follows from Proposition~\ref{prop:Borel-f_kN}.
\end{proof}

When the Dirichlet character $\chi_N$ entering the definition of the weighted sum of $q$-Pochhammer symbols in Eq.~\eqref{eq:fg-intro} is odd, \emph{i.e.}, it satisfies $\chi_N(-1)=-1$, a stronger result applies---namely, we prove the effectiveness of the median resummation of $\tfrakf$.

\begin{theorem}\label{thm:summability-f0}
    The function $\frakf\colon\IH\to\IC$ defined in Eq.~\eqref{eq:fg-intro} agrees with the median resummation of its asymptotic expansion $\tilde{\frakf}(y)$ if and only if the Dirichlet character $\chi_N$ is~odd. 
\end{theorem}

\begin{proof}
It follows from Eq.~\eqref{eq: median-fkN} that
\begin{equation}
\begin{aligned}
     \mathcal{S}^{\rm med}_{\frac{\pi}{2}} \tfrakf(y)&= \sum_{k\in\IZ_N}\chi_N(k)\Big[f_{k,N}(y)-\frac{1}{2}\Big(g_{k,N}\big(-\tfrac{1}{Ny}\big)+g_{\underline{k},N}\big(-\tfrac{1}{Ny}\big)\Big)\Big]\\
     &= \frakf(y)-\frac{1}{2}\sum_{k\in\IZ_N}\chi_N(k)g_{k,N}\big(-\tfrac{1}{Ny}\big)-\frac{1}{2} \sum_{k\in\IZ_N}\chi_N(k)g_{\underline{k},N}\big(-\tfrac{1}{Ny}\big)\\
     &=\frakf(y)-\frac{1}{2}\sum_{k\in\IZ_N}\Big( \chi_N(k) + \chi_N(-k) \Big) g_{k,N}\big(-\tfrac{1}{Ny}\big)\,,
\end{aligned}
\end{equation}
which gives the desired result if and only if $\chi_N$ is odd. An analogous computation holds for the median resummation $\mathcal{S}^{\rm med}_{-\frac{\pi}{2}} \tfrakf(y)$.
\end{proof}

Note that the above result reproduces Theorem~4.6 in~\cite{FR1phys} when taking $N=3$.

\subsubsection{Resurgence and summability of \texorpdfstring{$\mathfrak{g}$}{g}}

The asymptotic expansion $\tilde{\frakg}(y)$ of the function $\frakg(y)$ in Eq.~\eqref{eq:fg-intro} for $y\to 0$ with $\Im(y)>0$ is given by 
\begin{equation}\label{eq:formal-f-inf}
    \tilde{\frakg}(y)=\sum_{k\in\IZ_N}\chi_N(k)\tilde{g}_{k,N}(y)\,,
\end{equation}
where $\tilde{g}_{k,N}(y)$ is defined in Eq.~\eqref{eq: expansion-gkN}.

\begin{prop}\label{prop:resurgence-finf}
 The asymptotic expansion $\tilde{\frakg}(y)$ in Eq.~\eqref{eq:formal-f-inf} is Gevrey-1 and simple resurgent. 
\end{prop}
\begin{proof}
The statement follows from Proposition~\ref{prop:Borel-g_kN}.
\end{proof}

Let us now assume the validity of Conjecture~\ref{conj:summability-g-kN}. Analogously to the case of $\frakf$, when the Dirichlet character $\chi_N$ appearing in Eq.~\eqref{eq:fg-intro} satisfies $\chi_N(-1)=-1$, the effectiveness of the median resummation of $\tfrakg$ follows.

\begin{theorem}\label{thm:summability-f-inf}
    If Conjecture~\ref{conj:summability-g-kN} holds, then the function $\frakg\colon\IH\to\IC$ defined in Eq.~\eqref{eq:fg-intro} agrees with the median resummation of its asymptotic expansion $\tfrakg(y)$ if and only if the Dirichlet character $\chi_N$ is odd.
\end{theorem}

\begin{proof}
It follows from Eq.~\eqref{eq: median-gkN} that
\begin{equation}
\begin{aligned}
     \mathcal{S}^{\rm med}_{\frac{\pi}{2}} \tfrakg(y)=& \sum_{k\in\IZ_N}\chi_N(k)\Big[g_{k,N}(y)-\frac{1}{2}\Big(f_{k,N}\big(-\tfrac{1}{Ny}\big)-\log(1-\re^{-2\pi\ri k/N})\\
     &\qquad\qquad\qquad+f_{\underline{k},N}\big(-\tfrac{1}{Ny}\big)-\log(1-\re^{2\pi\ri k/N})\Big)\Big]\\
     =&\; \frakg(y)-\frac{1}{2}\sum_{k\in\IZ_N}\chi_N(k)f_{k,N}\big(-\tfrac{1}{Ny}\big)-\frac{1}{2} \sum_{k\in\IZ_N}\chi_N(k)f_{\underline{k},N}\big(-\tfrac{1}{Ny}\big)\\
     &\;+\frac{1}{2}\sum_{k\in\IZ_N}\chi_N(k)\log(1-\re^{-2\pi\ri k/N})+\frac{1}{2}\sum_{k\in\IZ_N}\chi_N(k)\log(1-\re^{2\pi\ri k/N})\\
     =&\;\frakg(y)-\frac{1}{2}\sum_{k\in\IZ_N}\Big( \chi_N(k) + \chi_N(-k) \Big)\Big[f_{k,N}\big(-\tfrac{1}{Ny}\big)-\log(1-\re^{-2\pi\ri k/N})\Big]\,,
\end{aligned}
\end{equation}
which gives the desired result if and only if $\chi_N$ is odd. An analogous calculation holds for the median resummation $\mathcal{S}^{\rm med}_{-\frac{\pi}{2}} \tfrakg(y)$.
\end{proof}

Note that the above result reproduces Conjecture~1 in~\cite{FR1phys} when taking $N=3$.

\subsubsection{Quantum modularity}

In Section~\ref{sec: single}, we proved that the $q$-Pochhammer symbols $f_{k,N}$ and $g_{k,N}$ introduced in Eq.~\eqref{eq:f_kN} are holomorphic quantum modular functions according to the definition given by Zagier~\cite{zagier-talk}. Here, we show that the same statement applies to the weighted sums $\frakf$ and $\frakg$ in Eq.~\eqref{eq:fg-intro}. 

\begin{cor}\label{cor:qm-weighted}
The functions $\frakf,\frakg\colon\IH\to\IC$ defined in Eq.~\eqref{eq:fg-intro} are holomorphic quantum modular functions for the group $\Gamma_N\subset\mathsf{SL}_2(\IZ)$ generated by the elements in Eq.~\eqref{eq: GammaN-generators}.
\end{cor}
\begin{proof}
The statement follows from Theorem~\ref{thm:qm}.
\end{proof}

Observe that the non-trivial cocycles $h_{\gamma_N}[\frakf]$ and $h_{\gamma_N}[\frakg]$ are linear combinations of the logarithm of Faddeev's quantum dilogarithm $\Phi_\mb$ in Eq.~\eqref{eq: seriesPhib}. Namely, 
\begin{subequations}
    \begin{align}
        h_{\gamma_N}[\frakf](y) &= \sum_{k\in\IZ_N}\chi_N(k) h_{\gamma_N}[f_{k,N}](y) = \sum_{k\in\IZ_N}\chi_N(k) h_{\gamma_N}[F_k](y) \, ,\\
        h_{\gamma_N}[\frakg](y) &= \sum_{k\in\IZ_N}\chi_N(k) h_{\gamma_N}[g_{k,N}](y) = \sum_{k\in\IZ_N}\chi_N(k) h_{\gamma_N}[G_k](y) \, ,
    \end{align}
\end{subequations}
where the functions $F_k$ and $G_k$ are written explicitly in terms of $\Phi_\mb$ in Eqs.~\eqref{eq: F-phib} and~\eqref{eq: G-phib}.

\subsection{Modular resurgence}\label{sec:modular-res-paradigm}

We will now compute the resurgent structures of the asymptotic series $\tilde{\frakf}(y)$ and $\tilde{\frakg}(y)$ in Eqs.~\eqref{eq:formal-f0} and~\eqref{eq:formal-f-inf}. Before doing so, we construct a pair of $L$-functions associated to the Dirichlet character $\chi_N$ by taking suitable products of its Dirichlet $L$-function and the Riemann zeta function and prove that their completions satisfy the same combined functional equation. In other words, they are one the analytic continuation of the other. Then, we find that for a certain class of Dirichlet characters, the inverse Mellin transform of these $L$-functions produces the weighted sums of $q$-Pochhammer symbols $\frakf$ and $\frakg$ that we are interested in. As a consequence, the asymptotic series $\tfrakf$ and $\tfrakg$ are proven to be modular resurgent and to fit in the modular resurgence paradigm of~\cite{FR1maths}. Hence, for every $N \in \IZ_{\ge 2}$, the functions in Eq.~\eqref{eq:fg-intro} provide new evidence of Conjectures~\ref{conj:quantum_modular1-intro} and~\ref{conj:quantum_modular2-intro}.

\subsubsection{Dirichlet characters, generating functions, and \texorpdfstring{$L$}{L}-functions}

We begin by showing that the weighted sums $\frakf$ and $\frakg$ in Eq.~\eqref{eq:fg-intro} are the generating functions of two interesting divisor sum functions.
\begin{lemma}\label{lem:f0N--gen}
 Let $\chi_N$ be a primitive Dirichlet character of modulus $N$ and fix $|x|<1$. Then,
 \be\label{eq:gen_func_S}
  \sum_{k\in\IZ_N}\chi_N(k)\log(\zeta_N^k;x)_\infty = - \mathscr{G}(\chi_N) \sum_{m=1}^\infty \sum_{d|m}\frac{1}{d}\,\overline{\chi_N}(d)\, x^m\,,
 \ee
  where $d$ is a positive integer divisor of $m$, $\mathscr{G}(\chi_N)=\sum_{j=1}^{N-1}\chi_N(j) \re^{\frac{2\pi\ri j}{N}}$ is the Gauss sum of $\chi_N$, and $\overline{\chi_N}=\chi_N^{-1}$ is its complex conjugate.  
\end{lemma}

\begin{proof}
Applying the definition of the $q$-Pochhammer symbol and Taylor expanding the logarithm, we have that
\be
 \begin{aligned}
    \sum_{k\in\IZ_N}\chi_N(k)\log(\zeta^k_N;x)_\infty
    &=-\sum_{k\in\IZ_N}\chi_N(k)\sum_{m=0}^\infty \sum_{d|m}\frac{1}{d} \, x^m \, \zeta_N^{k d} \\
    &=-\sum_{m=0}^\infty \sum_{d|m}\frac{1}{d} \, x^m \, \sum_{k\in\IZ_N}\, \chi_N(k)  \zeta_N^{k d} \,.
 \end{aligned}
 \ee
 Let us now distinguish between three cases. 
 \begin{itemize}
     \item[(1)] If ${\rm gdc}(d,N)=1$, then $\chi_N(d)\overline{\chi_N}(d)=1$. Hence, we have that
 \begin{equation}
     \sum_{k\in\IZ_N}\, \chi_N(k)  \zeta_N^{k d}= \overline{\chi_N}(d) \sum_{k\in\IZ_N}\chi_N(k d) \, \zeta_N^{k d} =  \overline{\chi_N}(d)\mathscr{G}(\chi_N)\,.
 \end{equation}
 \item[(2)] If $N \mid d$, then $\zeta_N^{kd}=1$ for every $k\in\IZ_N$. Therefore, 
  \begin{equation}
     \sum_{k\in\IZ_N}\, \chi_N(k)  \zeta_N^{k d}= \sum_{k\in\IZ_N} \chi_N(k)= 0\,,
 \end{equation}
 where in the last step we use that $\chi_N$ is primitive and thus not the principal character.
    \item[(3)] If $\mathrm{gcd}(d, N)=w >1$ and $N \nmid d$, we can write $d=q w$ and $N=p w$ with $\mathrm{gcd}(q,p)=1$ and $p >1$. Then, 
 \begin{equation}
     \sum_{k\in\IZ_N}\, \chi_N(k)  \zeta_N^{k d}=\sum_{\ell\in\IZ_{p}}\sum_{\substack{k\in\IZ_{N} \\ q k\underset{p}{\equiv}\ell}} \chi_N(k) \, \zeta_{p}^{\ell} = \sum_{\ell\in\IZ_{p}} \zeta_{p}^{\ell}\sum_{\substack{k\in\IZ_{N} \\ k\underset{p}{\equiv}\ell/q}} \chi_N(k) \,, 
 \end{equation}
where in the last step we used that $\gcd(q,N)=1$. We set $u=\ell/q$ and observe that, if we vary $n\in\IZ_w$ with fixed $u \in \IZ_p$, then $k=u+n p$ runs over all elements in $\IZ_N$ such that $k\equiv_p u$. Namely, we can write
 \begin{equation}\label{eq:step}
     \sum_{\substack{k\in\IZ_N \\ k\underset{p}{\equiv}u}} \chi_N(k)=\sum_{n\in\IZ_w}\chi_N(u+np)\,.
 \end{equation}
  Since $\chi_N$ is primitive, there exists an integer $c\equiv_p 1 $ that is coprime to $N$ and such that $\chi_N(c)\neq 1$. Hence,
\begin{equation}\label{eq:step2}
\chi_N(c)\sum_{\substack{k\in\IZ_N \\ k\underset{p}{\equiv}u}} \chi_N(k)=\chi_N(c)\sum_{n\in\IZ_w}\chi_N(u+np)=\sum_{n\in\IZ_w}\chi_N(uc+ncp)= \sum_{\substack{k\in\IZ_N \\ k\underset{p}{\equiv}u}} \chi_N(k)   \,,
\end{equation}
where in the first and last steps we used Eq.~\eqref{eq:step}.
 \end{itemize}
Summarizing, we obtain that 
\begin{equation}\label{eq:step-char}
    \sum_{k\in\IZ_N}\, \chi_N(k)  \zeta_N^{k d}=\begin{cases}
       \overline{\chi_N}(d)\mathscr{G}(\chi_N)  & \mbox{if} \quad  {\rm gdc}(d,N)=1 \\
       \\
      0   & \mbox{otherwise}
    \end{cases}\,,
\end{equation}
which concludes the proof.
\end{proof}

\begin{lemma}\label{lem:finfN--gen}
 Let $\chi_N$ be a Dirichlet character of modulus $N$ and fix $|x|<1$. Then,
 \be\label{eq:gen_func_R}
  \sum_{k\in\IZ_N}\chi_N(k)\log(x^{k/N};x)_\infty = -\sum_{m=1}^\infty \sum_{d|m}\frac{d}{m}\,\chi_N(d)\, x^{m/N} \,, 
 \ee
 where $d$ is a positive integer divisor of $m$.
\end{lemma}
\begin{proof}
Applying the definition of the $q$-Pochhammer symbol and Taylor expanding the logarithm, we have that
\be
 \begin{aligned}
    \sum_{k\in\IZ_N}\chi_N(k)\log(x^{k/N};x)_\infty&=-\sum_{k\in\IZ_N}\chi_N(k)\sum_{m=1}^\infty \sum_{\substack{d|m\\ d\equiv_N k}}\frac{d}{m} \, x^{m/N} \\
    &= -\sum_{k\in\IZ_N}\sum_{m=1}^\infty \sum_{\substack{d|m\\ d\equiv_N k}}\frac{d}{m} \, \chi_N(d) \, x^{m/N} \\
    &= -\sum_{m=1}^\infty \sum_{d|m}\frac{d}{m} \, \chi_N(d) \, x^{m/N} \,,
 \end{aligned}
 \ee
 where we have used that $\chi_N(d)=\chi_N(k)$ for $d\equiv_N k$.
\end{proof}

\begin{rmk}
Lemmas~\ref{lem:f0N--gen} and~\ref{lem:finfN--gen} can be equivalently restated in the form
\be \label{eq:gen_func_fg}
    \frakf(y) = - \mathscr{G}(\chi_N) \sum_{m=1}^\infty \sum_{d|m}\frac{1}{d}\,\overline{\chi_N}(d)\, q^m\,, \quad
        \frakg(y/N) = -\sum_{m=1}^\infty \sum_{d|m}\frac{d}{m}\,\chi_N(d)\, q^{m/N} \,, 
\ee
where $x=q=\re^{2 \pi \ri y}$, as before. We stress that the first equation above only applies to primitive characters of modulus $N$.
\end{rmk}

Let us now introduce a pair of $L$-series $L(s)$ and $L'(s)$ associated with a given Dirichlet character $\chi_N$. Specifically, following the prototype of~\cite{Rella22}, we define them as
\be \label{eq:LL'_N}
    L(s):=L(s+1,\overline{\chi_N})\zeta(s) \, , \quad L'(s):=L(s,\chi_N)\zeta(s+1)\,,
\ee
where $L(\chi_N,s)$ is the standard Dirichlet $L$-series associated with $\chi_N$ and $\zeta(s)$ is the Riemann zeta function. Note the symmetric unitary shift in the arguments of the factors.
When $\chi_N$ is primitive, that is, its modulus and conductor are equal, the Dirichlet $L$-series $L(s, \chi_{N})$ can be meromorphically continued to the whole complex $s$-plane. Thus, in this case, each of the two $L$-series in Eq.~\eqref{eq:LL'_N} defines an $L$-function via meromorphic continuation to $s \in \IC$. 

We assume that $\chi_N$ is, indeed, primitive and satisfies $\chi_N(-1)= (-1)^\delta$. Recall that the meromorphic continuations of the Dirichlet $L$-function $L(s, \chi_{N})$ and the Riemann zeta function $\zeta(s)$ can be written explicitly in the form of the completed $L$-functions
\be \label{eq: completed}
\Lambda(s,\chi_N) = \frac{N^{\frac{s}{2}}}{\pi^{\frac{s+\delta}{2}}} \Gamma\left( \frac{s+\delta}{2} \right) L(s, \chi_{N}) \, , \quad \Lambda_\zeta(s) = \frac{1}{\pi^{\frac{s}{2}}} \Gamma\left( \frac{s}{2} \right) \zeta(s) \, , \quad s \in \IC \, ,
\ee
which obey the well-known functional equations
\be \label{eq: funct-eqs}
\Lambda(s,\chi_N) =\frac{\mathscr{G}(\chi_N)}{\ri^\delta \sqrt{N}} \Lambda(1-s,\overline{\chi_N}) \, , \quad \Lambda_\zeta(s) = \Lambda_\zeta(1-s) \, .
\ee
The meromorphic continuation of the $L$-functions $L$ and $L'$ in Eq.~\eqref{eq:LL'_N} is then described as follows. 

\begin{lemma}\label{lem:func-eq}
Let $\chi_N$ be a primitive Dirichlet character of modulus $N$ with $\chi_N(-1)=(-1)^\delta$. The L-functions $L(s)$ and $L'(s)$ in Eq.~\eqref{eq:LL'_N} obey the functional equation
    \be\label{eq:functional_eq}
    \Lambda(s)=\frac{\ri^\delta \sqrt{N}}{\mathscr{G}(\chi_N)}\Lambda'(-s)\,,
    \ee
    where $\Lambda(s):=\Lambda(s+1,\overline{\chi_N}) \Lambda_\zeta(s)$ and $\Lambda'(s):=\Lambda(s,\chi_N) \Lambda_\zeta(s+1)$ are the completed $L$-functions.
\end{lemma}
\begin{proof}
    The statement follows from the functional equations for $L(s,\chi_N)$ and $\zeta(s)$ in Eq.~\eqref{eq: funct-eqs} and using the properties
    \be
        \mathscr{G}(\overline{\chi_N})= \chi_N(-1) \overline{\mathscr{G}(\chi_N)} \, , \quad |\mathscr{G}(\chi_N)|= \sqrt{N} \, .
    \ee
\end{proof}

Once more following the blueprint of~\cite{Rella22, FR1phys, FR1maths}, we can now suitably combine the objects introduced here---namely, the generating functions $\frakf$ and $\frakg$ in Eq.~\eqref{eq:gen_func_fg} and the $L$-functions $L$ and $L'$ in Eq.~\eqref{eq:LL'_N}. In particular, we prove that the same two arithmetic functions
\be \label{eq:arithmetic-functs}
    \sum_{d|m}\frac{1}{d}\,\overline{\chi_N}(d) \,, \quad
        \sum_{d|m}\frac{d}{m}\,\chi_N(d) \,, 
\ee
appear both as coefficients of the generating series in Eq.~\eqref{eq:gen_func_fg} and as coefficients of the Dirichlet series expansion of the $L$-functions in Eq.~\eqref{eq:LL'_N}. In other words, the $L$-functions $L$, $L'$ can be obtained via the Mellin transform of the $q$-series $\frakf$, $\frakg$, respectively.

\begin{lemma}\label{lem:L0-gen} Let $\chi_N$ be a Dirichlet character of modulus $N$. The $L$-series $L(s)$ in Eq.~\eqref{eq:LL'_N} satisfies 
    \be\label{eq:L-0-1}
    L(s) =\sum_{m>0}\sum_{d|m} \frac{1}{d} \overline{\chi_N}(d)\,\frac{1}{m^s}\,,
    \ee
    where $d$ is a positive integer divisor of $m$. Note that the coefficients are generated by the weighted sum of $q$-Pochhammer symbols $\frakf$ in Eq.~\eqref{eq:fg-intro}.
\end{lemma}
\begin{proof}
Following~\cite{Rella22}, let us introduce $F_\alpha(n)=n^\alpha$ for $\alpha\in\IR$ and employ the compatibility of the multiplication
of Dirichlet series with the Dirichlet convolution of arithmetic functions. Explicitly,
\be
   L(s)= L(s+1,\overline{\chi_N})\zeta(s) = \sum_{m>0}\frac{\big(F_{-1}\overline{\chi_N}\ast F_0\big)(m)}{m^{s}}\,,
\ee
where $\ast$ denotes the Dirichlet convolution. By definition, 
\be
\begin{aligned}
    \big(F_{-1}\overline{\chi_N}\ast F_0\big)(m)&=\sum_{d\vert m}\frac{1}{d}\overline{\chi_N}(d)\,,
\end{aligned}
\ee
and the claim follows. 
\end{proof}
\begin{lemma} \label{lem:Linf-gen}
Let $\chi_N$ be a Dirichlet character of modulus $N$. The $L$-series $L'(s)$ in Eq.~\eqref{eq:LL'_N} satisfies 
\be\label{eq:L-inf-1}
    L'(s) =\sum_{m>0}\sum_{d|m}\frac{d}{m}\,\chi_N(d) \, \frac{1}{m^s}\,.
    \ee
     where $d$ is a positive integer divisor of $m$. Note that the coefficients are generated by the weighted sum of $q$-Pochhammer symbols $\frakg$ in Eq.~\eqref{eq:fg-intro}.
\end{lemma}
\begin{proof}
Again, we generalize the same argument of~\cite{Rella22} used in Lemma~\ref{lem:L0-gen}. In particular, we introduce the function $F_\alpha(n)=n^\alpha$ for $\alpha\in\IR$ and write
\be
L'(s)=L(s,\chi_N)\zeta(s+1)=\sum_{m>0}\frac{\big(F_{0}\chi_N\ast F_{-1}\big)(m)}{m^{s}}=\sum_{m>0}\sum_{d\vert m}\frac{d}{m}\chi_N(d)\frac{1}{m^{s}}\,,
\ee
where $\ast$ denotes the Dirichlet convolution.
\end{proof}

\subsubsection{Modular resurgent structures}

\begin{theorem}\label{thm:resurgence-f0}
    Let $\chi_N$ be a primitive Dirichlet character of modulus $N$. Let $\tfrakf(y)$ be the asymptotic expansion of the function $\frakf(y)$ in Eq.~\eqref{eq:fg-intro} in the limit $y \to 0$ with $\Im(y)>0$. Then, $\tfrakf$ is modular resurgent if and only if $\chi_N$ is odd.
\end{theorem}

\begin{proof}
The asymptotic expansion $\tfrakf$ in Eq.~\eqref{eq:formal-f0} is governed by a linear combination of the formal power series $\psi_k$ in Eq.~\eqref{eq: tilde-psi}. Thus, by Corollary~\ref{cor:stokes-f}, the singularities of the Borel transform $\borel[\tfrakf]$ are simple poles located along a vertical tower at the points
\begin{equation}
      \eta_m =2\pi\ri\frac{m}{N}\,,\quad m\in\IZ_{\neq 0} \, , 
\end{equation}
and the corresponding Stokes constants are given by 
\be
      R_m =\sum_{k\in\IZ_N}\chi_N(k)\Bigg[\sum_{\substack{d\vert m\\ d\underset{N}{\equiv} k}}\frac{d}{m}-\sum_{\substack{d\vert m\\ d\underset{N}{\equiv} -k}}\frac{d}{m}\Bigg] =\sum_{d\vert m}\frac{d}{m} \big(\chi_N(d)-\chi_N(-d)\big)\,,
\ee
which are non-vanishing if and only if $\chi_N$ is odd. When this is the case, we obtain
\be \label{eq:Rm-final}
      R_m = 2 \sum_{d\vert m}\frac{d}{m} \chi_N(d) \, ,  
\ee
which by Lemma~\ref{lem:Linf-gen} are the coefficients of the $L$-series $2L'(s)$ defined in Eq.~\eqref{eq:LL'_N}. Furthermore, since $\chi_N$ is taken to be primitive, Lemma~\ref{lem:func-eq} implies that $2L'(s)$ is, in fact, an $L$-function and satisfies the functional equation in Eq.~\eqref{eq:functional_eq} with $\delta=1$.
\end{proof}

\begin{theorem}\label{thm:resurgence-f-inf}
   Let $\chi_N$ be a primitive Dirichlet character of modulus $N$. Let $\tfrakg(y)$ be the asymptotic expansion of the function $\frakg(y)$ in Eq.~\eqref{eq:fg-intro} in the limit $y \to 0$ with $\Im(y)>0$. Then, $\tfrakg$ is modular resurgent if and only if $\chi_N$ is odd.
\end{theorem}

\begin{proof}
The asymptotic expansion $\tfrakg$ in Eq.~\eqref{eq:formal-f-inf} is governed by a linear combination of the formal power series $\varphi_k$ in Eq.~\eqref{eq: tilde-phi}. Thus, by Corollary~\ref{cor:stokes-g}, the singularities of the Borel transform $\borel[\tfrakg]$ are simple poles at
\begin{equation}
      \rho_m =2\pi\ri m\,,\quad m\in\IZ_{\neq 0} \, ,
\end{equation}
and the corresponding Stokes constants are given by 
\begin{equation}\label{eq:Sm-final}
\begin{aligned}
      S_m &=- 2 \ri\sum_{k\in\IZ_N}\chi_N(k)\sum_{d\vert m}\frac{1}{d} \sin\left(\frac{2\pi k d}{N}\right) \\
      &= -\sum_{k\in\IZ_N}\chi_N(k)\sum_{d\vert m}\frac{1}{d} \Big(\zeta_N^{kd}-\zeta_N^{-kd}\Big) \\
      &=-(1-\chi_N(-1))\,\mathscr{G}(\chi_N)\sum_{d\vert m}\frac{1}{d} \overline{\chi_N}(d) \, , 
\end{aligned}
\end{equation}
 where in the second step we used Eq.~\eqref{eq:step-char}. Note that $S_m$ is  non-vanishing if and only if $\chi_N$ is odd. 

Then, by Lemma~\ref{lem:L0-gen}, the Stokes constants $S_m$ are the coefficients of the $L$-series $-2\mathscr{G}(\chi_N) L(s)$ defined in Eq.~\eqref{eq:LL'_N}. Moreover, since $\chi_N$ is primitive, Lemma~\ref{lem:func-eq} implies that $-2\mathscr{G}(\chi_N) L(s)$ is, in fact, an $L$-function and satisfies the functional equation in Eq.~\eqref{eq:functional_eq} with~$\delta=1$.
\end{proof}

Combining Eq.~\eqref{eq:gen_func_fg} with Eqs.~\eqref{eq:Rm-final} and~\eqref{eq:Sm-final}, each of the two functions $\frakf(y)$ and $\frakg(y)$ defined in Eq.~\eqref{eq:fg-intro} can be written as the generating function of the Stokes constants of the other. Namely, 
\be \label{eq:gen_func_fg-final}
    \frakf(y) =  \frac{1}{2} \sum_{m=1}^\infty S_m q^m\,, \quad
        \frakg(y/N) = -\frac{1}{2} \sum_{m=1}^\infty R_m q^{m/N} \,, 
\ee
where $x=q=\re^{2 \pi \ri y}$, as before, while $S_m$ and $R_m$ are the Stokes constants of the asymptotic series $\tfrakg(y)$ and $\tfrakf(y)$, respectively. Equivalently, we can write
\be \label{eq:gen_func_fg-final2}
    \frakf(y) =  \frac{1}{2}\mathrm{disc}_{\frac{\pi}{2}} \tfrakg \Big( -\tfrac{1}{y} \Big) \,, \quad
        \frakg(y/N) = - \frac{1}{2}\mathrm{disc}_{\frac{\pi}{2}} \tfrakf \Big( -\tfrac{1}{y} \Big) \,.
\ee
This observation is a general property of \emph{paired} MRSs. In fact, $\tfrakg(y)$ and $\tfrakf(y)$ are naturally embedded in the modular resurgence paradigm of~\cite{FR1maths}, that is, they satisfy the commutative diagram in Eq.~\eqref{diag:resurgence-L funct}.

\begin{cor}\label{thm:mod-res}
 Let $\chi_N$ be a primitive Dirichlet character modulo $N$. The functions $\frakf$ and $\frakg$ defined in Eq.~\eqref{eq:fg-intro} satisfy the modular resurgence paradigm if and only if $\chi_N$ is odd.
\end{cor}

\begin{proof}
The statement follows combining Lemmas~\ref{lem:func-eq},~\ref{lem:L0-gen}, and~\ref{lem:Linf-gen} with Theorems~\ref{thm:resurgence-f0} and~\ref{thm:resurgence-f-inf} and Eq.~\eqref{eq:gen_func_fg-final}.
\end{proof}

\section{An application to local weighted projective planes}\label{sec:weighted}

Let $\IP(1,m,n)$ be the projective surface of weights $(1,m,n)$ with $m,n\in\IZ_{>0}$, that is, 
\be
\IP(1,m,n)=\big(\IC^3\big)^*/\sim \, ,
\ee
where $(a,b,c)\sim (a',b',c')$ if there exists $\lambda\in\IC^*$ such that $a=\lambda a'$, $b=\lambda^m b'$, and $c=\lambda^n c'$. The local weighted projective plane known as local $\IP^{m,n}$ is the total space of the canonical line bundle over $\IP(1,m,n)$. Let us introduce the positive integer
\be
N=1+m+n \, .
\ee

In this section, we perform a detailed resurgent analysis of the spectral trace of the quantum-mechanical operator $\rho_{m,n}$, introduced in Eq.~\eqref{eq: rho-mn} and canonically associated with local $\IP^{m,n}$ by the procedure proposed in~\cite{GHM}. This operator, defined as the inverse of the three-term operator $\mO_{m,n}$ in Eq.~\eqref{eq:quantum-ops}, was shown to exist and to be positive-definite and trace-class on $L^2(\IR)$ by Kashaev and Mari\~no in~\cite{KM}. In the same work, the authors explicitly determined the integral kernel of $\rho_{m,n}$ in terms of $q$-Pochhammer symbols, which in turn allowed them to derive a closed-form expression for the corresponding spectral trace $\mathrm{Tr}(\rho_{m,n})$ as a function of the quantum deformation parameter 
\be
\hbar = \frac{2 \pi y}{N} \in \IR_{>0} \, .
\ee 
Specifically,  
\be \label{eq:trace}
\mathrm{Tr}(\rho_{m,n})=\frac{\re^{\frac{\pi\ri}{12 y}-\pi\ri y \left(\frac{m}{N^2} -\frac{1}{12}\right) +\frac{\pi\ri}{4}}}{2 \sqrt{y} \sin(\pi n/N)} \frac{(q^{1-m/N};q)_\infty \, (q^{1-1/N};q)_\infty \, (\re^{2\pi\ri n/N};\tilde{q})_\infty}{(\re^{-2\pi\ri m/N};\tilde{q})_\infty \, (\re^{-2\pi\ri/N};\tilde{q})_\infty \, (q^{n/N};q)_\infty}\,,
\ee
where $q=\re^{2\pi\ri y}=\re^{\ri N\hbar}$ and $\tilde{q}=\re^{-2\pi\ri/y}=\re^{-\frac{4\pi^2\ri}{N \hbar}}$. 
Note that the factorization in Eq.~\eqref{eq:trace} is not
symmetric in the exchange of $q$ and $\tilde{q}$, but it is under $m \leftrightarrow n$ (albeit, not manifestly)~\cite{MZ}.

It follows from the results of~\cite{KM} that the spectral trace $\mathrm{Tr}(\rho_{m,n})$ admits a well-defined analytic continuation to $\hbar \in \IC'$. In the rest of this section, we will consider the analytically continued spectral trace\footnote{The fermionic spectral traces of the quantum operator $\rho_{m,n}$ are the coefficients in the orbifold expansion of its Fredholm determinant~\cite{GHM}. As such, they do not retain a dependence on the moduli space of the CY, and their resurgent analysis is therefore non-parametric. Computations further simplify, leading to exact results, by considering the logarithm of the (first fermionic) spectral trace. Note, however, that the data obtained in different regimes in $\hbar$ appears to probe different chambers in moduli space, through a mechanism implicit in the TS/ST correspondence~\cite{ABK, GuM, Rella22}.} in order to perform a resurgent analysis. 

\subsection{Asymptotics at all orders}

Let us consider the formula for ${\rm Tr}(\rho_{m,n})$ in Eq.~\eqref{eq:trace-log} and derive its all-orders perturbative expansion in the weak-coupling limit $y \to 0$ with $\Im(y)>0$. 
To do so, we determine the asymptotic expansion of each contributing block. In particular, 
\be \label{eq:exp0-b1}
\frac{\re^{\frac{\pi\ri}{12 y}-\pi\ri y \left(\frac{m}{N^2} -\frac{1}{12}\right) +\frac{\pi\ri}{4}}}{2 \sqrt{y} \, \sin(\pi n/N)} \frac{(\re^{2\pi\ri n/N};\tilde{q})_\infty}{(\re^{-2\pi\ri m/N};\tilde{q})_\infty \, (\re^{-2\pi\ri/N};\tilde{q})_\infty } \sim \frac{\re^{\frac{\pi\ri}{12 y}-\pi\ri y \left(\frac{m}{N^2} -\frac{1}{12}\right) +\frac{\pi\ri}{4}}}{4 \ri \sqrt{y} \, \sin(\pi m/N) \sin(\pi/N) }
\,,
\ee
while Eqs.~\eqref{eq: expansion-gkN} and~\eqref{eq: tilde-phi} yield
\be \label{eq:exp0-b2}
\frac{(q^{1-m/N};q)_\infty \, (q^{1-1/N};q)_\infty}{ (q^{n/N};q)_\infty} 
\sim \frac{\re^{- \frac{\pi \ri }{12 y}+\pi \ri y\left(\frac{m}{N^2}-\frac{1}{12} \right) + \frac{\pi \ri}{4}}}{\sqrt{y} \, \Gamma(\underline{m}/N) \Gamma(\underline{1}/N)/ \Gamma(n/N)} \re^{- \varphi_{\underline{1}}(y)-\varphi_{\underline{m}}(y)+\varphi_n(y)} \, .
\ee
Putting together Eqs.~\eqref{eq:exp0-b1} and~\eqref{eq:exp0-b2} and using the property $\Gamma(1-z)\Gamma(z)=\frac{\pi}{\sin(\pi z)}$, $z \notin \IZ$, we obtain the all-orders semiclassical expansion of the spectral trace of local $\IP^{m,n}$ in the form
\be \label{eq: trace-lim-0}
\mathrm{Tr}(\rho_{m,n})\sim_{y \to 0} \frac{\Gamma(n/N)\Gamma(m/N)\Gamma(1/N)}{4\pi^2 y}  \re^{- \varphi_{\underline{1}}(y)-\varphi_{\underline{m}}(y)+\varphi_n(y)} \,.
\ee
Note that the exponential terms of order $y$ and $y^{-1}$ in Eq.~\eqref{eq:exp0-b1} cancel with opposite
contributions in Eq.~\eqref{eq:exp0-b2} to give a global leading-order behavior of the form $1/y$ independently of the choice of $m$ and $n$. 

Similarly, let us derive the full asymptotic expansion of $\mathrm{Tr}(\rho_{m,n})$ in the strong-coupling limit $-1/y \to 0$ with $\Im(y)>0$. 
In this regime, 
\be \label{eq:exp-inf-b1}
\frac{\re^{\frac{\pi\ri}{12 y}-\pi\ri y \left(\frac{m}{N^2} -\frac{1}{12}\right) +\frac{\pi\ri}{4}}}{2 \sqrt{y} \, \sin(\pi n/N)} 
\frac{(q^{1-m/N};q)_\infty \, (q^{1-1/N};q)_\infty}{ (q^{n/N};q)_\infty} 
\sim \frac{\re^{\frac{\pi\ri}{12 y}-\pi\ri y \left(\frac{m}{N^2} -\frac{1}{12}\right) +\frac{\pi\ri}{4}}}{2 \sqrt{y} \,\sin(\pi n/N) }
\,,
\ee
while Eqs.~\eqref{eq: expansion-fkN} and~\eqref{eq: tilde-psi} imply 
\be \label{eq:exp-inf-b2}
\frac{(\re^{2\pi\ri n/N};\tilde{q})_\infty}{(\re^{-2\pi\ri m/N};\tilde{q})_\infty \, (\re^{-2\pi\ri/N};\tilde{q})_\infty }
\sim \re^{-\frac{\CV}{2\pi \ri}y} \sqrt{\frac{\sin(\pi n/N)}{2 \ri \, \sin(\pi m/N) \sin(\pi/N)}} \re^{- \psi_{\underline{1}}\left(-\frac{1}{y}\right)-\psi_{\underline{m}}\left(-\frac{1}{y}\right)+\psi_n\left(-\frac{1}{y}\right)} \, ,
\ee
where we have introduced the constant
\be \label{eq: cv-def}
\CV= \left( \mathrm{Li}_2(\zeta_N^n)- \mathrm{Li}_2(\zeta_N^{-m})- \mathrm{Li}_2(\zeta_N^{-1})\right) \, .
\ee
Putting together Eqs.~\eqref{eq:exp-inf-b1} and~\eqref{eq:exp-inf-b2}, we obtain the all-orders strong-coupling expansion of the spectral trace of local $\IP^{m,n}$ in the form
\be \label{eq: trace-lim-inf}
\mathrm{Tr}(\rho_{m,n})\sim_{y \to \infty} \frac{\re^{\frac{\pi\ri}{12 y}+\pi\ri y \left(\frac{\CV}{2\pi^2}-\frac{m}{N^2} +\frac{1}{12}\right)}}{2 \sqrt{2 y \, \sin(\pi n/N) \sin(\pi m/N) \sin(\pi/N)}} \re^{- \psi_{\underline{1}}\left(-\frac{1}{y}\right)-\psi_{\underline{m}}\left(-\frac{1}{y}\right)+\psi_n\left(-\frac{1}{y}\right)}  \,.
\ee
Note that the exponential term of order $y^{-1}$ in the RHS of Eq.~\eqref{eq: trace-lim-inf} cancels with the opposite contribution from the asymptotic series $\psi_k\big(-\tfrac{1}{y}\big)$, where $k=\underline{1}, \underline{m}, n$, which give the leading-order term
\be
- \psi_{\underline{1}}\big(-\tfrac{1}{y}\big)-\psi_{\underline{m}}\big(-\tfrac{1}{y}\big)+\psi_n\big(-\tfrac{1}{y}\big) = - \frac{\pi \ri}{12 y} + \frac{\pi}{12 y}(\cot(\pi/N)+\cot(\pi m/N)+\cot(\pi n/N)) + O(y^{-3}) \, .
\ee
However, and differently from the semiclassical limit $y \to 0$, the exponential term of order $y$ in Eq.~\eqref{eq: trace-lim-inf} does not vanish. 

\begin{lemma}
The spectral trace of local $\IP^{m,n}$ has the strong-coupling leading-order behavior
\be \label{eq: lead-tr-D}
\mathrm{Tr}(\rho_{m,n})\sim_{y \to \infty} \exp \left(\frac{y}{\pi} D\left( \frac{1-\zeta_N^m}{1-\zeta_N^{-1}} \right)\right) \, , 
\ee
where $D(z)= \Im \big(\mathrm{Li}_2(z))+\arg (1-z) \log|z|$ is the Bloch--Wigner function and $\arg$ denotes the branch of the argument between $-\pi$ and $\pi$.
\end{lemma}
\begin{proof}
Applying the properties of the Bloch--Wigner function
\be
D(\bar{z}) = -D(z) \, , \quad D(1-z)=-D(z) \, ,
\ee
and the dilogarithm's inversion formula
\be
\mathrm{Li}_2(z^{-1})= -\mathrm{Li}_2(z)-\frac{\pi^2}{6}-\frac{1}{2}\log^2(-z) \, , 
\ee
we find that 
\be
D\Big( \tfrac{1-\zeta_N^m}{1-\zeta_N^{-1}} \Big) = -\frac{1}{2} \Im \big(\mathrm{Li}_2(\zeta_N^n)+\mathrm{Li}_2(\zeta_N^m)+\mathrm{Li}_2(\zeta_N)) = \pi^2 \ri \left(\frac{\CV}{2\pi^2} -\frac{m}{N^2} + \frac{1}{12}\right) \, ,
\ee
where $\CV$ is defined in Eq.~\eqref{eq: cv-def}. The desired statement then follows from Eq.~\eqref{eq: trace-lim-inf}.
\end{proof}

The formula in Eq.~\eqref{eq: lead-tr-D} agrees with the results obtained in~\cite{MZ} by considering the multi-cut matrix model integral for the fermionic spectral traces\footnote{The fermionic spectral traces of the quantum operator $\rho_{m,n}$ are related to the spectral traces in Eq.~\eqref{eq: all-traces} by simple combinatorics and admit a natural matrix model representation studied in detail in~\cite{MZ}. An analogous study in the case of local $\IP^1 \times \IP^1$ was performed in~\cite{KMZ}.} of $\rho_{m,n}$ in a 't Hooft coupling expansion and computing the classical potential around its minimum. 

We stress that the TS/ST correspondence implies a precise identification between the constant appearing in the exponent of the RHS of Eq.~\eqref{eq: lead-tr-D} and the value of the $A$-period of the CY at the conifold singularity in moduli space. This prediction, also referred to as conifold volume
conjecture in~\cite{GuM}, extends to all toric CY threefolds and has been tested in various examples of genus one and two~\cite{CGM2, CGuM, MZ, KMZ, Doran}.
As already observed in~\cite{Rella22} and here in Eq.~\eqref{eq: trace-lim-0}, there is no analogue of the conifold volume conjecture in the semiclassical regime $\hbar \to 0$.

\subsection{Resurgence, summability, and quantum modularity}\label{sec:resurgence-traces}

Let us write Eq.~\eqref{eq:trace} equivalently as
\be \label{eq:trace-log}
    \log \mathrm{Tr}(\rho_{m,n}) = \frac{\pi\ri}{12 y}-\pi\ri y \left(\frac{m}{N^2} -\frac{1}{12}\right) +\frac{\pi\ri}{4} -\log \left(2 \sqrt{y} \sin\left(\pi n/N\right) \right)+\mathcal{Z}_{1, m, n}(y)\,,
\ee
where we have introduced the sum of $q$-Pochhammer symbols
\be \label{eq:trace-qPochh}
\begin{aligned}
    \mathcal{Z}_{1, m, n}(y)= &g_{\underline{1},N}\Big(\tfrac{y}{N}\Big)+g_{\underline{m},N}\Big(\tfrac{y}{N}\Big)-g_{n,N}\Big(\tfrac{y}{N}\Big)\\
     &-f_{\underline{1},N}\Big( -\tfrac{1}{y} \Big)-f_{\underline{m},N}\Big( -\tfrac{1}{y} \Big)+f_{n,N}\Big( -\tfrac{1}{y} \Big)\,.
\end{aligned}
\ee
We adopt the notation $\underline{k}=N-k$ for every $k\in\IZ_N$, as before. 
In the following, we will apply the results of Section~\ref{sec: single} to determine the resurgent and arithmetic properties of $\mathcal{Z}_{1, m, n}$ in the dual asymptotic limits $\hbar \to 0$ and $\hbar \to \infty$ along the lines of~\cite{Rella22, FR1phys}.

\subsubsection{The limit \texorpdfstring{$\hbar \to 0$}{hbar to 0}}

Recall that $N\hbar=2\pi y \in \IC\setminus \IR_{\le 0}$ and let $\varphi_{m,n}(\hbar)$ be the formal power series that governs the asymptotic expansion of $\mathcal{Z}_{1, m, n}(y)$ as $y \propto \hbar \to 0$ with $\Im(y)>0$. 
Observe that
\be \label{eq: phi-mn}
\ba
\varphi_{m,n}(\hbar) &= -\varphi_{\underline{1}}\Big(\tfrac{N \hbar}{2 \pi}\Big) - \varphi_{\underline{m}}\Big(\tfrac{N \hbar}{2 \pi}\Big) + \varphi_{n}\Big(\tfrac{N \hbar}{2 \pi}\Big) \, ,
\ea
\ee
where $\varphi_k$, $k \in \IZ_N$, is defined in Eq.~\eqref{eq: tilde-phi} and satisfies $\varphi_{\underline{k}}= -\varphi_k$.

\begin{prop}\label{prop:stokes-P1nm-s}
The asymptotic series $\varphi_{m,n}\in\IC\llbracket \hbar\rrbracket$ is Gevrey-1 and simple resurgent. Moreover, the singularities of its Borel transform are simple poles located at the points
\begin{equation}
 \xi_\ell=\frac{4\pi^2\ri}{N}\ell\,, \quad \ell\in\IZ_{\neq 0}\,,
\end{equation}
while the corresponding Stokes constants are
\be \label{eq: stokes-phimn}
S_\ell^{m,n}= 2\ri \sum_{k\in\{1,m,n\}} \Bigg[\sum_{d | \ell} \frac{1}{d} \sin\Big(\frac{2\pi k d}{N}\Big) \Bigg] \, , 
\ee
where the sum runs over the positive integer divisors of $\ell$. Besides, $S_\ell^{m,n}=S^{m,n}_{-\ell} \in \ri \IR$. 
\end{prop}

\begin{proof}
In the limit $\hbar \to 0$, the perturbative series $\varphi_{m,n}(\hbar)$ is given by the sums of the contributions of the $q$-Pochhammer symbols $g_{\underline{1},N}$, $g_{\underline{m},N}$, and $g_{n,N}$ in Eq.~\eqref{eq:trace-qPochh}. The statement then follows from Proposition~\ref{prop:Borel-g_kN} and Corollary~\ref{cor:stokes-g}.
Specifically, we have that
\begin{equation}
    S_\ell^{m,n}= -S_\ell^{\underline{1}}-S_\ell^{\underline{m}}+S_\ell^{n} \,,
\end{equation}
where $S_\ell^k$, $k \in \IZ_N$, is defined in Eq.~\eqref{eq:Stokes-S-kN} and satisfies $S_\ell^{\underline{k}}= -S_\ell^k$.
\end{proof}
 
For $N=3,4$, the Stokes constants $S_\ell^{1,1}$ and $S_\ell^{1,2}=S_\ell^{2,1}$ reduce, up to a numerical prefactor, to the ones in Eqs.~\eqref{eq:sn-3} and~\eqref{eq:sn-4}, respectively. These are the only cases in which the divisor sum functions in the RHS of Eq.~\eqref{eq: stokes-phimn} can be simply expressed in terms of Dirichlet characters. 
In addition, as the Borel plane displays one infinite tower of simple poles repeated at all non-zero integer multiples of a constant, the perturbative coefficients of $\varphi_{m,n}(\hbar)$ are given by the Dirichlet series of the Stokes constants evaluated at the integers.\footnote{See~\cite[Proposition~3.2]{FR1maths} for the general statement and proof.} Explicitly, the coefficient of $\hbar^{2j}$ in the formal power series $\varphi_{m,n}(\hbar)$ is
\begin{equation} \label{eq: dir-Smn}
 \frac{\Gamma(2j)}{\pi\ri}\sum_{\ell=1}^\infty \frac{S_\ell^{m,n}}{\xi_\ell^{2j}}\,,\quad  j \in \IZ_{>0} \,.
\end{equation}
Note that only the even powers of $\hbar$ appear in the asymptotic series as a consequence of Eq.~\eqref{eq: tilde-phi}.
 
\begin{lemma}\label{lem:symm-phi}
The discontinuity of $\varphi_{m,n}(\hbar)$ across the Stokes ray at angle $\pi/2$ is 
 \be\label{eq:disc-phi-nm}
 \begin{aligned}
     \mathrm{disc}_{\frac{\pi}{2}}\varphi_{m,n}(\hbar)= & f_{\underline{1},N}\Big( -\tfrac{2 \pi}{N \hbar} \Big)-f_{1,N}\Big( -\tfrac{2 \pi}{N \hbar} \Big)+f_{\underline{m},N}\Big( -\tfrac{2 \pi}{N \hbar} \Big)-f_{m,N}\Big( -\tfrac{2 \pi}{N \hbar} \Big)\\
     &+f_{\underline{n},N}\Big( -\tfrac{2 \pi}{N \hbar} \Big)-f_{n,N}\Big( -\tfrac{2 \pi}{N \hbar} \Big)-\pi\ri\,. 
 \end{aligned}
\ee
\end{lemma}
\begin{proof}
The result follows from Eq.~\eqref{eq: disc-gkN}.
\end{proof}

Let us now define the generating series of the Stokes constants $\{ S_\ell^{m,n} \}$ as 
\begin{equation}\label{eq:f0}
 f_0^{m,n}(y):=\begin{cases}
     \sum_{\ell>0} S_\ell^{m,n} \re^{2\pi\ri \ell y} & \text{ if } \Im(y)>0 \\
     \\
     -\sum_{\ell<0} S_\ell^{m,n} \re^{2\pi\ri \ell y} & \text{ if } \Im(y)<0
 \end{cases}\,,
\end{equation}
which is a holomorphic function on $\IC\setminus\IR$ with the periodicity and parity properties
\be
f_0^{m,n}(y+1)= f_0^{m,n}(y) \, , \quad f_0^{m,n}(-y)=-f_0^{m,n}(y) \, .
\ee
It follows that $f_0^{m,n}(y)$ need only be specified in either the upper or lower half of the complex $y$-plane,
where it is known in closed form by means of the exact discontinuity formula in Corollary~\ref{cor: disc-Sn}, which implies that
 \begin{equation}\label{eq:disc-f0}
     f_0^{m,n}\Big(-\frac{2 \pi}{N \hbar}\Big)= \mathrm{disc}_{\frac{\pi}{2}}\varphi_{m,n}(\hbar)\,,\quad \hbar\in\IH\,.
 \end{equation}
 
\begin{theorem}\label{thm:disc-med-0}
    Let $\tilde{f}^{m,n}_0(y)$ be the asymptotic expansion in the limit $y\to 0$ with $\Im(y)>0$ of the generating function $f^{m,n}_0(y)$ in Eq.~\eqref{eq:f0}. Then, the median resummation at angle $\pi/2$ of $\tilde{f}^{m,n}_0(y)$ reconstructs the function $f_0^{m,n}(y)$, that is,
    \begin{equation}
        \CS^{\rm med}_{\frac{\pi}{2}}\tilde{f}^{m,n}_0(y)=f^{m,n}_0(y)\,,\quad y\in\IH\,.
    \end{equation}
\end{theorem}
\begin{proof}
From Eq.~\eqref{eq: median-fkN}, we deduce that 
\begin{equation}
    \CS^{\rm med}_{\frac{\pi}{2}}\big[\tilde{f}_{k,N}(y)-\tilde{f}_{\underline{k},N}(y)\big]=f_{k,N}(y)-f_{\underline{k},N}(y)\,,
\end{equation}
hence the desired result. An analogous computation holds for $\CS^{\rm med}_{-\frac{\pi}{2}}\tilde{f}^{m,n}_0(y)$.
\end{proof}

\subsubsection{The limit \texorpdfstring{$\hbar \to \infty$}{hbar to infinity}}

Define the variable $\tau=-1/y= -2 \pi /(N \hbar) \in \IC\setminus \IR_{\ge 0}$ and let $\psi_{m,n}(\tau)$ be the formal power series that governs the asymptotic expansion of $\mathcal{Z}_{1, m, n}(y)$ as $y \propto -1/\tau\to \infty$ with $\Im(y)>0$.
Observe that
\be \label{eq: psi-mn}
\ba
\psi_{m,n}(\tau) &= -\psi_{\underline{1}}(\tau) - \psi_{\underline{m}}(\tau) + \psi_{n}(\tau) \, ,
\ea
\ee
where $\psi_k$, $k \in \IZ_N$, is defined in Eq.~\eqref{eq: tilde-psi} and satisfies $\psi_{\underline{k}}= -1-\psi_k$. 

\begin{prop}
The asymptotic series $\psi_{m,n}\in\IC\llbracket \tau\rrbracket$ is Gevrey-1 and simple resurgent. Moreover, the singularities of its Borel transform are simple poles located at the points
\begin{equation}
 \eta_\ell=\frac{2 \pi \ri }{N}\ell\,, \quad \ell\in\IZ_{\neq 0}\,,
\end{equation}
while the corresponding Stokes constants are
\be \label{eq: stokes-psimn}
 R_\ell^{m,n}= \sum_{k\in\{1,m,n\}}\Bigg[\sum_{\substack{d|\ell \\ d\underset{N}{\equiv} k}}\frac{d}{\ell}-\sum_{\substack{d|\ell \\ d\underset{N}{\equiv} -k}}\frac{d}{\ell} \Bigg] \,,
\ee
where the sum runs over the positive integer divisors of $\ell$. Besides, $R_\ell^{m,n}=-R^{m,n}_{-\ell} \in \IQ$, while $\ell R_\ell^{m,n} \in \IZ$.
\end{prop}

\begin{proof}
In the limit $\tau \to 0$, the perturbative series $\psi_{m,n}(\tau)$ is given by the sums of the contributions of the $q$-Pochhammer symbols $f_{\underline{1},N}$, $f_{\underline{m},N}$, and $f_{n,N}$ in Eq.~\eqref{eq:trace-qPochh}. The statement then follows from Proposition~\ref{prop:Borel-f_kN} and Corollary~\ref{cor:stokes-f}. In particular, we have that
\begin{equation}
  R_\ell^{m,n}= -R_\ell^{\underline{1}}-R_\ell^{\underline{m}}+R_\ell^{n} \,,
\end{equation}
where $R_\ell^k$, $k \in \IZ_N$, is defined in Eq.~\eqref{eq:Stokes-R-kN} and satisfies $R_\ell^{\underline{k}}= -R_\ell^k$.
\end{proof}

\begin{rmk}
Note that the hidden symmetry of the spectral trace $\mathrm{Tr}(\rho_{m,n})$ under the exchange of $m$ and $n$ emerges in the perturbative series 
\be
\varphi_{m,n}(\hbar) = \varphi_{n,m}(\hbar) \, , \quad \psi_{m,n}(\tau) = \psi_{n,m}(\tau) \, ,
\ee
and it is, accordingly, preserved at the level of the Stokes constants
\be
S_\ell^{m,n} = S_\ell^{n,m} \, , \quad R_\ell^{m,n}=R_\ell^{n,m} \, .
\ee
\end{rmk}
 
As observed for the Stokes constants in the limit $\hbar\to 0$, taking $N=3,4$, the Stokes constants $R_\ell^{1,1}$ and $R_\ell^{1,2}=R_\ell^{2,1}$ reduce, up to a numerical prefactor, to the ones in Eqs.~\eqref{eq:rn-3} and~\eqref{eq:rn-4}, respectively. Again, these are the only cases in which the divisor sum functions in the RHS of Eq.~\eqref{eq: stokes-psimn} can be simply expressed in terms of Dirichlet characters. 
In addition, as the Borel plane displays one infinite tower of simple poles repeated at all non-zero integer multiples of a constant, the perturbative coefficients of $\psi_{m,n}(\tau)$ are given by the Dirichlet series of the Stokes constants evaluated at the integers.\footnote{See~\cite[Proposition~3.2]{FR1maths} for the general statement and proof.} Explicitly, the coefficient of $\tau^{2j-1}$ in the formal power series $\psi_{m,n}(\tau)$ is
\begin{equation} \label{eq: dir-Rmn}
 \frac{\Gamma(2j-1)}{\pi\ri}\sum_{\ell=1}^\infty \frac{R_\ell^{m,n}}{\eta_\ell^{2j-1}}\,,\quad  j \in \IZ_{>0} \,.
\end{equation}
Note that only the odd powers of $\tau$ appear in the asymptotic series as a consequence of Eq.~\eqref{eq: tilde-psi}.

\begin{rmk} \label{rmk: notL-funct-mn}
For fixed $N \in \IZ_{\ge 2}$ and $m,n \in \IZ_N$ such that $N=1+m+n$, the Dirichlet series obtained from the sequences of Stokes constants $\{S_\ell^{m,n}\}$ and $\{R_\ell^{m,n}\}$, $\ell \in \IZ_{>0}$, appearing in Eqs.~\eqref{eq: dir-Smn} and~\eqref{eq: dir-Rmn}, are not necessarily $L$-series. More precisely, for $N>4$, the arithmetic functions in Eqs.~\eqref{eq: stokes-phimn} and~\eqref{eq: stokes-psimn} are not multiplicative. As a consequence, the formal power series $\varphi_{m,n}$ and $\psi_{m,n}$ in Eqs.~\eqref{eq: phi-mn} and~\eqref{eq: psi-mn} are not modular resurgent, and there is no functional equation joining the Dirichlet series in the two regimes. 
Nonetheless, as we prove in Section~\ref{sec:strong-weak}, a generalization of the strong-weak resurgent symmetry of local $\IP^2$ remains at play.  
\end{rmk}

\begin{lemma}\label{lem:symm-psi}
    The discontinuity of $\psi_{m,n}(\tau)$ across the Stokes ray at angle $\pi/2$ is 
     \be\label{eq:disc-psi-nm}
     \begin{aligned}
       \mathrm{disc}_{\frac{\pi}{2}}\psi_{m,n}(\tau)= &-g_{1,N}\Big( -\tfrac{1}{N \tau} \Big)+g_{\underline{1},N}\Big( -\tfrac{1}{N \tau} \Big)-g_{m,N}\Big( -\tfrac{1}{N \tau} \Big)+g_{\underline{m},N}\Big( -\tfrac{1}{N \tau} \Big)\\
       &-g_{n,N}\Big( -\tfrac{1}{N \tau} \Big)+g_{\underline{n},N}\Big( -\tfrac{1}{N \tau} \Big)\,.   
     \end{aligned}
    \ee
\end{lemma}
\begin{proof}
The result follows from Eq.~\eqref{eq: disc-fkN}.
\end{proof}

As before, let us now define the generating series of the Stokes constants $\{ R_\ell^{m,n} \}$ as 
 \begin{equation}\label{eq:f-inf}
     f_\infty^{m,n}(y):=\begin{cases}
         \sum_{\ell>0} R_\ell^{m,n} \re^{2\pi\ri \ell y} & \text{ if } \Im(y)>0 \\
         \\
         -\sum_{\ell<0} R_\ell^{m,n} \re^{2\pi\ri \ell y} & \text{ if } \Im(y)<0
     \end{cases}\,,
 \end{equation}
 which is a holomorphic function on $\IC\setminus\IR$ with the periodicity and parity properties
\be
f_\infty^{m,n}(y+1)= f_\infty^{m,n}(y) \, , \quad f_\infty^{m,n}(-y)=f_\infty^{m,n}(y) \, .
\ee
It follows that $f_\infty^{m,n}(y)$ need only be specified in either the upper or lower half of the complex $y$-plane,
where it is known in closed form by means of the exact discontinuity formula in Corollary~\ref{cor: disc-Rn}, which implies that
 \begin{equation}\label{eq:disc-finf}
     f_\infty^{m,n}\Big(-\frac{1}{N \tau}\Big)= \mathrm{disc}_{\frac{\pi}{2}}\psi_{m,n}(\tau)\,,\quad \tau\in\IH\,.
 \end{equation}

\begin{theorem}\label{thm:disc-med-inf}
   Assume that Conjecture~\ref{conj:summability-g-kN} holds. Let $\tilde{f}_\infty^{m,n}(y)$ be the asymptotic expansion in the limit $y\to 0$ with $\Im(y)>0$ of the generating function $f_\infty^{m,n}(y)$ in Eq.~\eqref{eq:f-inf}. Then, the median resummation at angle $\pi/2$ of $\tilde{f}_\infty^{m,n}(y)$ reconstructs the function $f_\infty^{m,n}(y)$, that is,
    \begin{equation}
        \CS^{\rm med}_{\frac{\pi}{2}}\tilde{f}_\infty^{m,n}(y)=f_\infty^{m,n}(y)\,,\quad y\in\IH\,.
    \end{equation}  
\end{theorem}
\begin{proof}
From Eq.~\eqref{eq: median-gkN}, we deduce that 
\begin{equation}
    \CS^{\rm med}_{\frac{\pi}{2}}\big[\tilde{g}_{k,N}(y)-\tilde{g}_{\underline{k},N}(y)\big]=g_{k,N}(y)-g_{\underline{k},N}(y)\,,
\end{equation}
hence the desired result. An analogous computation holds for $\CS^{\rm med}_{-\frac{\pi}{2}}\tilde{f}^{m,n}_\infty(y)$. 
\end{proof}

\begin{rmk}
Recall that the function $\mathcal{Z}_{1, m, n}$, introduced in Eq.~\eqref{eq:trace-qPochh}, is the sum of $q$-Pochhammer symbols appearing in the logarithm of the spectral trace $\mathrm{Tr}(\rho_{m,n})$. This extends to a well-defined, analytic function on $\IC'$, allowing us to perform a full resurgent analysis.
Although the generating functions $f_0^{m,n}$ and $f_\infty^{m,n}$ were originally derived from the asymptotic expansions of $\mathcal{Z}_{1, m, n}$ at weak and strong coupling, respectively---and can be interpreted as non-perturbative completions of one another---they are not manifestly present in Eq.~\eqref{eq:trace-qPochh}. In fact, they satisfy the relation
\be \label{eq: Zmn-sym}
\mathcal{Z}_{1, m, n}(y)-\mathcal{Z}_{\underline{1}, \underline{m}, \underline{n}}(y) = f_\infty^{m,n}\big(\tfrac{y}{N} \big)-f_0^{m,n}\big(-\tfrac{1}{y} \big) - \pi \ri \, , 
\ee 
where $\mathcal{Z}_{\underline{1}, \underline{m}, \underline{n}}$ denotes the image of $\mathcal{Z}_{1, m, n}$ under the transformation
\be
k \to \underline{k}=N-k \, , \quad k =1,m,n \, . 
\ee
Thus, resurgence naturally complements $\mathcal{Z}_{1, m, n}$ with the companion function $\mathcal{Z}_{\underline{1}, \underline{m}, \underline{n}}$, enforcing an hidden symmetry that does not appear in the spectral trace itself.
\end{rmk}

\subsubsection{Quantum modularity}

In Sections~\ref{sec: single} and~\ref{sec:weighted-sum}, we proved that the $q$-Pochhammer symbols $f_{k,N}$ and $g_{k,N}$ in Eq.~\eqref{eq:f_kN} and their weighted sums $\frakf$ and $\frakg$ in Eq.~\eqref{eq:fg-intro} possess specific quantum modularity properties. Here, we show that the same statement applies to the generating functions of the Stokes constants for $\log {\rm Tr}(\rho_{m,n})$ in both limits $\hbar \to 0$ and $\hbar \to \infty$.

\begin{cor}\label{cor:disc-qm} 
    The functions $f_0^{m,n}, f_\infty^{m,n}: \IH \to \IC$ defined in Eqs.~\eqref{eq:f0} and~\eqref{eq:f-inf} are holomorphic quantum modular functions for the group $\Gamma_{N}\subset\mathsf{SL}_2(\IZ)$ generated by the elements in Eq.~\eqref{eq: GammaN-generators}.
\end{cor}

\begin{proof}
For the generating function at weak coupling $f_0^{m,n}$, the statement follows from Eqs.~\eqref{eq:disc-phi-nm} and~\eqref{eq:disc-f0} and Theorem~\ref{thm:qm}. Similarly, for the generating function at strong coupling $f_\infty^{m,n}$, the result follows from Eqs.~\eqref{eq:disc-psi-nm} and~\eqref{eq:disc-finf} and Theorem~\ref{thm:qm}.
\end{proof}

Note that the above result reproduces Theorem~4.7 in~\cite{FR1phys} when taking $m=n=1$. 
Moreover, as for the case of the weighted sums $\frakf$ and $\frakg$, the non-trivial cocycles $h_{\gamma_N}[f_0^{m,n}]$ and $h_{\gamma_N}[f_\infty^{m,n}]$ are linear combinations of the logarithm of Faddeev's quantum dilogarithm $\Phi_\mb$ in Eq.~\eqref{eq: seriesPhib}. Namely, 
\begin{subequations}
    \begin{align}
        h_{\gamma_N}[f_0^{m,n}](\hbar) &= \sum_{k\in\{1,m,n\}} 
        h_{\gamma_N}[f_{k,N}-f_{\underline{k},N}](\hbar) = \sum_{k\in\{1,m,n\}} 
        h_{\gamma_N}[F_k-F_{\underline{k}}](\hbar) \, ,\\
        h_{\gamma_N}[f_\infty^{m,n}](\tau) &= \sum_{k\in\{1,m,n\}} 
        h_{\gamma_N}[g_{\underline{k},N}-g_{k,N}](\tau) = \sum_{k\in\{1,m,n\}} 
        h_{\gamma_N}[G_{\underline{k}}-G_k](\tau) \, ,
    \end{align}
\end{subequations}
where the functions $F_k$ and $G_k$ are written explicitly in terms of $\Phi_\mb$ in Eqs.~\eqref{eq: F-phib} and~\eqref{eq: G-phib}.

Following Section~\ref{sec:fricke}, we can define a second pair of holomorphic quantum modular functions for $\Gamma_N$ by acting on the generating functions $f_0^{m,n}(y)$ and $f_\infty^{m,n}(y)$ with $y \mapsto -\frac{1}{Ny}$. Precisely, as we have done for the $q$-Pochhammer symbols in Eq.~\eqref{eq: fricke-fs}, we introduce the holomorphic functions $f_{0,\star}^{m,n}, f_{\infty,\star}^{m,n}\colon \IH \to \IC$ as
\be \label{eq: fricke-trace}
f_{0,\star}^{m,n}(y):=f_0^{m,n}\big(-\tfrac{1}{Ny}\big) \, , \quad  f_{\infty,\star}^{m,n}(y):=f_\infty^{m,n}\big(-\tfrac{1}{Ny}\big) \, ,
\ee
where again $^{\star}$ denotes the action of the Fricke involution.

\begin{cor}
The functions $f_{0,\star}^{m,n}, f_{\infty,\star}^{m,n}\colon \IH\to\IC$, defined in Eq.~\eqref{eq: fricke-trace} as the images of the generating functions under Fricke involution, are holomorphic quantum modular functions for $\Gamma_N$.
\end{cor}

\begin{proof}
For $f_{0,\star}^{m,n}$, the statement follows from Eqs.~\eqref{eq:disc-phi-nm} and~\eqref{eq:disc-f0} and Theorem~\ref{thm:fricke-qPochh}. Similarly, for $f_{\infty,\star}^{m,n}$, the result follows from Eqs.~\eqref{eq:disc-psi-nm} and~\eqref{eq:disc-finf} and Theorem~\ref{thm:fricke-qPochh}.
\end{proof}

Analyzing the resurgent structures of the spectral trace of local $\IP^2$, we observed in~\cite{FR1phys} how the generating functions of the Stokes constants, once interpreted as holomorphic quantum modular forms, present a remarkable duality under the action of the Fricke involution of the corresponding modular group $\Gamma_1(3)$, exchanging the roles played by the Stokes constants at weak and strong coupling. Unifying the results of this section with the ones of Section~\ref{sec:fricke}, we have now generalized this duality to the whole class of local weighted projective planes.

\subsection{Strong-weak resurgent symmetry}\label{sec:strong-weak}

Following the observations of~\cite{Rella22}, it was proven in~\cite{FR1phys} that the exact resurgent structures of the spectral trace of local $\IP^2=\IP^{1,1}$ in the dual regimes of $\hbar \to 0$ and $\hbar \to \infty$ obey a full-fledged strong-weak resurgent symmetry, dictating the mechanism by which the non-perturbative corrections in one regime produce the perturbative expansion in the other and vice versa. A somewhat weaker formulation of the same symmetry holds for every local $\IP^{m,n}$.

\begin{theorem}\label{thm:strong-weak}
For all $m,n\in\IZ_{>0}$, the following diagram commutes.  
\begin{equation} \label{diag: strong-weak2}
\begin{tikzcd}[column sep = 2.8em, row sep=2.8em]
&\mathrm{disc}_{\frac{\pi}{2}}\psi_{m,n}(\tau) \arrow[red,sloped]{dd} \arrow[r, "\hbar \rightarrow 0"] & \varphi_{m,n}(\hbar)\arrow[red,sloped]{dd}{\text{strong-weak}}[swap]{\text{symmetry}} \arrow[dr,sloped,gray,"\text{resurgence}"] &
\\
\color{gray}\{\eta_\ell, \, R_\ell^{m,n}\}\arrow[sloped,gray]{ur}{\text{generating}}[swap]{\text{series}} & & &  \color{gray}\{\xi_\ell, \, S_\ell^{m,n}\} \arrow[sloped,gray]{dl}{\text{generating}}[swap]{\text{series}}
\\
 & \arrow[red,sloped]{uu}{\text{strong-weak}}[swap]{\text{symmetry}} \psi_{m,n}(\tau) \arrow[ul,sloped,gray,swap,"\text{resurgence}"]& \mathrm{disc}_{\frac{\pi}{2}}\varphi_{m,n}(\hbar) \arrow[red,sloped]{uu} \arrow[l,swap, "\tau \rightarrow 0"] &
\end{tikzcd}
\end{equation}
\end{theorem}

\begin{proof}
We adopt the notation introduced in Remark~\ref{rmk: not-L-functs} and write
\begin{subequations}
    \begin{align}
    \mathrm{disc}_{\frac{\pi}{2}}\varphi_{m,n}(\hbar)= & -\mathcal{F}_1\Big( -\tfrac{2 \pi}{N \hbar} \Big)- \mathcal{F}_m\Big( -\tfrac{2 \pi}{N \hbar} \Big)- \mathcal{F}_n\Big( -\tfrac{2 \pi}{N \hbar} \Big)\,, \\
    \mathrm{disc}_{\frac{\pi}{2}}\psi_{m,n}(\tau)= &\mathcal{G}_1\Big( -\tfrac{1}{N \tau}\Big)+ \mathcal{G}_m\Big( -\tfrac{1}{N \tau}\Big)+\mathcal{G}_n\Big( -\tfrac{1}{N \tau}\Big)\,,
    \end{align}
\end{subequations}
where we have removed the underscript $N$ for simplicity. Then, Eq.~\eqref{eq: disc-akbk} implies that
\begin{subequations}
    \begin{align}
    \mathrm{disc}_{\frac{\pi}{2}}[\tilde{\mathcal{F}}_1+ \tilde{\mathcal{F}}_m+ \tilde{\mathcal{F}}_n](\tau) =& 2[\mathcal{G}_1+\mathcal{G}_m+\mathcal{G}_n]\Big( -\tfrac{1}{N \tau}\Big)= 2 \, \mathrm{disc}_{\frac{\pi}{2}}\psi_{m,n}(\tau)\,, \\
    \mathrm{disc}_{\frac{\pi}{2}}[\tilde{\mathcal{G}}_1+ \tilde{\mathcal{G}}_m+\tilde{\mathcal{G}}_n]\Big(\tfrac{\hbar}{2 \pi}\Big) =& -2[\mathcal{F}_1+\mathcal{F}_m+\mathcal{F}_n]\Big( -\tfrac{2 \pi}{N \hbar} \Big) = 2\, \mathrm{disc}_{\frac{\pi}{2}}\varphi_{m,n}(\hbar) \,,
    \end{align}
\end{subequations}
that is, the discontinuity of the asymptotic expansion of one generating function reproduces the other generating function in the appropriate variables.
In order to complete the proof, we observe that each of the discontinuity arrows (from bottom left to top left and from top right to bottom right) can be inverted by considering the composition of two operations. First, we reconstruct the Dirichlet series from the corresponding discontinuity by Mellin transform. As in Lemma~3.1 in~\cite{FR1phys}, we find that
\be
\int_0^\infty t^{j-1} f_0^{m,n}(\ri t) \, dt = \frac{\Gamma(j)}{(2 \pi)^j}\sum_{\ell =1}^{\infty}\frac{S_\ell^{m,n}}{\ell^j}\,,\quad  j \in \IZ_{> 0} \,.
\ee
Second, we apply the exact large-order relation in Eq.~\eqref{eq: dir-Smn} to obtain the perturbative coefficients from the Dirichlet series evaluated at integer points. The strong coupling regime is treated entirely analogously.
\end{proof}

Consider the commutative diagram in Eq.~\eqref{diag: strong-weak2}. Compared to the full-fledged strong-weak symmetry in the diagram in Eq.~(3.30) of~\cite{FR1phys}, some ingredients are missing. 
As pointed out in Remark~\ref{rmk: notL-funct-mn}, for $N> 4$, the Stokes constants are not coefficients of $L$-functions because they do not generally satisfy an Euler product expansion. Therefore, there is no functional equation joining the two Dirichlet series, and similarly, there is no simple arithmetic twist mapping one sequence of Stokes constants into the other. However, if the sum of weights $m+n$ is either $2$ or $3$, then the corresponding geometries, \emph{i.e.}, local $\IP^{2}$, $\IP^{1,2}$, and $\IP^{2,1}$, fit the original, complete strong-weak resurgent symmetry of~\cite{FR1phys}. 
In these two special cases, the diagram in Eq.~\eqref{diag: strong-weak2} reproduces the results of Section~\ref{sec:modular-res-paradigm}.

\begin{cor} \label{cor: 3,4-final}
    For $m+n=2,3$, the asymptotic series $\varphi_{m,n}$ and $\psi_{m,n}$ defined in Eqs.~\eqref{eq: phi-mn} and~\eqref{eq: psi-mn} are a pair of modular resurgent series and together with the functions $f_0^{m,n}$ and $f_\infty^{m,n}$ in Eqs.~\eqref{eq:f0} and~\eqref{eq:f-inf} satisfy the modular resurgence paradigm.
\end{cor}
\begin{proof}
The generating functions of the Stokes constants can be directly identified with the weighted sums of $q$-Pochhammer symbols $\frakf$ and $\frakg$ in Eq.~\eqref{eq:fg-intro} for $N=3,4$. Precisely,
\begin{equation}\label{eq:disc-3,4}
    f_0^{m,n}(y)=-\tfrac{N}{m+n-1} \frakf(y)-\pi \ri \,,\quad  f_\infty^{m,n}(y)=-\tfrac{N}{m+n-1} \frakg(y)\,,
\end{equation}
where the Dirichlet characters in the definition of $\frakf$ and $\frakg$ are the unique non-principal characters modulo $N$. 
Therefore, Theorems~\ref{thm:resurgence-f0} and~\ref{thm:resurgence-f-inf} and Corollary~\ref{thm:mod-res} apply.
\end{proof}

\begin{rmk}
The TS/ST correspondence of~\cite{GHM, CGM2} provides an explicit integral formula for the (fermionic) spectral traces of a toric CY threefold in terms of the so-called total grand potential of the topological string, where the latter can be expressed as the sum of two contributions encoding the standard and Nekrasov--Shatashvili free energies. 
Under the identification $g_s=4\pi^2/\hbar$, where $g_s$ is the topological string coupling constant, the total grand potential becomes a function of $\hbar$. Crucially, only one of its two components survives in each of the two perturbative limits of $\hbar\to0$ and $\hbar \to \infty$, while the other dictates the corresponding non-perturbative corrections~\cite{HMMO}.
In fact, this strong-weak coupling duality relating $g_s$ and $\hbar$ lies at the core of the TS/ST correspondence itself and allows one to access the strongly-coupled regime of the standard topological string on the given geometry from the semiclassical regime of the spectral problem defined by its quantum mirror curve, and vice versa.
Moreover, a similar $S$-duality exchanging perturbative and non-perturbative terms in $\hbar$ and $\hbar^{-1}$ appears in the exact quantization conditions for the energy spectrum of cluster integrable systems~\cite{HM, FHM}. There, the duality is deeply related to the existence of a modular double of the integrable system in the sense of Faddeev~\cite{faddeev_modular}. 
Finally, the strong-weak resurgent symmetry of the spectral traces presented in~\cite{FR1phys} and here can be interpreted as an exact, mathematically proven manifestation of the $S$-type duality between the two contributions to the total grand potential of the topological string for the specific geometries at hand.
See also the discussion in~\cite[Sec.~3.2]{FR1phys}.
\end{rmk}

\subsubsection{Reconstructing the weighted sums}

When $N=3,4$, the strong-weak symmetry relating the asymptotic expansions of the logarithm of the spectral trace for $\hbar \to 0$ and $\hbar \to \infty$, acquires an additional layer of detail, as it is coincides with the modular resurgence paradigm for sums of $q$-Pochhammer symbols weighted by Dirichlet characters that we discussed in Section~\ref{sec:weighted-sum}.
We will now show that a similar result holds for $N>4$ when considering particular linear combinations of the generating functions $f_0^{m,n}$ and $f_\infty^{m,n}$. 

\begin{lemma}\label{ref:lem-combination-0}
Let $\chi_N$ be any odd Dirichlet character of modulus $N$. The function $\frakf$ in Eq.~\eqref{eq:fg-intro} is a linear combination of the functions $f_0^{m,n}$ in Eq.~\eqref{eq:f0} over $m,n \in \IZ_N$ with $1+m+n=N$.
\end{lemma}

\begin{proof}
We adopt the notation introduced in Remark~\ref{rmk: not-L-functs} and drop the underscript $N$ for simplicity. 
By construction, $\frakf$ is a polynomial in the variables $\CF_k$. More precisely, $\frakf\in\IC[\CF_1, \CF_2,\ldots, \CF_{d}]$ with $d=\lfloor{\frac{N-1}{2}}\rfloor$. Indeed, we have that
\begin{equation} \label{eq: frakf-comb}
\ba
        \frakf(y)&=\sum_{k=1}^{d}\chi_N(k)\, f_{k,N}(y) + \sum_{k=1}^{d }\chi_N(-k)\, f_{\underline{k},N}(y) \\&=\sum_{k=1}^{d}\chi_N(k)\, \CF_k(y)+ \sum_{k=1}^{d}\chi_N(k) \frac{\pi \ri}{N}(2k-N)\,,
\ea
\end{equation}
where we used $\chi_N(-1)=-1$.
Similarly, Eqs.~\eqref{eq:disc-phi-nm} and~\eqref{eq:disc-f0} imply that 
\be \label{eq: f0-F}
f_0^{m,n}=- \CF_1 -\CF_m -\CF_n\in\IZ[\CF_1, \CF_2,\ldots, \CF_{d}] \, ,
\ee
because $\CF_{N-k}=-\CF_k$. We are then left with a simple linear algebra computation. 
Consider the vectors
\begin{subequations}
\begin{align}
\vec{f_0} &= \begin{pmatrix}
   f_0^{1,N-2} &
    f_0^{2,N-3} &
    \ldots &
     f_0^{d,N-(d+1)}
\end{pmatrix} \, , \\
\vec{\CF} &= \begin{pmatrix}
   \CF_1 &
    \CF_2 &
    \ldots &
     \CF_d
\end{pmatrix} \, , \\
\vec{\chi} &= \begin{pmatrix}
           \chi_N(1) &
            \chi_N(2) &
            \ldots &
             \chi_N(d)
\end{pmatrix} \, , 
\end{align}
\end{subequations}
which satisfy
\be
\vec{f_0}= - A \, \vec{\CF}\,^\intercal \, , 
\ee
where $A$ is the the $d\times d$ matrix
\begin{equation} \label{eq: matrixA} 
   A= \begin{pmatrix}
       2 & -1 & 0 &  &\ldots   &0 \\
       1 & 1 & -1 & 0  & & 0     \\
       1 & 0 & 1 & -1 & \ddots&  \vdots \\
       \vdots & \vdots &  &\ddots  &\ddots  & 0 \\
       1 & 0 & \ldots &  0&  1 & -1 \\
       1 & 0 &  & \ldots &  0 & 1
    \end{pmatrix} \, .
\end{equation}
Note that $A$ is invertible as $\det A=d+1$.
Finally, we obtain
\begin{equation}
    \frakf=- \vec{\chi} \, A^{-1} \, \vec{f_0}^\intercal + \omega\, , 
\end{equation}
where $\omega$ is the constant term in the RHS of Eq.~\eqref{eq: frakf-comb}.
\end{proof}

\begin{lemma}\label{ref:lem-combination-inf}
Let $\chi_N$ be any odd Dirichlet character of modulus $N$. The function $\frakg$ in Eq.~\eqref{eq:fg-intro} is a linear combination of the functions $f_\infty^{m,n}$ in Eq.~\eqref{eq:f-inf} over $m,n \in \IZ_N$ with $1+m+n=N$.
\end{lemma}

\begin{proof}
As in the proof of Lemma~\ref{ref:lem-combination-0}, we adopt the notation introduced in Remark~\ref{rmk: not-L-functs} and drop the underscript $N$ for simplicity. 
By construction, $\frakg$ is a polynomial in the variables $\CG_k$. More precisely, $\frakg\in\IC[\CG_1, \CG_2,\ldots, \CG_{d}]$ with $d=\lfloor{\frac{N-1}{2}}\rfloor$. Indeed, we have that 
\begin{equation}
\ba
        \frakg(y)&=\sum_{k=1}^{d}\chi_N(k)\, g_{k,N}(y) + \sum_{k=1}^{d }\chi_N(-k)\, g_{\underline{k},N}(y) \\&=- \sum_{k=1}^{d}\chi_N(k)\, \CG_k(y)\,,
\ea
\end{equation}
where we used $\chi_N(-1)=-1$.
Similarly, Eqs.~\eqref{eq:disc-psi-nm} and~\eqref{eq:disc-finf} lead to
\be \label{eq: finf-F}
f_\infty^{m,n}= \CG_1+\CG_m+\CG_n \in\IZ[\CG_1, \CG_2,\ldots, \CG_{d}]\, ,
\ee
since $\CG_{N-k}=-\CG_k$. Again, we are left with a simple linear algebra computation.
Consider the vectors
\begin{subequations}
\begin{align}
\vec{f_\infty} &= \begin{pmatrix}
   f_\infty^{1,N-2} &
    f_\infty^{2,N-3} &
    \ldots &
     f_\infty^{d,N-(d+1)}
\end{pmatrix} \, , \\
\vec{\CG} &= \begin{pmatrix}
   \CG_1 &
    \CG_2 &
    \ldots &
     \CG_d
\end{pmatrix} \, , \\
\vec{\chi} &= \begin{pmatrix}
           \chi_N(1) &
            \chi_N(2) &
            \ldots &
             \chi_N(d)
\end{pmatrix} \, , 
\end{align}
\end{subequations}
which satisfy
\be
\vec{f_\infty}= A \, \vec{\CG}\,^\intercal \, , 
\ee
where $A$ is the same the $d\times d$ matrix in Eq.~\eqref{eq: matrixA}.
Therefore, we obtain
\begin{equation}
    \frakg=-\vec{\chi} \, A^{-1} \, \vec{f_\infty}^\intercal \, . 
\end{equation}
\end{proof}
In the proofs of Lemmas~\ref{ref:lem-combination-0} and~\ref{ref:lem-combination-inf}, $d=1$ if and only if $N=3$ or $N=4$, which correspond to $m+n=2$ (\emph{i.e.}, local $\IP^2$) or $m+n=3$ (\emph{i.e.}, local $\IP^{1,2}$ and $\IP^{2,1}$), as expected from the equalities in Eq.~\eqref{eq:disc-3,4}. For all other choices of $N$ such that there exists a primitive odd Dirichlet character of modulus $N$, a complete, full-fledged strong-weak resurgent symmetry applies to suitable linear combinations of the generating functions $f_0^{m,n}$ and $f_\infty^{m,n}$ with $m,n \in \IZ_N$ and $1+m+n=N$. We stress that the physical or geometric interpretation of such linear combinations is yet to be understood and will be the object of future investigation. Finally, observe that the combinatorial argument of Lemmas~\ref{ref:lem-combination-0} and~\ref{ref:lem-combination-inf} can be translated in terms of the spectral traces themselves via the relation in Eq.~\eqref{eq: Zmn-sym}. 

\section{Conclusions}

In this paper, building on the results of~\cite{FR1maths, FR1phys}, we study the resurgent and quantum modular properties of $q$-Pochhammer symbols, which are the building blocks of the spectral traces canonically associated to local weighted projective planes $\IP^{m,n}$ in the context of the TS/ST correspondence. Our purpose is twofold. On the one hand, we advance the study of modular resurgence by constructing a new, infinite family of MRSs from sums of $q$-Pochhammer symbols weighted by Dirichlet characters. On the other hand, we investigate the generalization of the strong-weak resurgent symmetry of local $\IP^2$ to all local $\IP^{m,n}$ for $m,n \in \IZ_{>0}$. Our results highlight new features of the interplay between quantum modularity and resurgence and their application to the spectral theory of local CY threefolds. We summarize our results in the following table.

\medskip 

\begin{center}
\begin{small}
\renewcommand{\arraystretch}{1.3}
 \begin{tabular}{c?c|c|c|c|c|}
    & {\normalsize QM for $\Gamma_{N}$ } &  {\normalsize $\CS^{\rm med}_\theta\tilde{f}=f$ } & {\normalsize MRS } & {\normalsize MRP } & {\normalsize SW } \\ 
\specialrule{.1em}{.1em}{.1em}
{\normalsize $f_{k,N}$ }  & \ding{51} & \ding{55} & \ding{55} & \ding{55} & n.a.\\ 
& {\small Thm.~\ref{thm:qm}} & & & & \\
  \hline  
{\normalsize $g_{k,N}$ }  & \ding{51}  & \ding{55} & \ding{55} & \ding{55} & n.a. \\  
& {\small Thm.~\ref{thm:qm}} & & & & \\
  \hline 
  \hline
{\normalsize $\frakf$} &  \ding{51}  &  iff $\chi_N$ odd & $\chi_N$ primitive, iff odd & $\chi_N$ primitive, iff odd & n.a.\\
& {\small Cor.~\ref{cor:qm-weighted}} & {\small Thm.~\ref{thm:summability-f0}} & {\small Thm.~\ref{thm:resurgence-f0}} & {\small Cor.~\ref{thm:mod-res}} &  \\
      \hline  
{\normalsize $\frakg$} &  \ding{51} & Conj. iff $\chi_N$ odd & $\chi_N$ primitive, iff odd & $\chi_N$ primitive, iff odd & n.a. \\
 & {\small Cor.~\ref{cor:qm-weighted}} & {\small Thm.~\ref{thm:summability-f-inf} (Conj.~\ref{conj:summability-g-kN})} &  {\small Thm.~\ref{thm:resurgence-f-inf}} &{\small Cor.~\ref{thm:mod-res}}  &  \\
\hline
\hline
{\normalsize $f_0^{m,n}$ }& \ding{51} & \ding{51} & $\IP^2, \IP^{1,2}, \IP^{2,1}$ & $\IP^2, \IP^{1,2}, \IP^{2,1}$ & \ding{51} \\
& {\small Cor.~\ref{cor:disc-qm}} & {\small Thm.~\ref{thm:disc-med-0}} & {\small Cor.~\ref{cor: 3,4-final}} & {\small Cor.~\ref{cor: 3,4-final}} & {\small Thm.~\ref{thm:strong-weak}} \\
\hline
{\normalsize $f_\infty^{m,n} $} & \ding{51} & Conj. & $\IP^2, \IP^{1,2}, \IP^{2,1}$ & $\IP^2, \IP^{1,2}, \IP^{2,1}$ & \ding{51}\\
 & {\small Cor.~\ref{cor:disc-qm}} & {\small Thm.~\ref{thm:disc-med-inf} (Conj.~\ref{conj:summability-g-kN})} & {\small Cor.~\ref{cor: 3,4-final}} & {\small Cor.~\ref{cor: 3,4-final}} & {\small Thm.~\ref{thm:strong-weak}} \\
\hline
\end{tabular}  
\end{small}
\end{center}

\medskip 

More precisely, although the $q$-Pochhammer symbols $f_{k,N}$ and $g_{k,N}$ in Eq.~\eqref{eq:f_kN} are holomorphic quantum modular (QM) functions for the group $\Gamma_N\subseteq\Gamma_1(N)$, their resurgent structure is not generally modular resurgent, and they are not reconstructed via median resummation ($\CS^{\rm med}_\theta\tilde{f}\neq f$). This shows that quantum modularity of a $q$-series is a less constraining property compared to its asymptotic series being modular resurgent (MRS).

To potentially address the non-effectiveness of the median resummation for the $q$-Pochhammer symbols, it might be fruitful to employ the framework of vector-valued quantum modular forms, for instance by considering a vector of components $f_{k,N}, f_{\underline{k},N}, g_{k,N}$, and~$g_{\underline{k},N}$. 
Similarly, one might be able to promote $f_{k,N}$ and $g_{k,N}$ to a vector-valued holomorphic quantum modular function for the full modular group $\mathsf{SL}_2(\IZ)$. 
Vector-valued quantum modular forms have appeared in the literature in connection with the resurgence of quantum invariants of knots and three-manifolds~\cite{Garoufalidis_Zagier_2023,garoufalidis_zagier_2023knots,wheeler-thesis,Cheng1,Cheng2,CCKPG-revised}, but not yet within the framework of modular resurgence~\cite{FR1maths}. We plan to address these ideas in a forthcoming paper. 

\medskip 

Additional analytic and arithmetic structures manifest when looking at the weighted sums $\frakf$ and $\frakg$ in Eq.~\eqref{eq:fg-intro}. In fact, under suitable assumptions on the Dirichlet character $\chi_N$ dictating the weights, both functions have modular resurgent asymptotic series, fit together into the modular resurgence paradigm (MRP), and provide new evidence for Conjectures~\ref{conj:quantum_modular1-intro} and~\ref{conj:quantum_modular2-intro}. 
We stress that the effectiveness of the median resummation for $\frakf$ and $\frakg$ only requires the character $\chi_N$ to be odd. This parity constraint is deeply connected to the divergence of the asymptotic series. Indeed, for even Dirichlet characters, the asymptotic expansions of $\frakf$ and $\frakg$ reduce to rational functions as the asymptotic series vanish identically.

\medskip 

Finally, the exact resurgent structures at weak and strong coupling of the asymptotic series of the logarithm of the spectral trace $\mathrm{Tr}(\rho_{m,n})$, associated with local $\IP^{m,n}$ via the quantization procedure of~\cite{GHM, CGM2}, follow directly from our previous results.
The generating functions of the Stokes constants $f_0^{n,m}$ and $f_\infty^{n,m}$, given in Eqs.~\eqref{eq:f0} and~\eqref{eq:f-inf}, are expressed as sums of $q$-Pochhammer symbols and---like $\frakf$ and $\frakg$---can be reconstructed from their asymptotic expansions via median resummation. Moreover, an exact, strong-weak (SW) resurgent symmetry is present, exchanging the perturbative and non-perturbative data at weak and strong coupling in a mathematically precise way. 
This symmetry was discovered for local $\IP^2$ (\emph{i.e.}, for $m+n=2$) in~\cite{Rella22, FR1phys}.

However, except in the cases of local $\IP^2,\, \IP^{1,2},$ and $\IP^{2,1}$ (\emph{i.e.}, for $m+n \ne 2,3$), the number-theoretic structure underlying the modular resurgence paradigm breaks down, resulting in a somewhat reduced version of the original, full-fledged strong-weak symmetry of local $\IP^2$. In particular, the Stokes constants at both weak and strong coupling no longer correspond to the coefficients of an $L$-function. 
A possible route to recover the modular resurgence paradigm is the combinatorial approach we proposed in Lemmas~\ref{ref:lem-combination-0} and~\ref{ref:lem-combination-inf}. Indeed, when $\chi_N$ is a primitive odd character, a suitable linear combination of the functions $f_0^{m,n}$ with $n+m+1=N$ (similarly, of the functions $f_\infty^{m,n}$ with the same choice of $m,n$) yields a pair of MRSs. Although a geometric interpretation of this combinatorial property of all local weighted projective planes is currently lacking---and perhaps less plausible---it would nevertheless be interesting to investigate whether the absence of a complete number-theoretic structure has a geometric origin. 

\medskip

In the case of local $\IP^2$, we observed in~\cite{FR1phys} that the arithmetic and modular properties of the $L$-functions and $q$-series derived from the weak- and strong-coupling Stokes constants are reminiscent of those of the periods of the mirror curve~\cite[Appendix~A]{Bousseau22}, as well as of the free energies of both the standard and Nekrasov--Shatashvili topological strings~\cite{CI, Bousseau20, ASYZ, Klemm:ODE-mirror-symmetry}. 
Moreover, the modularity of the free energies of the standard topological string for local $\IP^{2,1}$ has been discussed in~\cite{open-TR-Zhou} from the perspective of topological recursion and seems to reflect the arithmetic structure we observed for the generating functions $f_0^{2,1}$ and $f_\infty^{2,1}$.
What about the other local weighted projective planes? Do the Stokes constants still capture information about the arithmetic and modularity of the underlying geometry and the associated topological string theory? At present, we are not aware of any results for local weighted projective planes other than $\IP^2$ and $\IP^{2,1}$ that could help us address these questions.

Again, for local $\IP^2$, the holomorphic quantum modular forms that generate the Stokes constants present a remarkable duality under the action of the Fricke involution of $\Gamma_1(3)$, exchanging the Stokes constants at weak and strong coupling~\cite{FR1phys}.
In this paper, we have extended this duality to all local $\IP^{m,n}$. 
From the perspective of the TS/ST correspondence~\cite{GHM, CGM2}, the weak and strong coupling limits in $\hbar$ of the spectral traces of toric CY threefolds are related to the standard and Nekrasov--Shatashvili topological string free energies expanded at different special points in moduli space~\cite{ABK, GuM, Rella22}. 
We are therefore tempted to interpret this Fricke duality of resurgent generating functions as a physical duality of topological string amplitudes. Across these lines, there appear to be promising links with the results of~\cite{ASYZ,Alim-fricke}. 
The authors of~\cite{ASYZ,Alim-fricke} express the free energies of many non-compact CY manifolds, including local $\IP^2$ and other local del Pezzo geometries, in terms of the generators of the ring of quasi modular forms and study the action of the Fricke involution on the complex moduli, exchanging two distinguished expansion loci.
This enticing connection invites us to examine an explicit relationship between the emergent symmetries of the resurgent structures of the spectral traces in the TS/ST correspondence and the geometry and arithmetic of toric CY threefolds and their moduli spaces.

\medskip

To conclude, the authors of~\cite{KM} observe that the quantum operators associated to many toric del Pezzo CY threefolds via the prescription of~\cite{GHM, CGM2} can be described as appropriate perturbations of the three-term operators in Eq.~\eqref{eq:quantum-ops}, that is, they can be written in the form $\mO_{m,n}+\mV$ with $\mV$ positive and self-adjoint. For example, the canonical operators derived by Weyl quantization of the mirror curve to the resolved $\IC^3/\IZ_6$ orbifold are perturbations of $\mO_{4,1}$ and $\mO_{1,1}$~\cite{CGM2, CGuM}. 
It would be interesting to understand whether the resurgent and modular structures observed in the $\mO_{m,n}$ models extend to these more general geometries. In particular, one may ask whether (a generalization of) the strong–weak resurgent symmetry persists in the presence of such perturbations, or whether the complete modular resurgence paradigm restores. 
We remind that the paradigm fails for almost all local $\IP^{m,n}$, possibly due to the given geometries being singular. 
Indeed, with our choice of quantum operator $\rho_{m,n}$, we are neglecting contributions from the internal points of the toric diagrams.
In this regard, the insights gained from the local $\IP^{m,n}$ case can serve as a guiding framework, but further work is required. 

\section*{Acknowledgements}
We would like to thank 
Charles Doran, 
Alba Grassi,
Valentin Hernandez,
Kohei Iwaki,
Marcos Mari\~no, 
Boris Pioline, 
Emanuel Scheidegger,
and
Changgui Zhang
for insightful discussions.
This work has been partially supported by the ERC-SyG project ``Recursive and Exact New Quantum Theory'' (ReNewQuantum), which received funding from the European Research Council (ERC) within the European Union's Horizon 2020 research and innovation program under Grant No. 810573. 
V.F. has been supported by the Fondation Mathématique Jacques Hadamard at Laboratoire Mathématique d'Orsay. 
C.R. has been supported by the Huawei Young Talents Program at Institut des Hautes \'Etudes Scientifiques (IH\'ES).

\section*{Declarations}

\subsection*{Competing interests}

The authors have no relevant financial or non-financial interests to disclose.

\appendix

\section{Technical proofs}\label{app: proofs}

In this section, we provide the complete proofs of Propositions~\ref{prop:Borel-f_kN} and~\ref{prop:Borel-g_kN} and Lemma~\ref{lemma:summability-f-kN} on the resurgence and summability of the asymptotic expansions for $y\to 0$ with $\Im(y)>0$ of the $q$-Pochhammer symbols $f_{k,N}(y)$, $ g_{k,N}(y)$, $y \in \IH$, defined in Eq.~\eqref{eq:f_kN} after fixing $N\in\IZ_{\geq 2}$ and $k\in\IZ_N$.

\begin{proof}[Proof of Prop.~\ref{prop:Borel-f_kN}]
The Borel transform of the asymptotic series $\psi_k(y)$ in Eq.~\eqref{eq: tilde-psi} can be decomposed as
\be
\begin{aligned}
    \borel\left[\psi_k\right](\zeta)&=\sum_{n=1}^\infty (2\pi \ri)^{2n-1}\frac{B_{2n}}{(2n)!}\mathrm{Li}_{2-2n}(\zeta_N^k)\frac{\zeta^{2n-2}}{(2n-2)!}\\
    &=\left(\sum_{n=1}^\infty \frac{B_{2n}}{(2n)!}\zeta^{2n-2} \right) \diamond \left(\sum_{n=1}^\infty (2\pi \ri)^{2n-1}\frac{\mathrm{Li}_{2-2n}(\zeta_N^k)}{(2n-2)!} \zeta^{2n-2} \right)=:\tilde{g}(\zeta)\diamond \tilde{f}(\zeta)\,,
\end{aligned}
\ee
where $\zeta$ is the formal variable conjugate to $y$ and $\diamond$ denotes the Hadamard product~\cite{Hadamard}. The formal power series $\tilde{g}(\zeta)$, $\tilde{f}(\zeta)$ have finite radius of convergence at $\zeta=0$ and can be resummed explicitly into the functions\footnote{We impose that $g(0)=1/12$ to eliminate the removable singularity of $g(\zeta)$ at the origin.}
\begin{subequations}
    \begin{align}
        g(\zeta)&=-\frac{1}{\zeta^2}+\frac{1}{2 \zeta}\coth\left(\frac{\zeta}{2}\right) \, , \quad |\zeta|<2\pi \, , \\
    f(\zeta)&=-\pi \ri\left(\frac{1}{1-\zeta_N^{-k}\re^{-2\pi \ri \zeta}}+\frac{1}{1-\zeta_N^{-k}\re^{2\pi \ri \zeta}}\right) \, , \quad |\zeta|<\frac{k}{N} \, , 
    \end{align}
\end{subequations}
respectively.
After being analytically continued to the whole complex $\zeta$-plane, the function $g(\zeta)$ has simple poles along the imaginary axis at
\be
\mu_m = 2 \pi \ri m \, , \quad m \in \IZ_{\ne 0} \, ,
\ee
while the function $f(\zeta)$ has simple poles along the real axis at 
\be
\nu_\ell^{\pm}=\pm\frac{k}{N}+\ell \, , \quad \ell \in \IZ \, .
\ee
Let us now consider a circle $\gamma$ in the complex $s$-plane with center $s=0$ and radius $0 < r < 2 \pi$ and apply Hadamard's multiplication theorem~\cite{Hadamard, Hadamard2}. The Borel transform of $\psi_k(y)$ can be written as the integral
\be \label{eq: intBorelinfty}
\begin{aligned}
    \borel \left[\psi_k\right](\zeta) &=\frac{1}{2\pi \ri}\int_\gamma g(s)\, f\left(\frac{\zeta}{s}\right) \frac{ds}{s} \\
    &= \frac{1}{2}\int_\gamma \left(\frac{1}{s}-\frac{1}{2}\coth\left(\frac{s}{2}\right) \right)\left(\frac{1}{1-\zeta_N^{-k}\re^{-2\pi \ri \zeta/s}}+\frac{1}{1-\zeta_N^{-k}\re^{2\pi \ri \zeta/s}}\right) \frac{ds}{s^2} \, ,
\end{aligned}
\ee
for $|\zeta|<rk/N$. Here, the function $s \mapsto f(\zeta/s)$ is singular at $s = \zeta/\nu^{\pm}_\ell$, $\ell \in \IZ$, which sit inside the contour of integration $\gamma$ and accumulate at the origin, and regular for $|s| > r$. Meanwhile, the function $g(s)$ has simple poles at $s= \mu_m$ with residues
\be
\underset{s= 2 \pi \ri m}{\text{Res}} g(s) = \frac{1}{2 \pi \ri m} \, , \quad m \in \IZ_{\ne 0} \, . 
\ee
By Cauchy's residue theorem, the integral in Eq.~\eqref{eq: intBorelinfty} can be evaluated by summing the residues at the poles of the integrand which lie outside $\gamma$. More precisely, we have that
\be \label{eq: intBorel2infty-proof-A}
\begin{aligned}
    \borel \left[\psi_k\right](\zeta)&=-\sum_{m\in\IZ_{\neq 0}} \underset{s=2 \pi \ri m}{\text{Res}} \left[ g(s) f\left(\frac{\zeta}{s}\right)\frac{1}{s}\right] \\
    &=\sum_{m\in\IZ_{\neq 0}}\frac{1}{4 \pi \ri m^2}\left(\frac{1}{1-\zeta_N^{-k}\re^{-\frac{\zeta}{m}}}+\frac{1}{1-\zeta_N^{-k}\re^{\frac{\zeta}{m}}}\right)\\
    &=\frac{1}{2 \pi \ri} \sum_{m=1}^{\infty}\frac{1}{m^2}\left(\frac{1}{1-\zeta_N^{-k}\re^{-\frac{\zeta}{m}}}+\frac{1}{1-\zeta_N^{-k}\re^{\frac{\zeta}{m}}}\right) \, ,
\end{aligned}
\ee
which is manifestly Gevrey-1 and simple resurgent.
\end{proof}

\begin{proof}[Proof of Prop.~\ref{prop:Borel-g_kN}]
The Borel transform of the asymptotic series $\varphi_k(y)$ in Eq.~\eqref{eq: tilde-phi} can be decomposed as
\be
\begin{aligned}
    \borel\big[\varphi_k\big](\zeta)&=\sum_{n=1}^\infty (2\pi \ri)^{2n}\frac{B_{2n}B_{2n+1}\left(\tfrac{k}{N}\right)}{2n (2n+1)!}\frac{\zeta^{2n-1}}{(2n-1)!}\\
    &=\left(\sum_{n=1}^\infty \frac{B_{2n}}{(2n)!}\zeta^{2n-1} \right) \diamond \left(\sum_{n=1}^\infty (2\pi\ri)^{2n}\frac{B_{2n+1}\left(\tfrac{k}{N}\right)}{(2n+1)!}\zeta^{2n-1} \right)=:\tilde{g}(\zeta)\diamond \tilde{f}(\zeta)\,,
\end{aligned}
\ee
where $\zeta$ is the formal variable conjugate to $y$ and $\diamond$ denotes the Hadamard product~\cite{Hadamard}. The formal power series $\tilde{g}(\zeta)$, $\tilde{f}(\zeta)$ have finite radius of convergence at $\zeta=0$ and can be resummed explicitly into the functions\footnote{We impose that $g(0)=0$ and $f(0)=0$ to eliminate the removable singularities of $g(\zeta)$ and $f(\zeta)$ at the origin.}
\begin{subequations}
    \begin{align}
        g(\zeta)&=-\frac{1}{\zeta}+\frac{1}{2 }\coth\left(\frac{\zeta}{2}\right) \, , \quad |\zeta|<2\pi \, , \\
        f(\zeta)&=-\frac{B_1\left(\tfrac{k}{N}\right)}{\zeta} +\frac{\sin\left(\left(\tfrac{2k}{N}-1\right)\pi\zeta\right)}{2\zeta\sin(\pi\zeta)} \, , \quad |\zeta|<1 \, , 
    \end{align}
\end{subequations}
respectively. 
After being analytically continued to the whole complex $\zeta$-plane, the function $g(\zeta)$ has simple poles along the imaginary axis at
\be
\mu_m = 2 \pi \ri m \, , \quad m \in \IZ_{\ne 0} \, ,
\ee
while the function $f(\zeta)$ has simple poles along the real axis at 
\be
\nu_\ell=\ell \, , \quad  \ell\in\IZ_{\ne 0}\, .
\ee
Let us now consider a circle $\gamma$ in the complex $s$-plane with center $s=0$ and radius $0 < r < 2 \pi$ and apply Hadamard's multiplication theorem~\cite{Hadamard, Hadamard2}. The Borel transform of $\varphi_k(y)$ can be written as the integral
\be \label{eq: intBorel0}
\begin{aligned}
    \borel [\varphi_k](\zeta) &=\frac{1}{2\pi \ri}\int_\gamma f(s)\, g\left(\frac{\zeta}{s}\right) \frac{ds}{s} \\
    &= \frac{1}{2\pi\ri}\int_\gamma  \left(-\frac{B_1\big(\tfrac{k}{N}\big)}{s} +\frac{\sin\big(\big(\tfrac{2k}{N}-1\big)\pi s\big)}{2s\sin\big(\pi s\big)}\right)\left(-\frac{s}{\zeta}+\frac{1}{2}\coth\left(\frac{\zeta}{2s}\right)\right) \frac{ds}{s} \\
\end{aligned}
\ee
for $|\zeta|<r$. Here, the function $s \mapsto g(\zeta/s)$ is singular at $s = \zeta/\mu_m$, $m \in \IZ_{\ne 0}$, which sit inside the contour of integration $\gamma$ and accumulate at the origin, and regular for $|s| > r$. 
Meanwhile, the function $f(s)$ has simple poles at $s= \nu_\ell$ with residues
\be
\underset{s= \ell}{\text{Res}}\, f(s) = \frac{\sin\big(\tfrac{2 \pi k \ell}{N} \big)}{2\pi\ell} \, , \quad \ell\in\IZ_{\ne 0} \,.
\ee
By Cauchy's residue theorem, the integral in Eq.~\eqref{eq: intBorel0} can be evaluated by summing the residues at the poles of the integrand which lie outside $\gamma$. More precisely, we have that
\be \label{eq: intBorel20-proof-A}
\begin{aligned}
    \borel [\varphi_k](\zeta)&=-\sum_{\ell \in \IZ_{\ne 0}} \underset{s=\ell}{\text{Res}} \left[ f(s) g\left(\frac{\zeta}{s}\right)\frac{1}{s}\right] \\
    &={-}\sum_{\ell \in \IZ_{\ne 0}}\frac{\sin\big(\tfrac{2 \pi k \ell}{N}\big)}{2\pi\ell}\left(-\frac{1}{\zeta}+\frac{1}{2\ell}\coth\left(\frac{\zeta}{2\ell}\right)\right)\\
    &={-}\sum_{\ell =1}^{\infty}\frac{\sin\big(\tfrac{2 \pi k \ell}{N}\big)}{\pi\ell}\left(-\frac{1}{\zeta}+\frac{1}{2\ell}\coth\left(\frac{\zeta}{2\ell}\right)\right) \, ,
\end{aligned}
\ee
which is manifestly Gevrey-1 and simple resurgent. 
\end{proof}

\begin{proof}[Proof of Lemma~\ref{lemma:summability-f-kN}]
We can now compute the Laplace transform of $\borel \left[\psi_k\right](\zeta)$ in Eq.~\eqref{eq: intBorel2infty-proof-A} along the positive and negative real axes. When $\Re (y)>0$, we have that
\be
\begin{aligned}
    s_0(\psi_k)(y)&=\frac{1}{2\pi \ri}\int_0^\infty \re^{-\zeta/y} \left[\sum_{m=1}^{\infty}\frac{1}{m^2}\left(\frac{1}{1-\zeta_N^{-k}\re^{-\frac{\zeta}{m}}}+\frac{1}{1-\zeta_N^{-k}\re^{\frac{\zeta}{m}}}\right)\right] d\zeta\\
    &=\frac{1}{2\pi \ri}\int_0^\infty \left(\sum_{m=1}^{\infty}\frac{\re^{-mt/y}}{m} \right) \left(\frac{1 }{1-\zeta_N^{-k}\re^{-t}}+\frac{1}{1-\zeta_N^{-k}\re^{t}}\right) dt \, ,
\end{aligned}
\ee
where we have applied the change of variable $\zeta = m t$. Resumming the series in $m$, we find that
\be \label{eq: sum0-proof1}
\begin{aligned}
    s_0(\psi_k)(y)
    &=-\frac{1}{2\pi \ri}\int_0^\infty \log(1-\re^{-t/y})\left(\frac{1}{1-\zeta_N^{-k}\re^{-t}}+\frac{1}{1-\zeta_N^{-k}\re^{t}}\right) dt\\
    &=-\frac{1}{2\pi \ri}\int_{-\infty}^\infty \frac{\log(1-\re^{-t/y})}{1-\zeta_N^{-k}\re^{-t}} dt +\frac{1}{2\pi \ri}\int_0^{-\infty}\frac{\log(-\re^{t/y})}{1-\zeta_N^{-k}\re^{-t}} dt \, ,
\end{aligned}
\ee
where we have divided the integral into the sum of two contributions for simplicity.
Indeed, we observe that the second term in the RHS of Eq.~\eqref{eq: sum0-proof1} gives
\be \label{eq: sum0-proof1-p1}
\frac{1}{2\pi \ri}\int_0^{-\infty}\frac{\log(-\re^{t/y})}{1-\zeta_N^{-k}\re^{-t}} dt = -\frac{\mathrm{Li}_2(\zeta_N^k)}{2\pi \ri y}-\frac{1}{2}\log(1-\zeta_N^{k}) \, ,
\ee
while, after rewriting $\zeta_N^{-k}$ as $\zeta_N^{N-k}$, the first term gives
\begin{align}
    -\frac{1}{2\pi \ri}\int_{-\infty}^\infty \frac{\log(1-\re^{-t/y})}{1-\zeta_N^{N-k}\re^{-t}} dt &= \log\frac{(\zeta_N^k; \, \re^{2\pi \ri y})_\infty}{(\re^{-2\pi \ri (N-k)/(Ny)}; \, \re^{-2\pi \ri/y})_\infty} \, , \label{eq: sum0-proof1-p2a}
\end{align}
The last two formulae follow from~\cite[Thm.~A-29]{wheeler-thesis}.
Therefore, substituting Eqs.~\eqref{eq: sum0-proof1-p1} and~\eqref{eq: sum0-proof1-p2a} into Eq.~\eqref{eq: sum0-proof1}, we find that 
\begin{align}
s_0(\psi_k)(y)=& -\frac{\mathrm{Li}_2(\zeta_N^k)}{2\pi \ri y}-\frac{1}{2}\log(1-\zeta_N^{k}) +\log\frac{(\zeta_N^k; \, \re^{2\pi \ri y})_\infty}{(\re^{-2\pi \ri (N-k)/(Ny)}; \, \re^{-2\pi \ri/y})_\infty} \, . \label{eq: sum0-proof1-p3a}
\end{align}
Finally, using Eqs.~\eqref{eq:f_kN},~\eqref{eq: expansion-fkN}, and~\eqref{eq: sum0-proof1-p3a}, we obtain that
\be
\begin{aligned}
    s_0(\tilde{f}_{k,N})(y)&=\log\frac{(\zeta_N^k; \, \re^{2\pi \ri y})_\infty}{(\re^{-2\pi \ri (N-k)/(Ny)}; \, \re^{-2\pi \ri/y})_\infty} \\
    &=f_{k,N}(y)-g_{\underline{k},N}\big(-\tfrac{1}{Ny}\big) \, .
\end{aligned}
\ee

Similarly, for $\Re (y)<0$, the Borel--Laplace sum of $\psi_k(y)$ along the negative real axis is
\be
\begin{aligned}
    s_\pi(\psi_k)(y)&=\frac{1}{2\pi \ri}\int_0^{-\infty} \re^{-\zeta/y} \left[\sum_{m=1}^{\infty}\frac{1}{m^2}\left(\frac{1}{1-\zeta_N^{-k}\re^{-\frac{\zeta}{m}}}+\frac{1}{1-\zeta_N^{-k}\re^{\frac{\zeta}{m}}}\right) \right] d\zeta\\
    &=\frac{1}{2\pi \ri}\int_0^{-\infty} \left(\sum_{m=1}^{\infty}\frac{\re^{-mt/y}}{m} \right) \left(\frac{1 }{1-\zeta_N^{-k}\re^{-t}}+\frac{1}{1-\zeta_N^{-k}\re^{t}}\right) dt \, ,
\end{aligned}
\ee
where we have applied the change of variable $\zeta = m t$. Resumming the series in $m$, we find that
\be \label{eq: sumpi-proof1}
\begin{aligned}
    s_\pi(\psi_k)(y)
    &=-\frac{1}{2\pi \ri}\int_0^{-\infty} \log(1-\re^{-t/y})\left(\frac{1}{1-\zeta_N^{-k}\re^{-t}}+\frac{1}{1-\zeta_N^{-k}\re^{t}}\right) dt\\
    &=\frac{1}{2\pi \ri}\int_0^{\infty} \log(1-\re^{t/y})\left(\frac{1}{1-\zeta_N^{-k}\re^{t}}+\frac{1}{1-\zeta_N^{-k}\re^{-t}}\right) dt\\
    &=\frac{1}{2\pi \ri}\int_{-\infty}^\infty \frac{\log(1-\re^{t/y})}{1-\zeta_N^{-k}\re^{-t}} dt -\frac{1}{2\pi \ri}\int_0^{-\infty}\frac{\log(-\re^{-t/y})}{1-\zeta_N^{-k}\re^{-t}} dt \, ,
\end{aligned}
\ee
where we have again divided the integral into the sum of two contributions for simplicity. 
Let us now set $x=-1/y$, so that $\Re(x)>0$ and $\Im (x)>0$, then the second term in the RHS of Eq.~\eqref{eq: sumpi-proof1} becomes
\be \label{eq: sumpi-proof1-p1}
-\frac{1}{2\pi \ri}\int_0^{-\infty}\frac{\log(-\re^{t x})}{1-\zeta_N^{-k}\re^{-t}} dt = \frac{x} {2\pi \ri}\mathrm{Li}_2(\zeta_N^k) + \frac{1}{2}\log(1-\zeta_N^k) \, ,
\ee
while the first term gives
\begin{align}
    \frac{1}{2\pi \ri}\int_{-\infty}^\infty \frac{\log(1-\re^{-tx})}{1-\zeta_N^{N-k}\re^{-t}} dt &= -\log\frac{(\zeta_N^{k}; \, \re^{2\pi \ri /x})_\infty}{(\re^{-2\pi \ri (N-k)x/N}; \, \re^{-2\pi \ri x})_\infty} \, , \label{eq: sumpi-proof1-p2a}
\end{align}
as before. 
Thus, substituting Eqs.~\eqref{eq: sumpi-proof1-p1} and~\eqref{eq: sumpi-proof1-p2a}, after changing back $x$ to $-1/y$, into the formula for the Borel--Laplace sum in Eq.~\eqref{eq: sumpi-proof1}, we find that 
\be
s_\pi(\psi_k)(y)= -\frac{\mathrm{Li}_2(\zeta_N^k)}{2\pi \ri y}+\frac{1}{2}\log(1-\zeta_N^k) -\log\frac{(\zeta_N^k; \, \re^{-2\pi \ri y})_\infty}{(\re^{2\pi \ri (N-k)/(Ny)}; \, \re^{2\pi \ri/y})_\infty} \, . \label{eq: sumpi-proof1-p3a} 
\ee
Finally, putting together Eqs.~\eqref{eq:f_kN},~\eqref{eq: expansion-fkN}, and~\eqref{eq: sumpi-proof1-p3a} yields
\be \label{eq: sumpi-proof2}
\begin{aligned}
    s_\pi(\tilde{f}_{k,N})(y)&=\log(1-\zeta_N^k)-\log\frac{(\zeta_N^k; \, \re^{-2\pi \ri y})_\infty}{(\re^{2\pi \ri (N-k)/(Ny)}; \, \re^{2\pi \ri/y})_\infty} \\
    &=f_{k,N}(y)-g_{k,N} \big(-\tfrac{1}{Ny}\big) \, ,
\end{aligned}
\ee
where we have applied the identity
\be \label{eq: qdilog-rel}
(x; \, q^{-1})_\infty = (xq; \, q)_\infty^{-1} \, .
\ee
and used that
\be
\log(\zeta_N^k \, \re^{2\pi \ri y}; \, \re^{2\pi \ri y})_\infty = \log(\zeta_N^k; \, \re^{2\pi \ri y})_\infty - \log(1-\zeta_N^k) \, .
\ee
\end{proof}

\addcontentsline{toc}{section}{References}
\bibliographystyle{JHEP}
\linespread{0.4}
\bibliography{localP2-biblio-LMP}

\end{document}